\documentclass[review, 3p, times]{elsarticle}

\journal{Physica D: Nonlinear Phenomena}

\biboptions{sort&compress}

\makeatletter
\AtBeginDocument{\def\NAT@space{}\def\NAT@spacechar{}}
\makeatother


\usepackage{amsmath, amssymb, amsthm}
\usepackage{bm, mathtools}
\usepackage{graphicx}
\usepackage{booktabs}
\usepackage{tabularx}
\usepackage{algorithm, algorithmic}          
\usepackage{xcolor}
\usepackage{microtype}
\usepackage{enumitem}
\usepackage{xspace}
\usepackage{caption}
\usepackage{verbatim}  

\makeatletter
\renewcommand\section{\@startsection {section}{1}{\z@}%
           {-18\p@ \@plus -6\p@ \@minus -3\p@}%
           {9\p@ \@plus 6\p@ \@minus 3\p@}%
           {\normalsize\bfseries\boldmath}}
\renewcommand\subsection{\@startsection{subsection}{2}{\z@}%
           {-12\p@ \@plus -6\p@ \@minus -3\p@}%
           {3\p@ \@plus 6\p@ \@minus 3\p@}%
           {\normalfont\normalsize\itshape}}
\renewcommand\subsubsection{\@startsection{subsubsection}{3}{\z@}%
           {-12\p@ \@plus -6\p@ \@minus -3\p@}%
           {\p@}%
           {\normalfont\normalsize\itshape}}
\usepackage[colorlinks=true, linkcolor=blue, citecolor=blue, urlcolor=blue, bookmarksnumbered=true]{hyperref}
\providecommand{\cref}[1]{Section~\ref{#1}}        
\providecommand{\Cref}[1]{Section~\ref{#1}}
\providecommand{\figref}[1]{Fig.~\ref{#1}}         
\providecommand{\Figref}[1]{Fig.~\ref{#1}}         
\providecommand{\tabref}[1]{Table~\ref{#1}}
\providecommand{\Tabref}[1]{Table~\ref{#1}}
\providecommand{\propref}[1]{Proposition~\ref{#1}}
\providecommand{\appref}[1]{\ref{#1}}
\providecommand{\Appref}[1]{Appendix~\ref{#1}}   

\usepackage{flafter}

\AtBeginDocument{\hypersetup{%
  pdftitle={Energy-Stable Parametric Operator Inference across Regime Transitions: Application to Unconfined Buoyant Plumes in Stratified Ambients},%
  pdfauthor={Youngjun Kwon, Abdollah Shafieezadeh}}}

\newcommand{\CROI}{CR-OpInf\xspace}
\newcommand{\Ab}{\texorpdfstring{\hat{\bm{A}}}{A}}

\newtheorem{proposition}{Proposition}

\theoremstyle{remark}
\newtheorem{remark}{Remark}
\theoremstyle{plain}

\begin{document}

\begin{frontmatter}

\title{Energy-Stable Parametric Operator Inference across Regime Transitions: Application to Unconfined Buoyant Plumes in Stratified Ambients}

\author[osu]{Youngjun Kwon}
\ead{kwon.683@osu.edu}
\author[osu]{Abdollah Shafieezadeh\corref{cor1}}
\ead{shafieezadeh.1@osu.edu}
\cortext[cor1]{Corresponding author.}
\affiliation[osu]{organization={Department of Civil, Environmental and
  Geodetic Engineering, The Ohio State University},
  city={Columbus}, state={OH}, postcode={43210}, country={USA}}

\pdfbookmark[1]{Abstract}{pdfabstract}%
\begin{abstract}

\noindent Buoyant plumes in stratified ambients arise in wildfire smoke transport, ventilation, and pollutant dispersion, where many-query prediction requires models that stay accurate and stable as the dynamics transition from free rise to internal-wave trapping. Operator Inference (OpInf) learns polynomial reduced models from snapshots, and the Boussinesq plume depends affinely on the squared buoyancy frequency $N^2$. However, this compatibility masks two stability gaps: a single affine operator family can distort the per-parameter dynamics into unstable spectra, and for quadratic OpInf models, finite-time blowup is governed by a cubic energy injection channel from the inferred quadratic operator. To address these gaps, we develop cyclic-skew regime-bridged OpInf (CR-OpInf). We derive an affine-inadequacy diagnostic that links the departure from a shared affine family to destabilizing spectral drift before integration. In place of the shared affine fit, we interpolate per-parameter operators with a Gaussian kernel. We also extend the cyclic-skew quadratic constraint to the buoyancy-source-forced parametric setting, closing the cubic energy channel and isolating this closure as the boundedness mechanism. We further introduce a regime-bridging basis split that resolves free rise and trapped intrusion structures while preserving the cyclic-skew structure under convex interpolation. Applied to an unconfined Boussinesq plume and a viscous Burgers problem spanning a smooth-to-shock transition, CR-OpInf remains bounded for five training horizons, avoids the divergence of affine and unconstrained variants, tracks the plume rise, and recovers stratification-dependent rise scaling at moderate horizons. To our knowledge, this is the first parametric OpInf study of an unconfined stratified plume. 

\end{abstract}

\pdfbookmark[1]{Keywords}{pdfkeywords}%

\end{frontmatter}




\section{Introduction}
\label{sec:intro}

\subsection{Buoyant plumes in stratified ambients}
\label{sec:intro:motivation}

A buoyant plume rising into a stably stratified ambient serves as a canonical model for a wide range of geophysical and engineering problems. These include wildfire smoke transport \citep{trentmann2006modeling}, pollutant dispersion from industrial stacks \citep{briggs1969plume, morton1956turbulent}, the natural ventilation of buildings \citep{linden1999fluid}, and buoyant discharges and stratified-fluid transport in geophysical and environmental flows \citep{hewitt2020subglacial, more2023stratified}. The strength of the ambient stratification controls a qualitative regime transition in these flows. That strength is set by the Brunt--V\"ais\"al\"a frequency, which measures the natural oscillation frequency of a vertically displaced fluid parcel. When the stratification is weak, the plume rises freely until turbulent entrainment dilutes its buoyancy and its rise is arrested \citep{morton1956turbulent}. When it is strong, the plume overshoots a level of neutral buoyancy, falls back, and spreads laterally as an intrusion \citep{bloomfield1998intrusion} that oscillates as an internal gravity wave near the Brunt--V\"ais\"al\"a frequency. Because the transition between the two regimes is gradual, the rise height, the intrusion depth, and the oscillation characteristics all shift substantially as the stratification carries the flow from one regime to the other \citep{briggs1969plume}.

Prediction across the full parameter range underlies real-time decision support, uncertainty quantification, and inverse problems in environmental hazard management, the many-query tasks for which reduced order surrogates are constructed \citep{peherstorfer2018combining}. Each such task requires repeated evaluation of the flow at different stratifications, and direct simulation of the governing partial differential equations (PDEs) is prohibitively expensive at the scale these tasks demand \citep{benner2015survey}. Reduced-order models (ROMs) constructed from snapshot data offer a natural remedy, since an offline training stage produces a low-dimensional surrogate whose online evaluation cost is orders of magnitude lower than that of the corresponding full-order model (FOM) \citep{peherstorfer2016data}. A qualitative change in the dynamics across a regime transition challenges the prevailing assumption that a single low-dimensional structure is valid over the entire parameter range. A parametric ROM must therefore do more than interpolate solutions across parameters. It must resolve distinct, regime-dependent flow structures while maintaining stable reduced dynamics as the dominant physics changes.

\subsection{Reduced-order modeling and Operator Inference}
\label{sec:intro:opinf}

Reduced-order modeling of nonlinear PDEs is supported by an extensive literature, including surveys \citep{benner2015survey, kramer2024learning} and standard references on projection-based and reduced-basis methods \citep{antoulas2005approximation, hesthaven2016certified}. Classical approaches project the governing equations onto a low-dimensional basis obtained from proper orthogonal decomposition (POD), including its balanced variants \citep{rowley2005model, willcox2002balanced}. The resulting reduced operators are mathematically faithful but require access to the discretized equations. Multifidelity strategies that combine such surrogates with the FOM have been used effectively in parameter study tasks \citep{peherstorfer2018combining}. Operator Inference (OpInf) \citep{peherstorfer2016data} infers polynomial form reduced operators non-intrusively from snapshot data through a least-squares regression in coordinates spanned by a POD basis, regularized in standard practice \citep{swischuk2020learning, mcquarrie2021combustion}. Other data-driven surrogates differ in what they assume about the dynamics. Dynamic mode decomposition \citep{schmid2010dynamic} learns a linear evolution operator, which is well suited to flows dominated by linear dynamics but does not by itself capture the quadratic convective nonlinearity of the Boussinesq equations that govern a buoyant plume. Sparse nonlinear identification \citep{brunton2016discovering} can discover nonlinear terms from data, a capability most valuable when the governing form is unknown. Neural operator surrogates such as deep operator networks \citep{lu2021learning} and Fourier neural operators \citep{li2021fourier} learn solution maps directly without a reduced basis and are correspondingly expressive, but they expose no explicit reduced operator on which spectral or structural constraints can act. OpInf suits a quadratically nonlinear system such as the Boussinesq plume, because its operators are interpretable linear algebraic objects whose spectrum and structure can be inspected and constrained, and its polynomial reduced model coincides with the exactly quadratic state dependence of the governing equations.

Since its introduction, OpInf has been extended along several directions that together broaden its applicability and strengthen its guarantees. On the representation side, lifting transformations \citep{kramer2019liftlearn} recast a broader class of nonlinear PDEs in polynomial form, and Lift \& Learn \citep{qian2020lift} thereby extends OpInf to non-polynomial advection--reaction systems, while the quadratic-manifold decoder of \citep{geelen2023operator} replaces the linear POD decoder to capture the leading curvature of the snapshot manifold. On the parametric side, the extension of \citep{mcquarrie2023nonintrusive} casts the reduced operators as an affine function of the varied parameter and fits a single joint least-squares regression across all training values. The same affine construction was developed independently for a parameterized rotating shallow water flow \citep{yildiz2021data} and later combined with additional structural constraints, as in a tensor-parametric Hamiltonian variant for canonical systems \citep{vijaywargiya2025tensor}. In the non-parametric setting, physical state constraints have also been embedded in the inference \citep{kim2025physically}. An element-wise scheme that sidesteps the affine form by inferring and interpolating per-parameter models has also been proposed \citep{kang2025parametric}. Because unconstrained OpInf does not guarantee bounded long-time trajectories, recent structure-preserving variants have introduced stability constraints on the linear operator, the quadratic operator, or both. The guaranteed-stable models of \citep{goyal2024guaranteed, pontesduff2024stability} parameterize the reduced linear operator in a dissipative-plus-skew form whose spectrum lies in the open left half-plane by construction, making the learned linear part Lyapunov-stable. The cyclic skew-symmetric constraint on the reduced quadratic operator, classical in energy-preserving Galerkin models \citep{schlegel2015long} and brought into OpInf by \citep{sharma2022hamiltonian, gruber2023canonical, gkimisis2025representation, koike2024energy}, preserves the energy neutrality of the convective nonlinearity, and the two constraints are usually imposed together. Preserving physical structure more broadly, such as Lagrangian structure, has been shown to extend the horizon over which a reduced model stays predictive \citep{sharma2024lagrangian}. Operator inference has also been applied to buoyancy-driven convection in a closed Rayleigh--B\'enard cell at a fixed Rayleigh number \citep{rocha2023parametric}.

\subsection{The regime-transition setting}
\label{sec:intro:gap}

The lifting, parametric, and structure-preserving extensions of OpInf have substantially broadened its applicability and strengthened its structural guarantees, yet the regime-transition setting raises distinct challenges for the parametric extension that standard single-regime applications do not encounter. It places requirements at three stages of the standard pipeline: the reduced basis, the operator inference step, and the structural prior. At the level of the reduced basis, a single global POD basis trained on the union of all training snapshots is a compromise between the two regimes' dominant spatial structures, since the modes that resolve the weakly stratified free rise phase differ from those that resolve the strongly stratified trapped oscillation. Local-basis constructions \citep{amsallem2012nonlinear} and bifurcation-aware localized reduced-order modeling \citep{hess2019localized} address related partitioning of the state space, and partitioning the training parameter domain with local bases is established \citep{eftang2010hp}, but none of these has been combined with the structure-preserving inference constraints developed for OpInf.

At the operator inference step, the joint affine least-squares fit of \citep{mcquarrie2023nonintrusive} constrains the inferred operators at every training parameter to a shared affine form, and the joint fit therefore balances the per-parameter optima against one another at every training point. This balancing has a limited effect when the dynamics vary slowly with the parameter, but across a qualitative regime change it can shift the inferred linear spectrum toward instability and degrade the reduced trajectories. The related observation that the choice of training parameters governs the stability and accuracy of parametric OpInf appears in a Bayesian active learning study \citep{mcquarrie2026active}. A recovery analysis shows that snapshots generated by re-projection let OpInf recover the intrusive Galerkin operators pre-asymptotically \citep{peherstorfer2020sampling}, and least-squares Petrov--Galerkin projection \citep{carlberg2011efficient} offers an intrusive, residual-minimizing alternative to Galerkin projection. Standard mitigations such as stronger regularization, a larger reduced dimension, or a denser training grid can improve the observed behavior but act on the fit while leaving its affine-compromise origin in place, and at present no a priori indicator signals when the affine form becomes inadequate.

At the level of the structural prior, which of the two constraints carries long-time stability on a forced dissipative flow remains unclear. The dissipative linear spectrum and the energy-neutral nonlinearity are the candidate mechanisms, and the two have not been compared as separable factors, because prior studies impose one constraint alone \citep{koike2024energy} or both together \citep{kaptanoglu2021promoting, goyal2024guaranteed}. The cyclic-skew quadratic constraint, classical in energy-preserving Galerkin models \citep{schlegel2015long}, entered OpInf through the canonical Hamiltonian setting \citep{sharma2022hamiltonian, gruber2023canonical} and the energy-preserving operator representations of \citep{gkimisis2025representation}. Although it has since been imposed in OpInf on dissipative viscous Burgers and Kuramoto--Sivashinsky systems \citep{koike2024energy}, its validity as a structural prior on a forced dissipative flow is not established. The Lyapunov-stable linear operator constraint comes from the guaranteed-stable quadratic models of \citep{goyal2024guaranteed, pontesduff2024stability}. The buoyant plume of interest adds finite viscosity together with an external buoyancy source that breaks the conservative structure, and the resulting system admits no known port-Hamiltonian formulation with an explicit dissipation operator. Because the plume offers no such form, the dissipation-aware port-Hamiltonian extension of \citep{geng2025porthamiltonian} does not directly apply. The three requirements are coupled, in that a basis-side change alone still leaves the multi-parameter compromise, an inference-side change alone still inherits the single-basis representation, and a structural prior carried over from the canonical setting calls for empirical validation on a dissipative, forced flow.

\subsection{Contributions}
\label{sec:intro:contributions}

\noindent This paper develops cyclic-skew regime-bridged Operator Inference (CR-OpInf), a parametric reduced-order modeling framework for forced dissipative flows whose dynamics span a regime transition. The framework addresses three coupled limitations of parametric OpInf: affine compromise in the inferred operators, nonphysical energy injection by the learned quadratic term, and loss of accuracy from using a single global basis across distinct flow regimes. The key contributions are as follows:

\begin{itemize}
    \item We introduce an affine-inadequacy diagnostic for standard affine parametric OpInf. The diagnostic measures how far the per-parameter operator fits depart from a shared affine family and, motivated by the Bauer--Fike eigenvalue perturbation bound, ties this mismatch to the spectral drift of the joint fit, whose magnitude and destabilizing direction we establish empirically before time integration. This analysis explains why the joint affine fit can move otherwise accurate per-parameter operators toward unstable linear spectra and motivates a non-affine treatment of parameter dependence.

    \item We develop a local operator interpolation strategy for parametric OpInf. Instead of enforcing one joint affine regression across all training parameters, we infer reduced operators independently at each training parameter and interpolate them through a Gaussian kernel at evaluation. The blend is a single convex combination, which preserves the structural constraint below at every evaluated parameter. This construction preserves the local regression optima, produces smooth parameter dependence, and avoids the multi-parameter compromise that destabilizes the affine fit.

    \item We extend the cyclic-skew quadratic constraint to buoyancy-source-forced dissipative reduced models. In this setting, the constraint no longer represents an exact conserved invariant of the full-order system. We show instead that it closes the cubic energy injection channel of quadratic OpInf models and thereby rules out finite-time blowup from the learned nonlinear term. By comparing cyclic-skew and Lyapunov-stable variants, we show that the Lyapunov-stable linear parameterization is not necessary for boundedness in the tested cases and can reduce accuracy for the plume.

    \item We introduce a regime-bridging basis split for the transition from free rise to trapping. The split resolves regime-dependent plume structures in separate reduced spaces, while convex operator interpolation preserves the cyclic-skew constraint at every evaluated parameter. The resulting CR-OpInf framework combines local inference, nonlinear energy-channel closure, and regime-aware representation without relying on a single affine operator family or a single global POD basis.

    \item We apply CR-OpInf to two parametric PDE settings. The primary study is a stratified Boussinesq plume, where we evaluate leave-one-out reconstruction accuracy, plume rise quantities of interest, stratification-dependent rise scaling, long-horizon boundedness, robustness to rank and training set size, and cross-regime parameter extrapolation. The second study uses a viscosity-parametric Burgers equation to test whether the stability behavior observed for the plume transfers to a different convective PDE with an energy-preserving quadratic nonlinearity. Across these studies, we compare CR-OpInf with affine parametric OpInf and with structural variants that isolate the effects of the cyclic-skew quadratic constraint and the Lyapunov-stable linear parameterization.
\end{itemize}

The remainder of the paper is organized as follows. \Cref{sec:background} reviews the elements of OpInf that the proposed framework builds on, including POD, polynomial form operator regression, the affine parametric extension, and two structure-preserving variants. \Cref{sec:methods} develops the proposed framework, covering the spectral diagnosis of the multi-parameter compromise, local operator interpolation, the cyclic-skew prior for non-canonical dissipative dynamics, the regime-bridging basis split, and the CR-OpInf construction. \Cref{sec:problem} describes the two parametric PDE problems used for evaluation. \Cref{sec:plume} presents the buoyant plume study, \cref{sec:burgers} examines the cross-problem transfer to the viscous Burgers problem, and \cref{sec:extrap} studies three parameter space extrapolation scenarios. \Cref{sec:discussion} synthesizes the empirical findings, and \cref{sec:conclusion} concludes. Three appendices provide supporting material. \Appref{sec:app_fom} covers the full-order solver, its validation, and the computational cost. \Appref{sec:app_theory} develops the spectral diagnosis and the energy boundedness theory. \Appref{sec:app_certified} presents the stability-constrained interpolation variant, the stability-accuracy tradeoff analysis, and the per-variant accuracy breakdown.


\section{Background}
\label{sec:background}

This section reviews the elements of OpInf that this paper builds on. We cover the POD that produces the reduced basis (\cref{sec:background:pod}), the polynomial form reduced model that defines the inference target (\cref{sec:background:opinf}), the data-driven least-squares regression that fits the reduced operators (\cref{sec:background:fit}), the affine parametric extension that is the baseline for this paper (\cref{sec:background:popinf}), and the two structure-preserving variants on which we draw (\cref{sec:background:structure}).

\subsection{Proper orthogonal decomposition}
\label{sec:background:pod}

The POD supplies the low-dimensional coordinate system in which every OpInf fit is expressed. Let $\bm{q}(\bm{x}, t; \mu) \in \mathbb{R}^n$ denote a discretized state of a parametric PDE at spatial points $\bm{x}$, time $t$, and parameter $\mu$. Here $n$ is the total number of spatial degrees of freedom. The discretized state collects the values of every solution variable on the computational mesh into a single vector. For a single value of the parameter $\mu_i$, let $K$ denote the number of temporal snapshots. Snapshots of the state at times $t_1, \ldots, t_K$ are collected into a snapshot matrix $\bm{Q}^{(i)} \in \mathbb{R}^{n \times K}$ whose $k$-th column is $\bm{q}(\bm{x}, t_k; \mu_i)$. Snapshots from $M$ training parameter values $\{\mu_1, \ldots, \mu_M\}$ are concatenated into a joint snapshot matrix $\bm{Q} = [\bm{Q}^{(1)} \mid \cdots \mid \bm{Q}^{(M)}] \in \mathbb{R}^{n \times MK}$, then centered by subtracting from each column the snapshot mean $\bar{\bm{q}} \in \mathbb{R}^n$ computed over all training snapshots.

When the state collects several physical variables of differing numerical magnitude, the POD energy would otherwise be dominated by whichever variable carries the largest scale. We therefore assume throughout that the solution variables are nondimensionalized on a common scale before the snapshots are assembled, so that no single variable dominates the decomposition on numerical-magnitude grounds alone. Any additional variable-wise scaling applied prior to the POD is stated where it is used.

The POD of the centered snapshot matrix is its truncated singular value decomposition (SVD). The leading $r$ left singular vectors are collected into a reduced basis $\bm{\Phi} \in \mathbb{R}^{n \times r}$ whose columns are orthonormal and span the $r$-dimensional subspace that minimizes the Frobenius projection error among all rank-$r$ subspaces. The truncation rank $r$ is chosen so that the residual energy beyond the leading $r$ singular values is below a prescribed tolerance. The reduced state
\begin{equation}
\hat{\bm{q}}(t; \mu) \;=\; \bm{\Phi}^\top \bigl(\bm{q}(t; \mu) - \bar{\bm{q}}\bigr)
\;\in\; \mathbb{R}^r
\label{eq:pod_state}
\end{equation}
is the coordinate vector of $\bm{q}(t; \mu) - \bar{\bm{q}}$ in the POD basis, and the full-order state is approximately recovered by the inverse projection $\bm{q}(t; \mu) \approx \bar{\bm{q}} + \bm{\Phi}\, \hat{\bm{q}}(t; \mu)$. For large-scale problems the truncated SVD is computed by a randomized algorithm \citep{halko2011finding}. This algorithm introduces controlled randomness into $\bm{\Phi}$ through the random projection it uses to approximate the column space of the snapshot matrix before the deterministic factorization.

\subsection{The polynomial form reduced model}
\label{sec:background:opinf}

The reduced operators that OpInf infers act on the POD coordinates of Eq.~\eqref{eq:pod_state} and obey a polynomial form that matches the structure of the underlying PDE. Specifying this reduced model form makes the inference problem finite-dimensional, since the unknowns are the polynomial coefficients of a low-order system. For a PDE whose right-hand side is polynomial in the state, the projection of the dynamics onto the reduced basis $\bm{\Phi}$ produces a reduced dynamical system whose right-hand side is also polynomial in the reduced state. OpInf \citep{peherstorfer2016data} bypasses the intrusive projection and infers the reduced operators directly from snapshot data. This polynomial setting is the one this paper targets. For a PDE with constant, linear, and quadratic terms the reduced dynamics take the form
\begin{equation}
\dot{\hat{\bm{q}}}(t; \mu) \;=\; \hat{\bm{c}}(\mu)
\;+\; \hat{\bm{A}}(\mu)\, \hat{\bm{q}}(t; \mu)
\;+\; \hat{\bm{H}}(\mu)\, \bigl(\hat{\bm{q}}(t; \mu) \otimes \hat{\bm{q}}(t; \mu)\bigr),
\label{eq:opinf_ansatz}
\end{equation}
where $\hat{\bm{c}} \in \mathbb{R}^r$ is the reduced constant forcing, $\hat{\bm{A}} \in \mathbb{R}^{r \times r}$ is the reduced linear operator, and $\hat{\bm{H}} \in \mathbb{R}^{r \times r^2}$ is the reduced quadratic operator. The symbol $\otimes$ denotes the Kronecker product, so that $\hat{\bm{q}} \otimes \hat{\bm{q}} \in \mathbb{R}^{r^2}$ is the vector of all pairwise products of the entries of $\hat{\bm{q}}$. The dot denotes a derivative with respect to time. The polynomial reduced model of Eq.~\eqref{eq:opinf_ansatz} is a natural choice for any PDE whose nonlinearity is quadratic in the state, and a state-lifting transformation \citep{qian2020lift} extends the applicability to a broader class of non-polynomial systems.

\subsection{Data-driven inference by Tikhonov-regularized least squares}
\label{sec:background:fit}

The polynomial coefficients of the reduced model are determined by a single least-squares regression against projected snapshots and their numerically estimated time derivatives, without ever projecting the discrete PDE. Given the polynomial reduced form of Eq.~\eqref{eq:opinf_ansatz}, the reduced operators $(\hat{\bm{c}}, \hat{\bm{A}}, \hat{\bm{H}})$ are inferred from the projected snapshot data and a numerical estimate of the projected time derivative. At a single training parameter $\mu_i$, let $\bm{X}^{(i)} = \bm{\Phi}^\top (\bm{Q}^{(i)} - \bar{\bm{q}} \bm{1}_K^\top) \in \mathbb{R}^{r \times K}$ denote the projected centered snapshots, where $\bm{1}_K$ is the vector of ones of length $K$. Let $\bm{Y}^{(i)} \in \mathbb{R}^{r \times K}$ denote the projected time-derivative matrix obtained by finite differencing the projected snapshots in time. With the polynomial terms on the right-hand side of the reduced model stacked into a column-wise data matrix
\begin{equation}
\bm{D}^{(i)} \;=\;
\bigl[\,
\bm{1}_K\;\big|\;
\bigl(\bm{X}^{(i)}\bigr)^\top\;\big|\;
\bigl(\bm{X}^{(i)} \otimes_{\rm col} \bm{X}^{(i)}\bigr)^\top
\,\bigr]
\;\in\; \mathbb{R}^{K \times (1 + r + r^2)},
\label{eq:design_matrix}
\end{equation}
where $\otimes_{\rm col}$ denotes the Khatri--Rao (column-wise Kronecker) product, the inference reduces to a linear least-squares regression
\begin{equation}
\min_{\hat{\bm{c}},\, \hat{\bm{A}},\, \hat{\bm{H}}}
\;\Bigl\|
\bigl[\hat{\bm{c}}\;\big|\;\hat{\bm{A}}\;\big|\;\hat{\bm{H}}\bigr]
\,\bm{D}^{(i)\,\top}
\;-\; \bm{Y}^{(i)}
\Bigr\|_F^2
\;+\;\beta\,
\Bigl\|
\bigl[\hat{\bm{c}}\;\big|\;\hat{\bm{A}}\;\big|\;\hat{\bm{H}}\bigr]
\Bigr\|_F^2,
\label{eq:opinf_regression}
\end{equation}
where $\|\cdot\|_F$ is the Frobenius norm and $\beta > 0$ is the Tikhonov weight that controls the conditioning of the regression \citep{swischuk2020learning, mcquarrie2021combustion}. The unregularized normal equations admit a closed-form solution when the data matrix, in its compact deduplicated form, has full column rank \citep{peherstorfer2016data}. The regularized version is solved equivalently by the augmented normal equations or by a singular value decomposition of the data matrix \citep{mcquarrie2021combustion}. The choice of $\beta$ is itself an inference step, made through an L-curve criterion or a held-out validation set \citep{mcquarrie2021combustion}. A variant assessed for robustness in \cref{sec:burgers:alpha} scales the regularization for the $k$-th reduced mode by the $k$-th singular value of the snapshot matrix raised to a power chosen by the user, analogous to the mode-dependent weights that arise in Bayesian OpInf \citep{guo2022bayesian}. Accurate estimation of the projected time derivative $\bm{Y}^{(i)}$ is essential, because low-order finite differences introduce a bias that propagates into the inferred operators. A higher-order scheme is therefore recommended when the snapshot spacing is fine enough, and higher-order central differences are standard options in OpInf practice \citep{kramer2024learning}. All results reported here estimate the projected time derivative $\bm{Y}^{(i)}$ with a sixth-order finite difference.

The Kronecker representation of the quadratic operator admits a compact alternative. We write the model in the full Kronecker form $\hat{\bm{q}} \otimes \hat{\bm{q}} \in \mathbb{R}^{r^2}$ for the analysis, so that $\hat{\bm{H}} \in \mathbb{R}^{r \times r^2}$ carries $r^3$ entries. The regression itself uses the common compact form that deduplicates the equal pairs $\hat{q}_i \hat{q}_j = \hat{q}_j \hat{q}_i$ and works with $\hat{\bm{H}}$ of width $r(r+1)/2$ in place of $r^2$, which removes the duplicated columns from the data matrix without changing the model. The cyclic-skew structural prior of \cref{sec:background:structure} removes a further third of the independent parameters under either convention.

\subsection{Affine parametric extension}
\label{sec:background:popinf}

The standard way to incorporate a continuous parameter into OpInf is the affine parametric extension, in which the reduced operators of Eq.~\eqref{eq:opinf_ansatz} are allowed to depend on the parameter $\mu$. The simplest parameterization \citep{mcquarrie2023nonintrusive} is an affine form in a known scalar function $g(\mu)$ of the parameter,
\begin{equation}
\hat{\bm{A}}(\mu) \;=\; \hat{\bm{A}}_0
\;+\; g(\mu)\, \hat{\bm{A}}_1,
\label{eq:affine_ansatz}
\end{equation}
where $\hat{\bm{A}}_0, \hat{\bm{A}}_1 \in \mathbb{R}^{r \times r}$ are parameter-independent operators inferred jointly from all training parameters, and the form of $g(\mu)$ is chosen to reflect the parameter dependence of the underlying PDE. The constant and quadratic operators $\hat{\bm{c}}$ and $\hat{\bm{H}}$ are commonly taken to be independent of the parameter. In a centered coordinate system, affine parameter dependence in the full-order linear operator can also induce affine dependence in the reduced constant term through the action of that operator on $\bar{\bm{q}}$. The simplified baseline used here follows the common setting in which the dominant parameter dependence is assigned to $\hat{\bm{A}}(\mu)$. Variants may also parameterize $\hat{\bm{c}}(\mu)$ affinely when the centering choice makes this term non-negligible. Substituting Eq.~\eqref{eq:affine_ansatz} into the reduced model above produces a single joint least-squares regression in the unknowns $(\hat{\bm{c}}, \hat{\bm{A}}_0, \hat{\bm{A}}_1, \hat{\bm{H}})$ that fits the reduced dynamics across all $M$ training parameter values simultaneously through a stacked data matrix analogous to Eq.~\eqref{eq:design_matrix}. The joint fit yields a single reduced model that can be evaluated at any parameter value in the training range and exposes the parameter dependence in closed form. A similar construction was introduced independently for shallow water flow \citep{yildiz2021data}. The construction defined by Eq.~\eqref{eq:affine_ansatz} is \emph{affine parametric OpInf}, the baseline against which the proposed framework is compared. In the comparisons that follow, its single multi-parameter regression is referred to as the affine joint fit.

\subsection{Structure-preserving variants}
\label{sec:background:structure}

The unconstrained regression of Eq.~\eqref{eq:opinf_regression} carries no built-in guarantee that the inferred reduced operators produce long-time bounded trajectories, and for canonical Hamiltonian and conservative systems a faithful reduced model should preserve the energetic structure of the underlying PDE. The absence of a stability guarantee on the inferred operators is more generally a recurring concern for OpInf \citep{kramer2024learning, sawant2023pireg}, and two complementary developments target this issue.

The first development \citep{goyal2024guaranteed, pontesduff2024stability} constrains the reduced linear operator through a dissipative-plus-skew parameterization. We state a representative instance
\begin{equation}
\hat{\bm{A}} \;=\; -\bm{P}\bm{P}^\top - \varepsilon \bm{I}
\;+\; \operatorname{skew}(\bm{S}),
\label{eq:lyap_A}
\end{equation}
with $\bm{P}, \bm{S} \in \mathbb{R}^{r \times r}$ free matrices, $\varepsilon > 0$ a small slack parameter, $\bm{I}$ the identity matrix, and $\operatorname{skew}(\bm{S}) = (\bm{S} - \bm{S}^\top)/2$ the skew-symmetric part; throughout, $\operatorname{sym}(\bm{S}) = (\bm{S} + \bm{S}^\top)/2$ denotes the symmetric part. The decomposition places the spectrum of $\hat{\bm{A}}$ in the open left half-plane by construction, rendering the linear part asymptotically stable, and we follow the source literature in labeling this variant Lyapunov-stable.

The second development constrains the reduced quadratic operator $\hat{\bm{H}}$ so that the cubic energy contribution of the underlying convective nonlinearity vanishes exactly, a condition classical in energy-preserving Galerkin models \citep{schlegel2015long} and enforced within OpInf by the Hamiltonian and energy-preserving constructions of \citep{sharma2022hamiltonian, gruber2023canonical, gkimisis2025representation, koike2024energy}. The constraint is most cleanly stated in terms of the third-order tensor $\hat{F}_{ijk}$ obtained from the symmetric expansion of $\hat{\bm{H}}$ via $\hat{\bm{H}}(\hat{\bm{q}} \otimes \hat{\bm{q}}) = \sum_{i,j,k} \hat{F}_{ijk}\, \hat{q}_j \hat{q}_k\, \bm{e}_i$, where $\bm{e}_i$ is the $i$-th canonical basis vector of $\mathbb{R}^r$. The cyclic-skew condition is
\begin{equation}
\hat{F}_{ijk} + \hat{F}_{jki} + \hat{F}_{kij} \;=\; 0
\qquad \forall i, j, k.
\label{eq:cyclic_skew_background}
\end{equation}
The condition is equivalent to the energy-conservation identity $\langle \hat{\bm{q}},\, \hat{\bm{H}}(\hat{\bm{q}} \otimes \hat{\bm{q}}) \rangle = 0$, which holds for the convective term of systems of Navier--Stokes type under periodic or no-flux boundary conditions \citep{schlegel2015long, gkimisis2025representation}. On the sponge-forced plume problem treated later, this cancellation no longer represents conservation of the full-order energy, because the sponge layer and the buoyancy source break the conservative structure. The constraint is therefore used as a reduced-model energy channel closure, with the finite-horizon bound established in \appref{sec:app_energy}. We refer to a fit that imposes Eq.~\eqref{eq:cyclic_skew_background} on the inferred $\hat{\bm{H}}$ as a \emph{cyclic-skew quadratic} fit, naming the tensor identity behind the energy-preserving constraint. The constraint reduces the number of independent operator entries of $\hat{\bm{H}}$ from $r^2(r+1)/2$ in the unconstrained case to $r(r^2 - 1)/3$ after the cyclic-skew identification. Here $r^2(r+1)/2$ counts the independent entries of the quadratic operator across all $r$ output rows, where $r(r+1)/2$ is the compact column width noted earlier, while $r(r^2 - 1)/3$ counts the degrees of freedom that survive the cyclic-skew constraint.

The two structural constraints of Eqs.~\eqref{eq:lyap_A} and \eqref{eq:cyclic_skew_background} are independent of each other and can be applied together or separately. The energy-preserving quadratic constraint has already been applied within OpInf to the dissipative, non-canonical systems of viscous Burgers and Kuramoto--Sivashinsky \citep{koike2024energy}, while the stable-$\hat{\bm{A}}$ parameterization arises in the guaranteed-stable quadratic models developed in \citep{goyal2024guaranteed, pontesduff2024stability}, where the two are typically paired.


\section{Proposed methodology}
\label{sec:methods}

This section develops the proposed parametric OpInf framework. We begin with a diagnosis of why the affine parametric form of \cref{sec:background:popinf} fails on a problem with a regime transition. We then present three methodological components that address the inference, stability, and representation gaps: local operator interpolation, a cyclic-skew quadratic prior, and a regime-bridging basis split. Finally, we integrate these components into the unified CR-OpInf construction. The presentation is self-contained at the level of the method, and the empirical evidence for the choices made here is deferred to the evaluation sections that follow.

\subsection{Diagnosing the multi-parameter compromise}
\label{sec:methods:diagnosis}

Fitting the affine form $\hat{\bm{A}}(\mu) = \hat{\bm{A}}_0 + g(\mu)\, \hat{\bm{A}}_1$ jointly across $M$ training parameter values on a problem with a regime transition produces an inferred linear operator whose eigenvalues reach a positive maximum real part of order unity, quantified below. The resulting reduced trajectories diverge rapidly. The same OpInf regression performed at any single training parameter with no parameter coupling yields a maximum real part that is substantially smaller. The instability is therefore associated with the joint affine fit and not with the individual per-parameter regressions.

We derive an eigenvalue perturbation analysis that links the affine joint-fit mismatch to the spectral drift. Denote by $\hat{\bm{A}}^{\rm single}(\mu_i)$ the single-parameter least-squares optimum at training point $\mu_i$, the result of applying the regression of Eq.~\eqref{eq:opinf_regression} to the snapshots of $\mu_i$ alone, and by $\hat{\bm{A}}^{\rm joint}(\mu_i)$ the affine joint fit of Eq.~\eqref{eq:affine_ansatz} evaluated at the same point. The per-parameter perturbation that the joint fit incurs at $\mu_i$ is then $\bm{E}_i = \hat{\bm{A}}^{\rm joint}(\mu_i) - \hat{\bm{A}}^{\rm single}(\mu_i)$. For every eigenvalue $\lambda$ of $\hat{\bm{A}}^{\rm joint}(\mu_i)$, the Bauer--Fike theorem \citep{bauer1960norms} yields the perturbation bound
\begin{equation}
\min_{k}\,\bigl|\lambda
- \lambda_k\bigl(\hat{\bm{A}}^{\rm single}(\mu_i)\bigr)\bigr|
\;\leq\;
\kappa_2\!\bigl(\bm{V}^{(i)}\bigr) \cdot \|\bm{E}_i\|_2,
\label{eq:bauer_fike_main}
\end{equation}
where $\bm{V}^{(i)}$ is the eigenvector matrix of the diagonalizable $\hat{\bm{A}}^{\rm single}(\mu_i)$ (the finite conditioning $\kappa_2(\bm V^{(i)})$ below confirms diagonalizability) and $\kappa_2$ is the spectral condition number. Because the eigenvector conditioning remains mild on the problems considered here (\figref{fig:bauer_fike}(c)), the bound caps the spectral drift by a multiple of $\|\bm{E}_i\|_2$, which grows when the per-parameter operators are poorly represented by a shared affine $g(\mu)$-family.

Eq.~\eqref{eq:bauer_fike_main} bounds the spectral drift by $\kappa_2 \|\bm{E}_i\|_2$, the largest drift that the affine constraint can induce. The bound itself is directionless and loose in absolute terms on the plume fits (\appref{sec:app_theory}), so we use it qualitatively: it identifies $\|\bm{E}_i\|_2$ as the driver of the drift. The destabilizing direction of the drift is an empirical observation. The single-parameter fit minimizes the residual at $\mu_i$ exactly. The joint affine constraint can therefore only perturb its spectrum, and that perturbation has no a priori reason to leave the maximum real part in the stable half-plane. The joint fit displaces the spectrum toward instability in every case examined here. We make the affine inadequacy quantitative through the relative residual $\rho_{\rm aff}$ of the best affine fit to the per-parameter operators, the residual of that fit normalized by the total operator magnitude (defined in Eq.~\eqref{eq:rho_affine}). On a representative design spanning the parameter range, $\rho_{\rm aff} \approx 0.26$. It drops to $\approx 0.10$ within the sub-regime and to $\approx 0.21$ within the super-regime, the parameter ranges below and above the regime boundary. A substantial fraction of the operator variation is therefore not captured by the affine family, most of it removable on the sub-regime side by splitting at the regime boundary but persisting within the super-regime (\appref{sec:app_theory}). This affine residual provides the scale of the spectral perturbation induced by the joint fit.

\figref{fig:bauer_fike} shows the diagnosis on the plume problem through two empirical facts, both quantified in \appref{sec:app_theory}. First, at the production training set ($M = 17$), the joint affine fit produces a maximum real eigenvalue between $+0.97$ and $+1.78$ across the seventeen training values of the plume stratification parameter, the Brunt--V\"ais\"al\"a frequency $N$ of \cref{sec:problem} with $\mu = N$, and the worst case $+1.78$ sits at the smallest $N$. The single-parameter fits at the same nodes, with the same regularization, give maximum real eigenvalues between $+0.04$ and $+0.10$, smaller by a factor of $18$ to $38$ at every node. The instability belongs to the joint fit, and each per-parameter fit is only mildly unstable. Second, the joint-fit instability cannot be regularized away. Varying the Tikhonov weight over its full range drives the operator norm toward zero. The maximum real eigenvalue therefore descends toward zero from above, yet it stays strictly positive at every $\beta$ and never crosses into the stable half-plane. The descent leads only to a vanishing operator and offers no route to stability, because a penalty large enough to push the growth rate near zero also suppresses the dynamics the operator must represent. No single penalty value therefore yields an operator that is both stable and accurate. This motivates a non-affine treatment of parameter dependence, followed by a separate structural constraint on the inferred quadratic operator.

\begin{figure}[htbp]
\centering
\includegraphics[width=\linewidth]{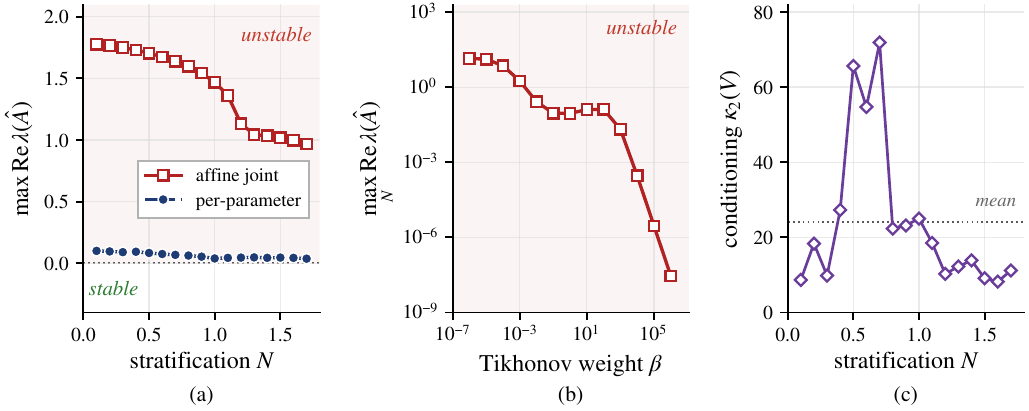}
\caption{Diagnosis of the affine joint fit. (a) Maximum real part of the eigenvalues of the inferred linear operator $\hat{\bm{A}}$ versus $N$ for the affine joint fit and the per-parameter fits ($\beta = 10^{-3}$). (b) The maximum eigenvalue real part, taken over all $N$, versus the Tikhonov weight $\beta$. (c) Eigenvector conditioning $\kappa_2(\bm V^{(i)})$ of the per-parameter operators $\hat{\bm A}^{\rm single}(\mu_i)$ versus $N$.}
\label{fig:bauer_fike}
\end{figure}

\subsection{Local operator interpolation}
\label{sec:methods:loi}

Guided by the Bauer--Fike diagnosis, we develop an alternative treatment that decouples the per-parameter fits. This local operator interpolation (LOI) performs an independent regression at each training parameter and interpolates the resulting operator set when the model is evaluated at a new parameter. Interpolating per-parameter operators provides a non-affine alternative to a shared affine embedding. The affine embedding in \citep{mcquarrie2023nonintrusive} avoids the need for operator interpolation, and the element-wise scheme of \citep{kang2025parametric} infers and interpolates per-parameter OpInf models. The constraint of \cref{sec:methods:cyclicskew} defines a convex set, and a convex kernel blend therefore keeps every blended operator inside that set at every evaluated parameter (\appref{sec:app_energy}), a guarantee that unstructured operator interpolation does not provide. The unconstrained plain least-squares baseline of \cref{sec:burgers} blends per-parameter fits in exactly this way with no constraint, and its divergence on both problems shows the limit of unstructured operator interpolation here.

Concretely, given a training set $\mathcal{M}_{\rm train} = \{\mu_1, \ldots, \mu_M\}$ and a shared reduced basis $\bm{\Phi}$, the LOI construction solves the single-parameter regression of Eq.~\eqref{eq:opinf_regression} independently at every training parameter,
\begin{equation}
\bigl(\hat{\bm{c}}_i, \hat{\bm{A}}_i, \hat{\bm{H}}_i\bigr)
\;=\; \operatorname*{arg\,min}_{\hat{\bm{c}},\,\hat{\bm{A}},\,\hat{\bm{H}}}
\;J^{(i)}\bigl(\hat{\bm{c}}, \hat{\bm{A}}, \hat{\bm{H}}\bigr),
\qquad i = 1, \ldots, M,
\label{eq:loi_per_mu}
\end{equation}
where $J^{(i)}$ is the same Tikhonov-regularized objective with the data matrix $\bm{D}^{(i)}$ of Eq.~\eqref{eq:design_matrix} and the projected time derivatives $\bm{Y}^{(i)} = \bm{\Phi}^\top\!\dot{\bm{Q}}^{(i)}$ formed from the centered snapshots at $\mu_i$, possibly augmented with a structural constraint as developed in \cref{sec:methods:cyclicskew}. The operator at a test parameter $\mu^\ast$ is obtained as the average of the per-parameter operators weighted by a Gaussian kernel,
\begin{equation}
\hat{\bm{O}}(\mu^\ast)
\;=\; \sum_{i=1}^M w_i(\mu^\ast)\, \hat{\bm{O}}_i,
\qquad
w_i(\mu^\ast)
\;=\; \frac{\exp\!\bigl[-(\mu^\ast - \mu_i)^2 / 2\ell^2\bigr]}
            {\sum_{j} \exp\!\bigl[-(\mu^\ast - \mu_j)^2 / 2\ell^2\bigr]},
\label{eq:loi_blend}
\end{equation}
where $\hat{\bm{O}}$ is any of the operators in Eq.~\eqref{eq:loi_per_mu} and $\ell > 0$ is the kernel bandwidth, the single hyperparameter of the interpolation step. The reduced trajectory is integrated in the polynomial form of Eq.~\eqref{eq:opinf_ansatz} with the blended operators of Eq.~\eqref{eq:loi_blend}.

LOI can be read as kernel regression in parameter space, a counterpart of the operator interpolation on matrix manifolds of \citep{amsallem2011online}, which in turn generalizes the basis interpolation on the Grassmann manifold of \citep{amsallem2008interpolation}. The construction simplifies that operator interpolation in two ways. All per-parameter fits use a \emph{common} reduced basis $\bm{\Phi}$ computed from the joint snapshot matrix as in \cref{sec:background:pod}. This removes the coordinate transformation that operator interpolation on matrix manifolds otherwise requires. The principal-angle perturbation of the common basis relative to the ideal per-parameter bases is controlled by a Davis--Kahan companion bound (\propref{prop:davis_kahan}). The interpolation is a Gaussian kernel in Euclidean parameter space, a closed-form alternative to a geodesic interpolation on the Grassmann manifold, and it keeps the construction differentiable in the parameter.

By construction, the per-parameter operator quality of LOI is independent of $M$, since each $\hat{\bm{A}}_i$ is the single-parameter optimum at $\mu_i$ alone and incurs no joint compromise. The eigenvalue perturbation analysis of \cref{sec:methods:diagnosis} therefore does not apply. The quality of the parameter space interpolation is controlled instead by the kernel bandwidth $\ell$ and the density of the training grid, a tradeoff we examine in \cref{sec:plume:robustness}. The per-parameter regression in Eq.~\eqref{eq:loi_per_mu} is written here as a single-step residual on the projected time derivative. A broader family of rollout residuals extends this fit by integrating the model over multiple time steps while sharing the same operator parameterization and the same per-parameter independence. We use the single-step residual throughout this paper and leave the multi-step family to future study.

\subsection{Cyclic-skew prior in a non-canonical dissipative setting}
\label{sec:methods:cyclicskew}

The unconstrained least-squares fit of the reduced quadratic operator $\hat{\bm{H}}$ does not in general satisfy the cyclic-skew condition of Eq.~\eqref{eq:cyclic_skew_background} of \cref{sec:background:structure}, and the inferred quadratic nonlinearity can amplify the reduced trajectory in directions unrelated to the underlying physics. We impose the cyclic-skew condition at the per-parameter inference step. The constraint is encoded through a structured parameterization of the third-order tensor $\hat F_{ijk}$ that automatically satisfies the condition, and the constrained inference is solved by a limited-memory Broyden--Fletcher--Goldfarb--Shanno (BFGS) optimization with the Frobenius residual against the projected time derivative as the loss. Because the cyclic-skew condition is a linear constraint on the operator entries and the Tikhonov-regularized loss is strictly convex, the constrained inference has a unique global minimizer. The limited-memory BFGS iteration converges to that global optimum.

The cyclic-skew condition is classical in energy-preserving Galerkin models \citep{schlegel2015long} and entered operator inference through the canonical Hamiltonian setting \citep{sharma2022hamiltonian, gruber2023canonical}, in which it preserves an underlying energy-conservation identity of the full-order PDE. A buoyant plume governed by the Boussinesq equations with finite viscosity, a stratification source, and an external buoyancy source is not in this setting. The convective term still satisfies the cyclic-skew identity (under appropriate boundary conditions), but the underlying full-order PDE no longer conserves energy. In this non-canonical dissipative setting, we show that the cyclic-skew condition acts as an energy-channel closure that removes cubic energy injection from the learned quadratic term, in place of the conserved invariant it would enforce in the canonical setting. With the reduced energy $E = \tfrac{1}{2}\|\hat{\bm{q}}\|_2^2$, the cyclic-skew condition is exactly the statement that the quadratic term does no net work on the energy, $\hat{\bm{q}}^\top\hat{\bm{H}}[\hat{\bm{q}}\otimes\hat{\bm{q}}] = 0$ for all $\hat{\bm{q}}$. The energy rate $\dot{E} = \hat{\bm{q}}^\top\hat{\bm{c}} + \hat{\bm{q}}^\top\hat{\bm{A}}\hat{\bm{q}}$ is then governed by the linear and constant terms alone. The nonlinearity can redistribute energy among the reduced modes but cannot inject it. The cyclic-skew constraint closes the cubic blowup channel left open by an unconstrained quadratic operator. The cancellation makes the energy rate at most quadratic in $\|\hat{\bm{q}}\|_2$, and a Gr\"onwall estimate, the standard comparison bound for differential inequalities, then caps the growth at exponential and rules out finite-time divergence even for an unstable inferred linear operator (\appref{sec:app_energy}). The reduced trajectory then stays bounded over the long but finite horizons we integrate, and the boundedness observed in practice is quantified in \cref{sec:plume:longtime}. A uniform-in-time bound is not provided by the cyclic-skew constraint alone and is available only under the optional dissipative stability-constrained variant of \appref{sec:app_certified}. Because the cyclic-skew condition is linear in the operator entries, the convex kernel blend of \cref{sec:methods:loi} preserves it exactly. The convex interpolation therefore extends the cyclic-skew property beyond the training nodes to every evaluation parameter (\appref{sec:app_energy}). We recast this classical constraint as a structural prior for the inference of a buoyancy-source-forced reduced model without assuming a reduced Hamiltonian or an explicit dissipation operator.

The structural parameterization also reduces the regression dimension. The cyclic-skew condition removes about one third of the independent entries of $\hat{\bm{H}}$ and yields a convex constrained regression. The constrained fit is a convex program comparable in cost to a standard compact unconstrained least-squares fit. We document the costs in \appref{sec:app_cost}.

\subsection{Regime-bridging via basis splitting}
\label{sec:methods:regimebridge}

We introduce a regime-bridging basis split to address the single-basis representation gap discussed in \cref{sec:intro:gap}. On a problem whose dynamics exhibit a qualitative transition across the parameter range, a single global POD basis trained on the union of all training snapshots compromises between the two regimes' dominant spatial structures and captures neither end well. We construct two regime-side reduced bases by splitting the training set at a regime midpoint $\mu_c$, chosen empirically as the best-performing candidate value (\cref{sec:plume:robustness}) so that it lies within the diffuse physical transition band. The two regime-side training sets $\mathcal{M}_{\rm sub} = \{\mu_i: \mu_i \leq \mu_c + \delta\}$ and $\mathcal{M}_{\rm sup} = \{\mu_i: \mu_i \geq \mu_c - \delta\}$ share a small overlap radius $\delta > 0$ at the boundary. The radius admits the training values nearest the boundary into both regime sides, and is fixed a priori without tuning against test error. The snapshots in each regime-side training set are POD-truncated independently to a regime-specific reduced basis $\bm{\Phi}_{\rm sub}$ or $\bm{\Phi}_{\rm sup}$, each with its own snapshot mean $\bar{\bm{q}}_{\rm sub}$ or $\bar{\bm{q}}_{\rm sup}$ removed as in \cref{sec:background:pod}. To evaluate the model at a parameter $\mu^\ast$, the reduced trajectory is integrated independently in each regime-side basis, using the LOI-blended operators of \cref{sec:methods:loi} fitted on the regime-side training set. We construct the prediction as a sigmoid-weighted combination of the two regime-side reconstructions in the full-order state space,
\begin{equation}
\bm{q}_{\rm pred}(t; \mu^\ast)
\;=\; \bigl(1 - \alpha(\mu^\ast)\bigr)
       \bigl(\bar{\bm{q}}_{\rm sub} + \bm{\Phi}_{\rm sub}\hat{\bm{q}}_{\rm sub}(t; \mu^\ast)\bigr)
\;+\; \alpha(\mu^\ast)
       \bigl(\bar{\bm{q}}_{\rm sup} + \bm{\Phi}_{\rm sup}\hat{\bm{q}}_{\rm sup}(t; \mu^\ast)\bigr),
\label{eq:sigmoid_blend}
\end{equation}
with $\hat{\bm{q}}_{\rm sub}(t; \mu^\ast)$ and $\hat{\bm{q}}_{\rm sup}(t; \mu^\ast)$ the reduced trajectories integrated in the two regime-side bases using the LOI-blended operators at $\mu^\ast$, the sigmoid blend weight $\alpha(\mu^\ast) = [1 + \exp(-(\mu^\ast - \mu_c)/w_\alpha)]^{-1}$, and blend width $w_\alpha > 0$. We combine the reconstructed full-order fields of Eq.~\eqref{eq:sigmoid_blend} directly instead of interpolating the two reduced bases, because basis interpolation would also require transforming the learned third-order quadratic operators between incompatible reduced coordinates. Because the partition acts in parameter space, each regime-side trajectory is integrated in a fixed basis, and the regime-bridging split adds only the blend of reconstructed fields in Eq.~\eqref{eq:sigmoid_blend} on top of the per-regime LOI models.

\subsection{Composition of the framework}
\label{sec:methods:composition}

The three ingredients introduced above are local operator interpolation, the cyclic-skew quadratic constraint, and the regime-bridging basis split. We integrate these components into a single parametric OpInf construction, shown in \figref{fig:framework}. We refer to this construction as \emph{cyclic-skew regime-bridged Operator Inference} (\CROI). We use the cyclic-skew constraint of \cref{sec:methods:cyclicskew} as the nonlinear stability prior, because it closes the cubic energy injection channel identified in that section. We retain the algebraic cyclic-skew form but reinterpret its role. In the present forced dissipative setting it acts as an energy-channel closure in place of an exact conservation law, and it acts together with the regime-bridging basis split of \cref{sec:methods:regimebridge} on the per-parameter operators of LOI. The full configuration is specified by five hyperparameters: the regime midpoint $\mu_c$, the overlap radius $\delta$, the sigmoid blend width $w_\alpha$, the kernel bandwidth $\ell$, and the reduced dimension $r$ per regime side. Their values and motivating indicators are reported in \cref{sec:plume}, together with sensitivity studies over the midpoint $\mu_c$, the bandwidth $\ell$, the blend width $w_\alpha$, and the reduced dimension $r$.

\begin{figure}[htbp]
\centering
\includegraphics[width=\linewidth]{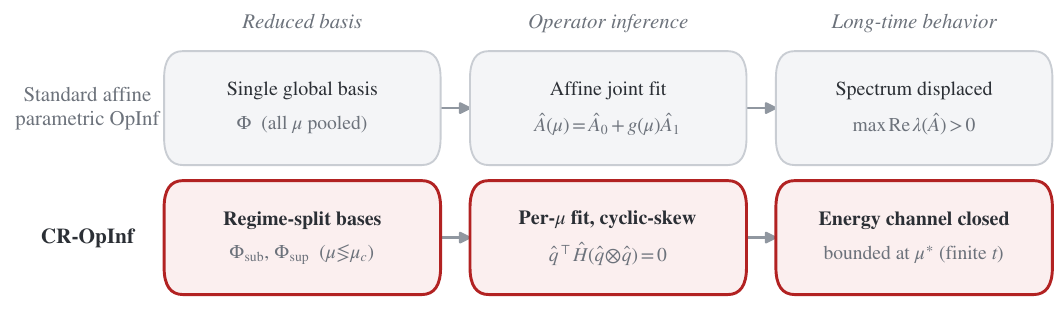}
\caption{The \CROI{} pipeline (bottom) beside the standard affine pipeline (top), in three stages. The reduced basis $\bm{\Phi}$ is a single global basis above and the regime-split pair $\bm{\Phi}_{\mathrm{sub}}, \bm{\Phi}_{\mathrm{sup}}$ below, divided at the regime midpoint $\mu_c$. The inferred operators act on the reduced state $\hat{\bm{q}}$, with the affine linear fit $\hat{\bm{A}}(\mu) = \hat{\bm{A}}_0 + g(\mu)\hat{\bm{A}}_1$ above and the cyclic-skew quadratic condition $\hat{\bm{q}}^\top\hat{\bm{H}}(\hat{\bm{q}}\otimes\hat{\bm{q}})=0$ below. The blended operator is integrated at a test parameter $\mu^\ast$. This gives the displaced spectrum $\max\operatorname{Re}\lambda(\hat{\bm{A}})>0$ above and the finite-horizon bounded trajectory below.}
\label{fig:framework}
\end{figure}


\section{Problem setting}
\label{sec:problem}

The proposed framework is evaluated on two parametric PDEs. The principal application is a two-dimensional Boussinesq buoyant plume in a linearly stratified ambient, with the Brunt--V\"ais\"al\"a frequency as the varied parameter. The cross-problem comparison uses a one-dimensional viscous Burgers equation with viscosity as the varied parameter. Both parametric problems are defined here under a shared discretization convention, so that the later cross-problem comparison of \cref{sec:burgers} rests on one consistent setup.

\subsection{Buoyant plume in a stratified ambient}
\label{sec:problem:plume}

The principal application is a two-dimensional incompressible buoyant plume in a vertical plane $\Omega = [0, L_x] \times [0, L_z]$ with $L_x = 2.5$ and $L_z = 1.5$ in nondimensional units. The domain is open, with periodic lateral boundaries and a wave-absorbing sponge near the top that emulates an unbounded vertical extent. The bottom boundary is a no-slip wall that carries the buoyancy source, the buoyancy satisfies zero-flux conditions at the bottom and top walls, and the velocity is free-slip beneath the sponge at the top (\appref{sec:app_fom}). The plume is therefore free to rise and spread without confinement by solid boundaries. The state vector is the triple $\bm{q}(\bm{x}, t) = \bigl(u, w, b\bigr)$ of the horizontal velocity $u$, the vertical velocity $w$, and the buoyancy $b$, each a function of position and time. Under the Boussinesq approximation, the dynamics are governed by
\begin{subequations}
\label{eq:plume_pde}
\begin{align}
\partial_t u + \bm{u} \cdot \nabla u
   &\;=\; -\partial_x p + \nu \nabla^2 u, \\
\partial_t w + \bm{u} \cdot \nabla w
   &\;=\; -\partial_z p + \nu \nabla^2 w + b, \\
\partial_t b + \bm{u} \cdot \nabla b
   &\;=\; \kappa \nabla^2 b - N^2 w + S(\bm{x}), \\
\nabla \cdot \bm{u} &\;=\; 0,
\end{align}
\end{subequations}
in which $\bm{u} = (u, w)$ is the velocity field and $p$ is the pressure that enforces incompressibility. The coefficients $\nu$ and $\kappa$ are the kinematic viscosity and the buoyancy diffusivity, both set to $10^{-3}$. The source $S(\bm{x})$ is a localized Gaussian buoyancy injection at the lower boundary, and $N$ is the Brunt--V\"ais\"al\"a frequency that parameterizes the ambient stratification strength. The buoyancy equation term $-N^2 w$ encodes the restoring effect of the stable stratification. A fluid parcel displaced vertically by an amount $\delta z$ experiences a buoyancy restoring force proportional to $-N^2\, \delta z$.

The parameter dependence of the dynamics is exactly affine in $N^2$, since the parameter appears linearly in the buoyancy equation and nowhere else. The reduced linear operator inherits this structure, and the two-term basis $\{1, N^2\}$ is therefore physics-exact. This corresponds to $g(N) = N^2$ in the affine reduced model form of \cref{sec:background:popinf}. A more general affine combination such as $\{1, N, N^2\}$ would introduce an $N$-linear operator term with no counterpart in the governing equations. This makes the plume a natural test problem for parametric OpInf, because at the PDE level the affine form is a physics-grounded encoding of the parameter dependence.

The motivating applications are three-dimensional and turbulent, but we treat this two-dimensional laminar model as a grid- and time-converged proxy that reproduces the qualitative regime phenomenology across the transition, the buoyancy-driven rise, overshoot, internal-wave trapping, and intrusion at the level of neutral buoyancy, without turbulence closure or quantitative validation against laboratory or three-dimensional data. The three-dimensional and wind-forced extensions are left to future work.

The training range is $N \in \{0.1, 0.2, \ldots, 1.7\}$, $M = 17$ values spaced uniformly in $N$. Each trajectory is integrated to $t = 20$ and sampled at output interval $\Delta t = 0.025$, for $K = 801$ snapshots per training parameter on a $240 \times 144$ grid. The reduced model is trained on the in-window horizon $t_{\rm train} = 12$. Companion long-time runs to the test horizon $t_{\rm final} = 60$, five times that window, are generated at every training parameter value for the forecast evaluation and the long-time stability assessment of \cref{sec:plume:longtime}. Eq.~\eqref{eq:plume_pde} is discretized with the full-order numerics detailed in \appref{sec:app_fom}. That appendix covers the staggered grid, the time integrator, the pressure projection, the sponge layer, the lateral boundary conditions, and the grid-refinement study.

\paragraph{Regime transition driven by $N$}
The parameter $N$ controls a qualitative regime transition. At weak stratification ($N \lesssim 0.4$) the plume rises unimpeded into the absorbing sponge layer at the top of the domain, a region that damps upward-propagating internal waves so the truncated vertical extent produces no spurious reflections. The weakly stratified end is therefore domain-limited, and its trajectory is dominated by the buoyant-rise phase. Here \emph{domain-limited} means only that the buoyant rise reaches the sponge layer. These weakly stratified trajectories remain valid, high-fidelity training data. A transparency study in \appref{sec:app_fom} confirms the sponge does not contaminate the bulk plume dynamics. The unconfined behavior the method targets is the trapping regime, where the plume settles at a height of neutral buoyancy well below every boundary. The lateral boundaries are periodic, and the lateral extent is chosen so that the intrusion at the level of neutral buoyancy reaches them only late in the run. This choice limits periodic wrap-around contamination of the trapping height. At strong stratification ($N \gtrsim 0.8$) the plume is trapped at a height of neutral buoyancy determined by the balance between source injection and ambient stratification, and it overshoots and settles with lateral intrusion in the manner of turbulent fountains in a stratified ambient \citep{bloomfield1998intrusion}. The trajectory is then dominated by internal-wave oscillation around the trapping height. At intermediate $N$ both phases coexist. The transition is gradual, and several physical indicators place it within the band $N \in [0.35, 0.8]$. The rise time $\tau_{\rm rise}$ and the Brunt--V\"ais\"al\"a oscillation period $\tau_{\rm BV} = 2\pi / N$ cross at $N \approx 0.35$. The first post-peak oscillation of the rise trajectory appears between $N = 0.5$ and $0.6$. A spectral indicator of the buoyancy field marks the onset of oscillation dominance near $N = 0.8$. None of these indicators isolates the transition by itself. We take a representative value $N \approx 0.5$ within this band at the onset of internal-wave oscillation as the single regime boundary used by the basis split and as the divider between free rise and trapping in the figures. The basis-split study of \cref{sec:plume:robustness} confirms it as the empirical optimum, and the primary in-window error metric is insensitive to the tested midpoints $\mu_c \in \{0.4, 0.5, \ldots, 0.9\}$. \Figref{fig:plume_N_regime} contrasts the two regimes at the weak and strong representative values of $N$ and shows that the two regimes produce qualitatively distinct long-time buoyancy fields by $t = 60$. At the weak value the plume rises into the sponge and leaves a buoyant column, while at the strong value it settles into a trapped, laterally spreading intrusion that oscillates about its height of neutral buoyancy. The centroid trajectories of all $17$ training $N$ interpolate continuously between these two endpoints across the transition band. The parameter dependence is therefore smooth even where the spatial structure changes character. The qualitative differences across the regimes are the modeling challenge a single-basis parametric ROM must overcome.

\begin{figure}[htbp]
\centering
\includegraphics[width=\linewidth]{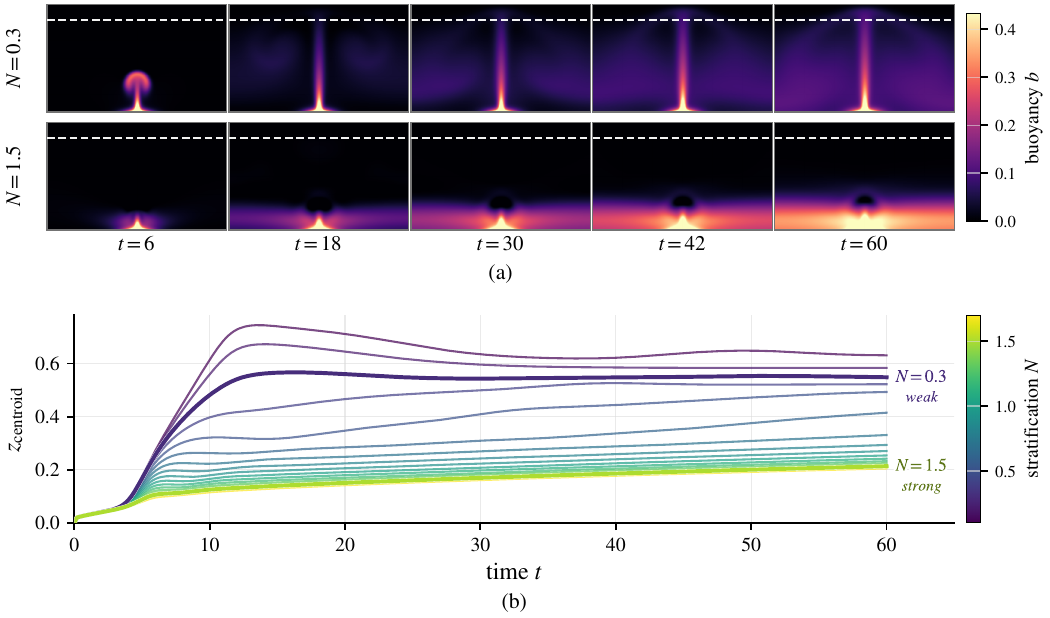}
\caption{(a) Full-order buoyancy field $b(x, z)$ at $N = 0.3$ and $N = 1.5$ at five times. The dashed line marks the absorbing-sponge base. (b) Buoyancy-weighted vertical centroid $z_{\rm centroid}(t)$ for the $17$ training $N$. Color encodes $N$.}
\label{fig:plume_N_regime}
\end{figure}

\subsection{Viscosity-parametric Burgers problem}
\label{sec:problem:burgers}

The cross-problem comparison uses a one-dimensional viscous Burgers equation on the periodic interval $[0, 1]$,
\begin{equation}
\partial_t u + u\, \partial_x u
\;=\; \nu\, \partial_x^2 u,
\qquad u(0, t) = u(1, t),
\label{eq:burgers_pde}
\end{equation}
with the Gaussian initial condition $u_0(x) = A \exp[-(x - x_0)^2 / 2 w_0^2]$ ($A = 1$, $x_0 = 0.3$, $w_0 = 0.1$) and kinematic viscosity $\nu$ as the varied parameter. The training set consists of the six values $\nu \in \{10^{-2},\, 5{\times}10^{-3},\, 2{\times}10^{-3},\, 10^{-3},\, 5{\times}10^{-4},\, 10^{-4}\}$ that span a smooth-to-shock transition centered near $\nu_c = 2 \times 10^{-3}$. The midpoint $\nu_c$ is identified from a characterization of the full-order data. The maximum spatial gradient $\max_x |\partial_x u(\cdot, t)|$ stays bounded by a fixed constant over the simulated time interval when $\nu \geq \nu_c$ but grows monotonically toward a near-singular profile when $\nu < \nu_c$. This reflects the balance between the diffusive term $\nu \partial_x^2 u$ and the nonlinear steepening produced by the convective term $u \partial_x u$ on the simulation time scale. The Gaussian initial condition and the periodic domain admit the classical estimate of the shock-formation time $t_{\rm shock} \sim w_0 / A$ \citep{whitham1974linear}. With these parameters and a simulation horizon comparable to $t_{\rm shock}$, the critical viscosity above which diffusion suppresses shock formation falls in the same decade as $\nu_c$, in agreement with the empirical midpoint we use. Below $\nu_c$ the solution exhibits sharp shock formation, while above $\nu_c$ it remains smooth.

The viscous Burgers equation is the standard one-dimensional model problem in parametric reduced-order modeling. It is unforced and has a single quadratic nonlinearity. It exhibits a smooth-to-shock parametric transition that parallels the subcritical-to-supercritical regime transition of the plume. It differs from the plume in field structure, evolving a single scalar field where the plume couples three. It carries no source term where the plume has a continuous boundary source, and it is one-dimensional where the plume is two-dimensional. Using Burgers as the comparison problem isolates which of the proposed methodological ingredients are problem-specific and which are general. This is the distinction that organizes the cross-problem comparison of \cref{sec:burgers}. Eq.~\eqref{eq:burgers_pde} is discretized by a second-order finite-difference scheme in space on a grid of $1{,}024$ points in $x$ and advanced by an explicit fourth-order Runge--Kutta time stepper. The in-window training horizon is $t_{\rm train} = 0.5$ and the test horizon is $t_{\rm final} = 2.5$, five times the training window and the same horizon multiple used for the plume. With snapshots sampled uniformly in time as for the plume, each trajectory gives $K = 501$ snapshots per parameter at output interval $\Delta t = 0.005$ on the $1{,}024$-point grid. For both problems the reduced order training pipeline takes the snapshot matrices as input and produces the parameter-dependent reduced operators of Eq.~\eqref{eq:opinf_ansatz} through the methodology of \cref{sec:methods}, at a computational cost documented in \appref{sec:app_cost}.


\section{Application to the buoyant plume}
\label{sec:plume}

This section evaluates \CROI on the two-dimensional buoyant plume of \cref{sec:problem:plume}. We report the accuracy of the reduced model with respect to the full-order solution (\cref{sec:plume:accuracy}), the recovery of the Morton--Taylor--Turner plume rise scaling (\cref{sec:plume:mtt}), long-time stability beyond the training horizon (\cref{sec:plume:longtime}), and robustness to the training set size, the reduced dimension, and the randomized basis construction (\cref{sec:plume:robustness}).

\subsection{Accuracy under leave-one-out evaluation}
\label{sec:plume:accuracy}

The principal accuracy assessment is leave-one-out over the $17$ training values of $N$. For each test parameter, the proposed framework is retrained on the remaining $16$ training values and evaluated on the held-out parameter through the integration of Eq.~\eqref{eq:opinf_ansatz} with the blended operators of \cref{sec:methods:loi}. The training is performed at the fixed configuration $(\mu_c, \delta, w_\alpha, \ell, r) = (0.5, 0.15, 0.05, 0.05, 20)$, where $\mu_c$ is the regime midpoint, $\delta$ the overlap radius, $w_\alpha$ the sigmoid blend width, $\ell$ the kernel bandwidth, and $r$ the reduced dimension per regime side (\cref{sec:methods:composition}). The per-parameter regression uses uniform Tikhonov regularization at weight $\beta = 10^{-3}$. The spectrum-aware per-mode weighting of \cref{sec:background:fit} is assessed separately as a robustness check (\cref{sec:burgers:alpha}). \Tabref{tab:plume_accuracy} summarizes the resulting accuracy and stability statistics. Before turning to the aggregate metrics, \figref{fig:plume_reconstruction} shows that the reduced reconstruction recovers the full-order fields in window across both stratification regimes and across all three state fields, with the two velocity components reconstructed as faithfully as the buoyancy. The rising mushroom cap of the weakly stratified regime and the trapped, laterally spreading structure of the strongly stratified regime are both reproduced to within a few percent relative Frobenius error. At the snapshot time $t = 6$ shown, the relative Frobenius error is $2.9\%$ ($u$), $1.9\%$ ($w$), and $6.3\%$ ($b$) at $N = 0.3$ and $4.6\%$, $3.9\%$, and $2.7\%$ at $N = 1.5$. The pointwise field error grows in the time-forecast windows, fastest in the weak free rise regime, while the integrated rise stays accurate (\cref{sec:plume:longtime}). The energy mechanism bounds the trajectory without pinning the pointwise field, which is why the field error and the integrated-rise error separate. The forecast field errors and rise-quantity errors are broken down across the structural variants in \appref{sec:app_multimetric}.

\begin{figure}[htbp]
\centering
\includegraphics[width=\linewidth]{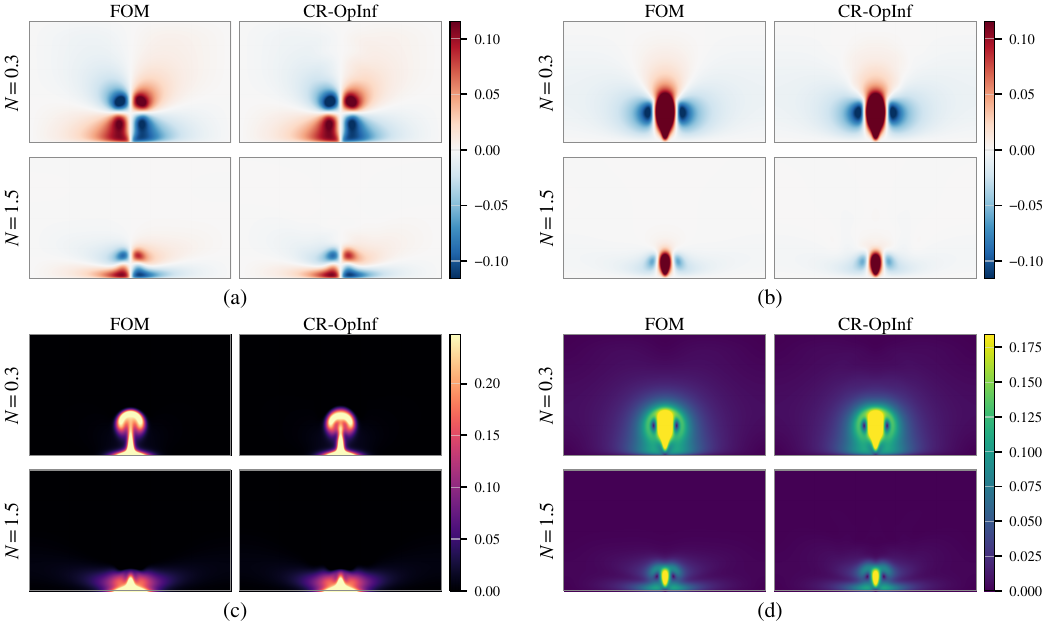}
\caption{In-window ($t = 6$) full-order and \CROI{} reconstruction at $N = 0.3$ and $N = 1.5$. (a) Horizontal velocity $u$. (b) Vertical velocity $w$. (c) Buoyancy $b$. (d) Speed $|U| = \sqrt{u^2 + w^2}$. The two velocities share one diverging color scale, clamped so that the $N = 0.3$ vertical-velocity core saturates. The buoyancy and the speed use sequential scales.}
\label{fig:plume_reconstruction}
\end{figure}

\begin{table}[t]
\centering
\caption{Accuracy of the proposed framework on the buoyant plume, leave-one-out over $17$ training values of $N$. The forecast windows reach $t = 20$ and $t = 60$ against the training horizon $t_{\rm train} = 12$. Here $z_{\max}(t)$ is the trajectory of the plume top and $z_{\rm centroid}(t)$ the buoyancy-weighted vertical centroid.}
\label{tab:plume_accuracy}
\small
\begin{tabular}{lrr}
\toprule
Metric & Mean & Maximum \\
\midrule
In-window Frobenius error ($t \leq 12$) & $7.91\%$ & $13.67\%$ \\
Forecast-1 error ($12 < t \leq 20$) & $26.65\%$ & $48.68\%$ \\
Forecast-2 error ($20 < t \leq 60$) & $61.74\%$ & $78.61\%$ \\
QoI correlation for $z_{\max}(t)$ (Pearson) & $0.90$ & -- \\
$z_{\max}$ relative root mean square error (over $N$) & $15.4\%$ & -- \\
$z_{\rm centroid}$ relative root mean square error & $11.2\%$ & -- \\
Trajectories stable at $t = 20$ & \multicolumn{2}{c}{$17/17$} \\
Trajectories stable at $t = 60$ & \multicolumn{2}{c}{$17/17$} \\
\bottomrule
\end{tabular}
\end{table}

The in-window Frobenius error and the quantity of interest correlation of \tabref{tab:plume_accuracy} measure complementary aspects of the reconstruction, the first the full state-level discrepancy and the second the bulk plume rise. The in-window Frobenius error of $7.91\%$ is the most demanding metric the table contains, because it penalizes every spatial and temporal residual at full resolution. To calibrate it, an unconstrained OpInf model fitted to a single stratification at its own optimal rank reaches a best-case in-window error of about two to five percent on this flow, $2.2\%$ at $N = 0.5$ and $4.9\%$ at $N = 1.3$. This floor is set by the advection- and oscillation-dominated, slowly decaying POD spectrum and is not a fitting deficiency. The $7.91\%$ leave-one-out error therefore sits within a small factor of the per-parameter optimum, and the residual gap is the cost of one surrogate that stays stable across the whole parameter range, relative to seventeen separately tuned per-parameter models. The quantity of interest (QoI) Pearson correlation of $0.90$ on the maximum-rise trajectory $z_{\max}(t)$ is the bulk-behavior counterpart, computed on the integrated time series instead of the full state field. Bulk plume rise is the decision-relevant quantity for the motivating engineering applications. A Pearson correlation of $0.90$ on the trajectory indicates that the reduced model reproduces the timing and shape of the plume rise dynamics, with the amplitude error recorded separately by the $z_{\max}$ root mean square error (RMSE) in \tabref{tab:plume_accuracy}. This fixed configuration is used unchanged across all $17$ test parameters and is stable at every one of them, and it carries over to Burgers as one cell of the cross-problem comparison of \cref{sec:burgers}, where the structural constraints apply unchanged while the problem-scale settings are fixed once from the Burgers data. The breakdown of these metrics across the four structural variants is given in \appref{sec:app_multimetric} (\figref{fig:plume_multimetric}).

The spectral diagnosis of \cref{sec:methods:diagnosis} already indicates that the instability of the affine baseline cannot be regularized away, and the long-time integration confirms it directly. The unconstrained affine baseline completes $0$ of $17$ trajectories to $t = 60$ at any regularization strength. Its lightest stable setting reaches a $4.5\%$ in-window error on the $12$ of $17$ trajectories it keeps bounded to $t = 20$, and its strongest setting holds all $17$ only to $t = 20$ at a degraded $29\%$, whereas the cyclic-skew constraint holds all $17$ trajectories to $t = 60$ at a $7.9\%$ in-window error.

\subsection{Recovery of the Morton--Taylor--Turner plume rise scaling}
\label{sec:plume:mtt}

A physics-anchored check for any plume ROM is the Morton--Taylor--Turner similarity prediction $z \sim N^{-3/4}$ for the trapping height of a buoyant plume in a stratified ambient \citep{morton1956turbulent}. We extract the buoyancy-weighted vertical centroid $z_{\rm centroid}(t)$ at long time from each full-order and reduced trajectory and fit a log-log power law in $N$. The more obvious choice $z_{\max}$ is sponge-clipped at the weakly stratified end (the plume reaches the sponge layer before the stratification has time to act) and sensitive to the phase of the internal-wave oscillation at the strongly stratified end. We therefore use the mass-weighted $z_{\rm centroid}$, a robust indicator of the trapping height to which the Morton--Taylor--Turner similarity applies.

The fit covers the thirteen training values $N \in \{0.5, 0.6, \ldots, 1.7\}$ in the trapping band. Below $N \approx 0.5$ the rising plume reaches the sponge layer before the stratification sets the trapping height, and $z_{\rm centroid}$ therefore saturates and falls off the power law. The entire strongly stratified end, including $N = 1.6$ and $1.7$, remains consistent with the fitted scaling. The data themselves therefore set the upper limit, and no arbitrary cutoff is imposed. Across this band, the full-order trapping height is well fit by a power law in $N$, and a least-squares log-log fit yields a slope of $-0.731$, with $R^2 = 0.968$ (\figref{fig:plume_mtt}(a)). Interpreting this exponent requires keeping the dimensionality of the simulation in mind. The classical Morton--Taylor--Turner result $z \sim N^{-3/4}$ is derived for an axisymmetric plume, whereas dimensional analysis for an idealized two-dimensional line source gives $z \sim N^{-1}$. The measured slope is far closer to the axisymmetric $-3/4$ than to the two-dimensional line-source $-1$. It matches the axisymmetric value to within about $3\%$ and misses the line-source value by more than a quarter. That a two-dimensional simulation tracks the axisymmetric exponent instead of the ideal line-source exponent is consistent with the finite source width, the finite vertical extent of the computational domain, and the post-transient window in which $z_{\rm centroid}$ is evaluated. All of these place the flow outside the asymptotic self-similar line-source regime. The measured value is an effective exponent, and we report its proximity to $-3/4$ as an empirical match to the axisymmetric law, without claiming recovery of the asymptotic line-source similarity.

Recovering the $N$-dependence alone is a necessary check that stops short of decisive evidence, since that exponent is partly fixed by dimensional consistency. We therefore also vary the buoyancy flux $F$, the vertically integrated rate at which the source supplies buoyancy and the second parameter of the classical scaling. For the Gaussian source of \appref{sec:app_fom}, this flux is proportional to the source amplitude. Varying the source strength (proportional to $F$) at fixed $N = 1.0$ in the full-order sensitivity study gives $z_{\rm centroid} \sim F^{+0.244}$ ($R^2 \approx 1$). This is within about $2.5\%$ of the axisymmetric flux exponent $+1/4$ and again far from the two-dimensional line-source value $+1/3$. The trapping heights from the two studies collapse onto the single similarity variable $F^{1/4} N^{-3/4}$ with a residual coefficient of variation of $1.3\%$. Because the two exponents are fit independently, the collapse only checks their consistency and does not by itself establish the similarity. We next test whether the reduced model reproduces this full-order behavior.

Moreover, the full-order slope settles over the deep window: fitting $z_{\rm centroid}$ over $t \in [50, 60]$ changes it by less than $0.01$ relative to the $t = 60$ value, whereas over the moderate window $t \in [16, 20]$ the full-order fit is steeper, at $-0.80$, because the weakly stratified trapping heights are still settling at that horizon. The slope recovered by the reduced model also depends on the evaluation horizon. Over the deep window $t \in [50, 60]$, five times the training horizon, it is $-0.40$, with $R^2 = 0.73$. Over the moderate window $t \in [16, 20]$, where the reduced trajectory still carries only the moderate forecast error of \tabref{tab:plume_accuracy}, the recovered slope is $-0.70$ ($R^2 = 0.95$), within $0.04$ of the converged full-order value and shallower than the same-window full-order fit of $-0.80$. The horizon dependence is the expected signature of forecast-error accumulation. The reduced model reproduces the scaling over the horizon where its trajectory remains reliable, while the fitted slope becomes shallower as error accumulates over the long forecast, and we therefore report the slope at both horizons so that the agreement over the reliable window and the degradation over the long horizon are both visible.

\begin{figure}[htbp]
\centering
\includegraphics[width=\linewidth]{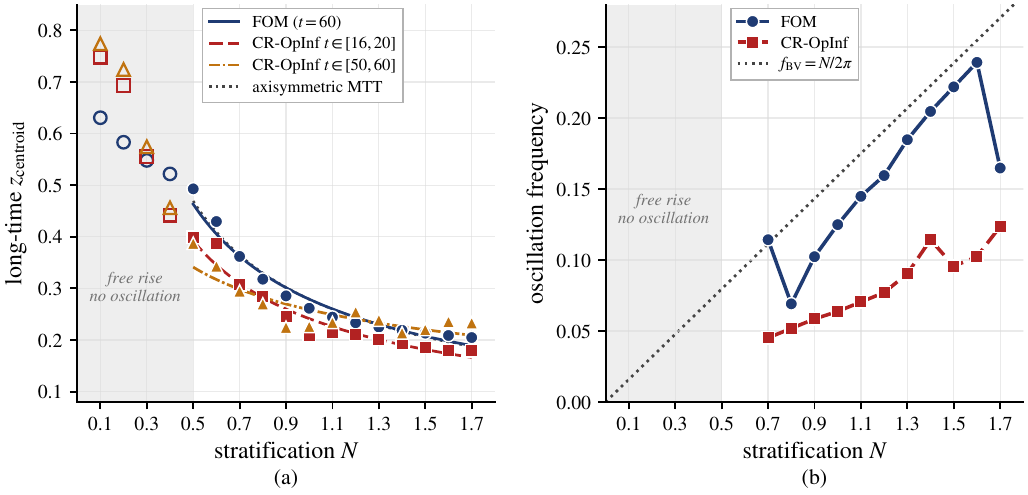}
\caption{(a) Long-time buoyancy-weighted centroid $z_{\rm centroid}$ versus $N$ for the full-order solution at $t = 60$ and \CROI{} over windows $t \in [16,20]$ and $[50,60]$. Solid markers and the log-log power-law fits cover $N \in [0.5, 1.7]$. Hollow markers lie in the shaded $N < 0.5$ band. The dotted line is the axisymmetric Morton--Taylor--Turner slope ($-0.75$). (b) Dominant post-transient centroid-oscillation frequency versus $N$ for the full-order solution and \CROI{} over the trapped regime $N \geq 0.7$, with the Brunt--V\"ais\"al\"a line $f_{\rm BV} = N/2\pi$.}
\label{fig:plume_mtt}
\end{figure}

This scaling recovery situates the present work within the small body of data-driven plume modeling. The closest data-driven effort regresses a low-order analytical model of an unconfined plume's centerline and edge envelope from video \citep{salas2025plume}. That model is tied to a single unstratified laboratory configuration and does not represent the resolved field. It addresses neither the parametric stratification nor the trapping at the level of neutral buoyancy that organize the present problem. The reduced model here instead resolves the full two-dimensional field across the stratification range and recovers the trapping-height scaling from the resolved trajectory. To our knowledge this is the first data-driven reduced-order model of a stratified buoyant plume to recover the Morton--Taylor--Turner plume rise scaling.

\subsection{Long-time stability and time-resolved trajectories}
\label{sec:plume:longtime}

The leave-one-out trajectories of \CROI remain non-divergent for all $17$ test parameters throughout $t \in [0, 60]$, five times the training horizon $t_{\rm train} = 12$. This boundedness is consistent with the energy mechanism of \cref{sec:methods:cyclicskew}. The per-parameter cyclic-skew operators are themselves mildly unstable, with a worst-case growth rate $\lambda_{\max}(\operatorname{sym}\hat{\bm{A}})$ (the energy-based measure taken up in \cref{sec:discussion}) of roughly $+0.24$ to $+0.46$ at the super-regime nodes ($N \geq 0.6$). This symmetric part drives the energy balance, and its largest eigenvalue reaches $+0.46$ at the worst blended parameter in the leave-one-out evaluation (\appref{sec:app_certified}). The reduced trajectory nonetheless stays bounded because the cyclic-skew constraint makes the quadratic term do no net work on the energy (\propref{prop:energy_neutral}). With the finite-time blowup channel closed, the mildly unstable linear part permits at most exponential growth, and the integrated trajectories in fact grow slowly and do not diverge over the integration horizon. In contrast, the affine joint fit displaces the spectrum to positive real parts of order unity (\figref{fig:bauer_fike}) while leaving the quadratic channel open, and diverges. An exploratory unconstrained plain least-squares configuration at $r = 25$, $\mu_c = 0.9$ shows an abrupt transition to divergence that is not observed under the cyclic-skew constraint, a conditioning property of the unconstrained fit independent of the basis (\cref{sec:discussion:synthesis}).

A spectral check of the strongly stratified trajectories ($N \geq 0.7$) makes the interpretation as a trapped oscillation precise. On the post-transient window the full-order centroid oscillates at a frequency that broadly increases with $N$ and approaches the cyclic Brunt--V\"ais\"al\"a frequency $f_{\rm BV} = N/2\pi$ from below (\figref{fig:plume_mtt}(b)). We report a frequency only for $N \geq 0.7$. Below that, near the regime boundary, the plume is only marginally trapped and the post-transient centroid signal is weak. Its dominant spectral peak lies below $f_{\rm BV}$ but is barely stronger than the slow settling trend, with a peak-to-floor power ratio near unity compared with more than threefold at the trapped nodes, and direct inspection of the spectra at $N = 0.5$ and $0.6$ confirms that no internal-wave frequency is cleanly resolved there. Over $N \in [1.1, 1.6]$ the two agree to within about $17\%$. At $N \leq 1.0$ the short post-transient window resolves the spectral peak only to the nearest few bins of the discrete Fourier transform. This resolution supports the increasing trend but is too coarse for a quantitative match, and the isolated $N = 1.7$ estimate is limited by the same bin resolution. The approach to $f_{\rm BV}$ is the behavior expected at leading order for the trapped motion here. In the post-transient window the buoyancy fluctuation energy is almost entirely horizontally uniform. More than $98\%$ lies in the zero-horizontal-wavenumber component (a fraction that rises monotonically with $N$), and a single dominant vertical mode sits at $k_z \approx 2\pi / L_z$. The trapped oscillation is therefore a near-vertical buoyancy heave, for which the frequency tends to the buoyancy frequency as the horizontal wavenumber vanishes. The small offset below $f_{\rm BV}$ is the finite-amplitude and finite-domain correction to that limit. The reduced model reproduces the oscillatory regime and its broadly increasing trend with $N$, but at roughly half the full-order frequency. This factor-of-two error in the dominant internal-wave frequency is a qualitative limitation in the trapped regime. The reduced model is therefore not suited to spectral or resonance many-query tasks there, while the integrated-rise quantities that motivate this class of surrogates are insensitive to it. This sharpens the distinction between integrated-rise accuracy and pointwise field fidelity taken up in \cref{sec:extrap:correlation}. The single-basis baselines are omitted from this frequency comparison, since they diverge before the post-transient window the frequency is measured on. Both fits diverge before $t = 60$ at every one of the four parameters of \figref{fig:plume_time_resolved}, the plain least-squares fit between $t \approx 8$ and $t \approx 26$ and the Lyapunov-$\bm{A}$ fit between $t \approx 18$ and $t \approx 48$, and a post-transient frequency is therefore undefined for them.

\Figref{fig:plume_time_resolved} shows that the cyclic-skew model tracks the integrated rise across the whole regime range while the two single-basis variants do not survive the forecast. The two rows show complementary rise diagnostics. The top row is $z_{90}(t)$, the height below which $90\%$ of the vertically integrated buoyancy lies, a robust measure of the plume top that is smoother and less sensitive to transient cap detachment than the pointwise maximum-height trajectory. The bottom row is the buoyancy-weighted centroid $z_{\rm centroid}(t)$. At the production reduced dimension $r = 20$ selected in \cref{sec:plume:robustness}, the reduced model tracks the full-order trajectory closely in the in-window block at all four parameters. The reduced plume top begins above the full-order solution at $t = 0$ and settles onto it within the first few time units, since the mean-subtracted basis represents the near-empty initial field only weakly, leaving the dynamics and the long-time behavior unaffected. In the deep-forecast block beyond $t = 20$ the reduced trajectory remains qualitatively faithful, and the rise and trapping height are preserved while it accumulates phase and amplitude errors of the magnitude expected from the in-window error level. Both rows also show the two single-basis variants that drop the cyclic-skew constraint. The plain least-squares and Lyapunov-$\bm{A}$ models both diverge well before $t = 60$, each marked at its integrator-failure time (the plain least-squares fit fails inside the training window at the weakest stratification shown), while only the cyclic-skew model integrates to $t = 60$ at every parameter, the trajectory-level counterpart of the cross-problem structural comparison of \cref{sec:burgers}.

\begin{figure}[htbp]
\centering
\includegraphics[width=\linewidth]{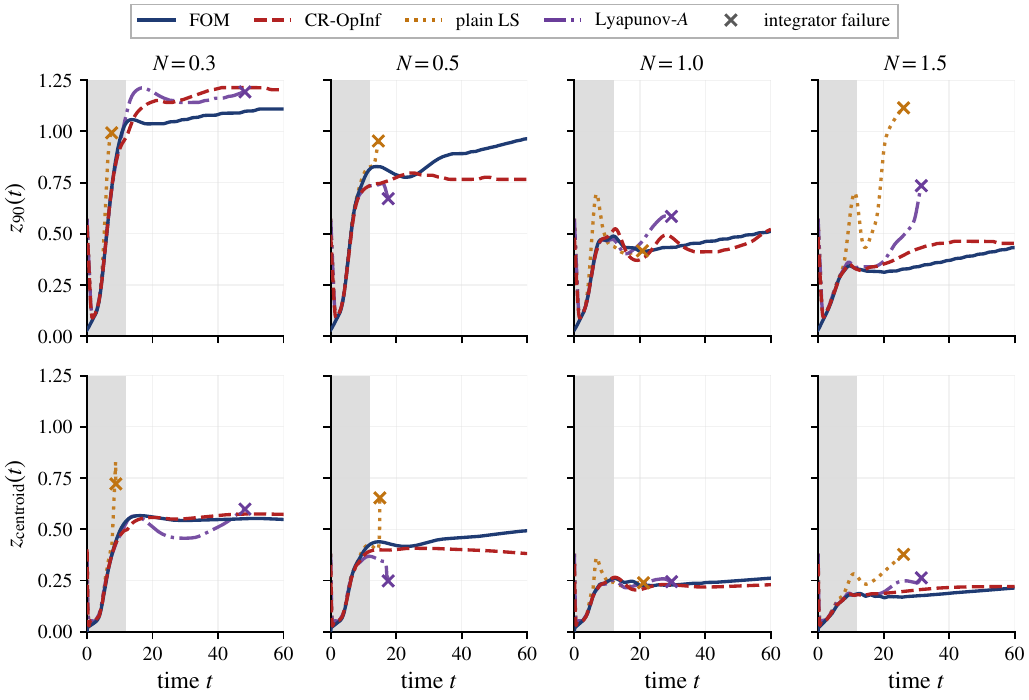}
\caption{Buoyancy-mass 90th-percentile height $z_{90}(t)$ in the top row and buoyancy-weighted centroid height $z_{\rm centroid}(t)$ in the bottom row, for the full-order solution, \CROI{}, and the plain least-squares and Lyapunov-$\bm{A}$ single-basis variants, one column per parameter $N \in \{0.3, 0.5, 1.0, 1.5\}$, integrated to $t = 60$. Both rows share one vertical scale. Crosses mark integrator failure of the diverging variants, and the shaded band is the training window $t \leq 12$.}
\label{fig:plume_time_resolved}
\end{figure}

\subsection{Robustness to the training set, hyperparameters, and initial condition}
\label{sec:plume:robustness}

The accuracy and stability reported above are evaluated at the fixed training configuration. We now assess the sensitivity of these results to six control variables in turn: the size of the training parameter set, the reduced dimension, the location of the regime midpoint, the kernel bandwidth, the sigmoid blend width, and the random seed used by the truncated singular value decomposition that constructs the reduced basis. Finally, we test robustness to perturbations of the initial condition.

\paragraph{Training-set size}
The local operator interpolation framework of \cref{sec:methods:loi} fits a separate operator at each training parameter. The per-parameter fit quality therefore does not depend on how many training parameters are present. What depends on the training set size is the kernel-interpolation gap between them. We rebuild the construction with training sets of size $M \in \{3, 5, 7, 9, 11, 13, 15, 17\}$ obtained by uniform sub-sampling of the full $17$-parameter grid. Every $M$ is evaluated on the same leave-one-out measure over all $17$ parameters as the full grid. The in-window error distribution at each $M$ is therefore directly comparable. \Figref{fig:sample_complexity} shows that distribution as a box over the $17$ parameters, for the proposed cyclic-skew model and for the two alternative fits that share its regime-split basis, the unconstrained least squares and the Lyapunov-stable linear operator. These two baselines are placed here on the same regime-split basis as the proposed model to isolate the effect of the constraint, whereas the long-time trajectories of \cref{sec:plume:longtime} show them in their single-basis form (\tabref{tab:transfer_matrix}, regime bridging off). For the cyclic-skew model the median in-window error falls smoothly with training set size, from $33\%$ at $M = 3$ to $8\%$ at the full grid, where it reproduces the accuracy of the proposed model reported in \cref{sec:plume:accuracy}. The distribution stays tight at every $M$. The cyclic-skew and Lyapunov-stable fits track each other closely on this in-window measure, since the difference the constraints make is one of long-time stability and leaves in-window accuracy essentially unchanged. The unconstrained fit separates from both in the upper tail. Its median is comparable, but its in-window error blows up at isolated parameters at every $M$. The full-range whisker reaches several thousand percent at the worst parameter even on the full grid, whereas the constrained fits stay below $90\%$ throughout. This is the in-window manifestation of the long-time separation reported in \cref{sec:plume:longtime} and \appref{sec:app_multimetric}. The smooth degradation of the median with coarser sampling reflects the widening of the kernel-interpolation gap. The construction is therefore robust to a factor-of-six reduction in the training set size, and the accuracy degrades smoothly without ever collapsing.

\begin{figure}[htbp]
\centering
\includegraphics[width=\linewidth]{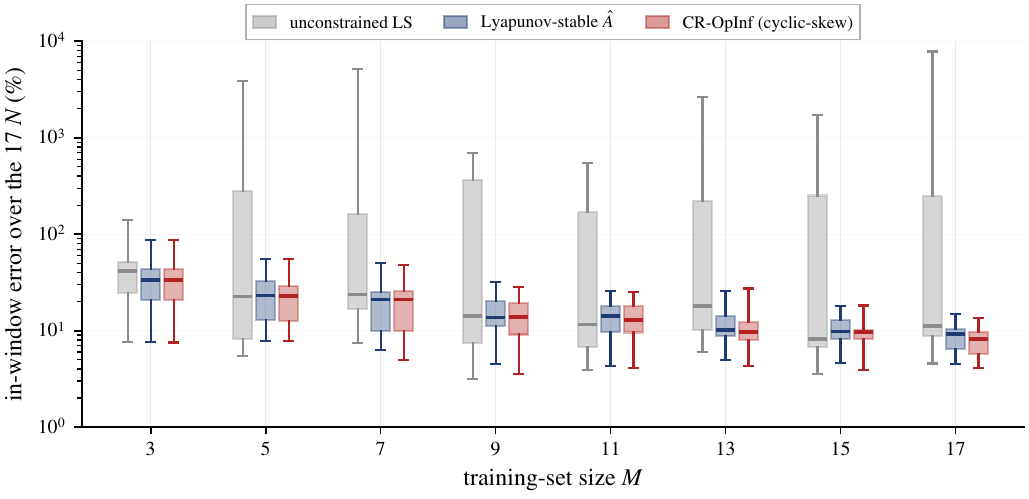}
\caption{In-window Frobenius error versus training set size $M$, for three fits sharing the regime-split basis, the unconstrained least squares, the Lyapunov-stable linear operator, and the proposed cyclic-skew constraint. Each box aggregates the $17$ leave-one-out parameters at one $M$ and spans the interquartile range and median, with full-range whiskers.}
\label{fig:sample_complexity}
\end{figure}

\paragraph{Reduced dimension}
We vary the reduced dimension $r \in \{10, 15, 20, 25, 30\}$ on the production training set. All five rank values yield full long-time stability and produce in-window errors that decrease monotonically from $12.4\%$ at $r = 10$ to $7.8\%$ at $r = 30$, and the production choice $r = 20$ sits within $0.2$ percentage points of the asymptotic in-window error level. The choice is bounded below by the accuracy that low rank sacrifices and above by the conditioning of the per-parameter regression. Each per-parameter fit estimates $1 + r + r(r+1)/2$ coefficients per reduced coordinate from the constant, the linear row, and the compressed quadratic row. This amounts to $231$ coefficients at $r = 20$ against the $481$ in-window snapshots available per parameter. The regression is therefore overdetermined at the production rank and first becomes underdetermined near $r = 30$. The quadratic block holds $210$ of those $231$ coefficients and is weakly constrained by the smooth in-window trajectory, whose effective rank lies well below the truncation rank. The Tikhonov regularization that the standard OpInf regression carries \citep{mcquarrie2021combustion, mcquarrie2023nonintrusive} resolves the resulting conditioning, and the rapid decay of the POD singular value spectrum keeps the truncation admissible. The modes beyond $r = 20$ carry little energy and the in-window error saturates, while the regime-side snapshot blocks nominally support $r \approx 60$ (sub side) and $r \approx 80$ (super side) at a relative singular value tolerance of $10^{-3}$. Raising the rank past the saturation only enlarges the regression and weakens the conditioning of the quadratic block without improving accuracy. The production choice $r = 20$ is therefore a deliberate tradeoff, and the cyclic-skew structural constraint maintains stability across the full range of $r$ we tested.

\paragraph{Regime midpoint}
The regime midpoint $\mu_c$ of \cref{sec:methods:regimebridge} is chosen by comparing candidate values. We evaluate $\mu_c \in \{0.4, 0.5, 0.6, 0.7, 0.8, 0.9\}$ under the same leave-one-out protocol. All six configurations preserve full long-time stability and yield in-window errors within one percentage point of each other, from $7.9\%$ at $\mu_c = 0.5$ to $8.68\%$, and the $\mu_c = 0.5$ value matches the production $7.91\%$ within the reported seed spread. A broad shallow minimum sits near $\mu_c = 0.5$. We adopt $\mu_c = 0.5$, the empirical minimum. Selecting $\mu_c$ on the same leave-one-out metric that is later reported carries a risk of selection bias. The sub-percentage-point spread of the plateau limits the effect. The value is moreover consistent with the physics, since the dynamical transition is diffuse and the candidate indicators of \cref{sec:problem:plume} bracket $N \in [0.35, 0.8]$. The midpoint $\mu_c = 0.5$ sits within this band at the onset of internal-wave oscillation.

\paragraph{Kernel bandwidth}
The kernel bandwidth $\ell$, the single hyperparameter of the operator-interpolation step, has a broad insensitive range. Reducing it tenfold from the adopted $0.05$ moves the leave-one-out error by under $0.03$ percentage points, and widening it past the training-node spacing degrades accuracy smoothly. Every tested value $\ell \in \{0.005, 0.01, 0.02, 0.05, 0.1, 0.2, 0.5, 1.0\}$ keeps all $17$ trajectories stable to $t = 60$. The sigmoid blend width behaves the same way: over $w_\alpha \in \{0.01, 0.02, 0.05, 0.1, 0.2, 0.5, 1.0\}$ the error stays within $0.1$ points of its value at the adopted width for all $w_\alpha \leq 0.1$, it degrades only at widths that approach the parameter range and wash out the regime distinction, and every tested width keeps all $17$ trajectories stable.

\paragraph{Randomization and initial condition}
The randomized truncated singular value decomposition that builds the basis introduces controlled stochasticity. Repeating the proposed pipeline across five random seeds gives a mean in-window error of $7.93\%$ with a standard deviation of $0.01\%$, with the forecast-block errors, the Morton--Taylor--Turner slope, and full long-time stability all within sampling uncertainty of the production-seed value. This seed-insensitivity is consistent with the well-separated leading $r = 20$ singular values and the Davis--Kahan basis-perturbation bound of \appref{sec:app_theory}. Boundedness is also a property of the dynamics that extends well beyond the data initial condition. Integrating from the data initial condition and from six perturbations of it, through amplification by factors of two, five, and ten and through three random directions at matched norm, leaves every one of the $119$ runs across all $17$ parameters bounded to $t = 60$. Most relax back toward the data amplitude, a behavior indicating that the data-scale state is locally attracting (\appref{sec:app_energy}).


\section{Cross-problem transfer to viscous Burgers}
\label{sec:burgers}

The plume application of \cref{sec:plume} establishes that \CROI is accurate, long-time stable, and physics-consistent on the two-dimensional Boussinesq problem with the Brunt--V\"ais\"al\"a frequency as the varied parameter. This section tests whether the same construction transfers to a qualitatively different parametric PDE, and which of the structural constraint choices we adopted on the plume are problem-specific versus problem-independent. We execute a $2 \times 4$ comparison matrix on the viscous Burgers problem of \cref{sec:problem:burgers}. The matrix crosses the on and off settings of the regime-bridging basis split of \cref{sec:methods:regimebridge} with the four fit methods. The four fit methods are plain least squares, Lyapunov-stable $\bm{A}$ only, cyclic-skew $\bm{H}$ only, and both constraints. The resulting eight cells are compared against the corresponding eight cells on the plume. \Figref{fig:burgers_reconstruction} shows the full-order Burgers field beside the reduced reconstruction in the two regimes the split bridges, a smooth diffusively spreading profile at high viscosity and a steepening shock front at low viscosity.

\begin{figure}[htbp]
\centering
\includegraphics[width=\linewidth]{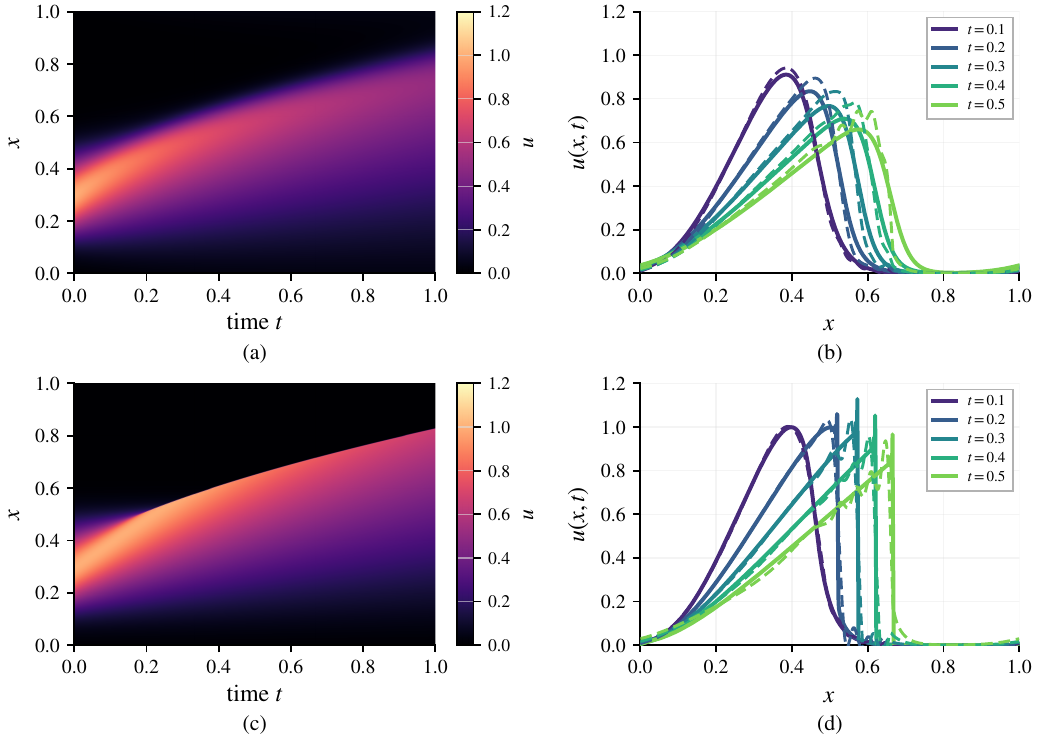}
\caption{Viscous Burgers reconstruction. (a) Full-order space--time field $u(x, t)$ for the smooth case ($\nu = 10^{-2}$). (b) Full-order (solid) and \CROI{} (dashed) profiles $u(x)$ at five in-window times for the smooth case, with color encoding time. (c) Full-order space--time field for the shock case ($\nu = 10^{-4}$). (d) Profiles as in (b) for the shock case.}
\label{fig:burgers_reconstruction}
\end{figure}

\subsection{The cross-problem comparison matrix}
\label{sec:burgers:matrix}

The Burgers training set is the six-parameter set of \cref{sec:problem:burgers}. The regime split is at $\nu_c = 2 \times 10^{-3}$, the smooth-to-shock midpoint identified in \cref{sec:problem:burgers} from the gradient-blowup characterization of the full-order data and consistent with the classical shock-formation-time estimate. The plume midpoint of \cref{sec:plume:robustness} is the empirical optimum of a study over candidate midpoints. This Burgers split is instead fixed at the data-characterized midpoint that the shock-formation-time estimate independently supports, and it is not separately tuned. The reduced dimension is $r = 10$ per regime side, the kernel blend is taken in $\log\nu$ space, and the overlap and sigmoid construction match the plume. Each per-parameter regression is Tikhonov-regularized and therefore well posed at this rank.

\Tabref{tab:transfer_matrix} reports the eight cells on each PDE. Each cell contains the mean in-window Frobenius error and the number of test trajectories that integrate to the full test horizon without numerical divergence. For the plume this horizon is $t_{\rm final} = 60$, five times its training window $t_{\rm train} = 12$, and for Burgers it is $t_{\rm final} = 2.5$, five times its training window $t_{\rm train} = 0.5$, the same horizon multiple on both problems. As on the plume, each Burgers cell is a leave-one-out average over the six training viscosities. The table shows that the stability hierarchy of \cref{sec:plume} transfers across problems. Across both PDEs, the cyclic-skew quadratic constraint alone keeps every trajectory bounded to the full test horizon of five training horizons, and it does so even though the inferred reduced linear operator is unstable at every parameter. On both problems the inferred linear operators are nowhere dissipative in the worst-case sense (\appref{sec:cert_cost_law}). The Lyapunov-stable linear constraint paired with it by default in the guaranteed-stable models of \citep{goyal2024guaranteed} does not improve trajectory completion in these tests and increases the in-window error, because long-time boundedness here follows from the closed cubic energy channel and requires no stabilized linear spectrum. The remaining cells clarify how the cyclic-skew constraint, the Lyapunov-stable linear constraint, and the regime-bridging split each affect trajectory completion and in-window error.

\begin{table}[t]
\centering
\caption{Cross-problem $2 \times 4$ comparison matrix of four fit methods (rows) across regime-bridging off or on and two PDEs (column groups for the plume and Burgers). Each cell reports in-window Frobenius error and trajectory count. RB denotes the regime-bridging split (\cref{sec:methods:regimebridge}). Cells are evaluated at a single representative basis seed; \tabref{tab:certified} in \appref{sec:app_certified} reports five-seed averages.}
\label{tab:transfer_matrix}
\small
\begin{tabular}{lcccc}
\toprule
& \multicolumn{2}{c}{Plume ($M = 17$, $\mu = N$)}
& \multicolumn{2}{c}{Burgers ($M = 6$, $\mu = \nu$)} \\
\cmidrule(lr){2-3}\cmidrule(lr){4-5}
Fit method & RB off & RB on & RB off & RB on \\
\midrule
Plain LS                            & NaN, $0/17$    & $733\%, 2/17$  & $602.7\%, 0/6$  & $107.0\%, 0/6$ \\
Lyapunov-$\bm{A}$ only          & $9.9\%, 1/17$  & $9.5\%, 0/17$  & $12.5\%, 6/6$   & $9.6\%, 1/6$ \\
\textbf{Cyclic-skew $\bm{H}$}   & \textbf{10.8\%, 17/17} & \textbf{7.9\%, 17/17} & \textbf{11.8\%, 6/6} & \textbf{7.7\%, 6/6} \\
Both (Lyapunov + skew)              & $21.8\%, 17/17$ & $16.8\%, 17/17$ & $12.3\%, 6/6$  & $7.8\%, 6/6$ \\
\bottomrule
\end{tabular}
\end{table}

The plain least-squares baseline fails in all four of its cells in the matrix, where no more than $2$ of $17$ plume trajectories and none of the $6$ Burgers trajectories stay bounded. Here \emph{plain LS} is the local operator interpolation of \cref{sec:methods:loi} with no structural constraint. Operators are inferred per parameter and blended exactly as in the proposed method, but the quadratic operator is left unconstrained. The divergence of the plain-LS fit is therefore distinct from the linear-spectrum instability of the \emph{affine joint} fit of \cref{sec:methods:diagnosis}. The per-parameter linear operators here are only mildly unstable, with maximum real part between $+0.04$ and $+0.10$. The unconstrained quadratic operator leaves the nonlinear energy injection channel open (\propref{prop:energy_neutral}). Fifteen of the $17$ trajectories diverge at times spanning $t = 4$ to $t = 44$ with median $8$, while the remaining two complete the horizon, consistent with the $2/17$ entry of \tabref{tab:transfer_matrix}. The linear part alone cannot produce finite-time divergence, and even with the measured eigenvector conditioning its amplification by $t = 4$ is bounded. The early onset therefore indicates that the open quadratic channel is injecting energy. The constraint of \propref{prop:energy_neutral} removes the finite-time blowup route this channel admits, and with it the divergence. The affine joint fit fails through the linear part that the per-parameter interpolation removes, while the unconstrained per-parameter fit fails through the quadratic part that the cyclic-skew constraint removes. These results show that constraining the quadratic operator is necessary for robust trajectory completion on both problems. An additional cell at a single representative basis seed separates the two roles cleanly. Imposing the cyclic-skew constraint on the plume's affine joint fit yields $17$ of $17$ trajectories bounded to $t = 60$, and the constraint reduces the joint fit's maximum real eigenvalue from the unconstrained $[+0.97, +1.78]$ to $[+0.38, +0.51]$. Yet its training-parameter (in-sample) error still averages $43\%$ and exceeds $90\%$ at the free rise nodes where the affine form is most compromised.

In every cell the cyclic-skew quadratic constraint alone achieves full stable trajectory completion and, among the configurations that complete, the lowest in-window error. The lower plume RB-off error of the Lyapunov-stable-only fit ($9.9\%$) comes from a model that integrates only $1$ of $17$ trajectories to $t = 60$, and thus it is not a competing operating point. The Lyapunov-stable linear operator constraint of \citep{goyal2024guaranteed} improves neither accuracy nor stability when used in combination with the cyclic-skew constraint. The combined constraint increases the in-window error relative to the cyclic-skew constraint alone on both problems, markedly on the plume ($16.8\%$ versus $7.9\%$) and marginally on Burgers ($7.8\%$ versus $7.7\%$). On these two problems, the Lyapunov-stable linear operator constraint is therefore redundant once the cyclic-skew quadratic constraint is in place and does not improve trajectory completion. It may still help where a residual instability remains, and whether the redundancy is universal remains open. The accuracy it costs is the price of enforcing strict dissipativity on an advection-dominated flow whose data-optimal operators are far from dissipative (\appref{sec:cert_cost_law}).

The Lyapunov-stable-only alternative is not robustly long-time stable in the matched five-horizon test. On Burgers, it completes all six trajectories with the single basis ($12.5\%$) but keeps only one of six bounded under the regime-split basis ($9.6\%$), and it fares no better on the plume. Only the cyclic-skew constraint completes every trajectory on both problems and both basis settings. That the Lyapunov-stable constraint at least survives the single-basis Burgers case reflects the underlying quadratic structure. The one-dimensional Burgers nonlinearity is a single quadratic in $u$, whereas the plume couples three fields $(u, w, b)$ through convective and stratification terms. That richer quadratic interaction is what the linear-only constraint does not control.

\Figref{fig:transfer_hierarchy} shows the principal result of the section, that the two ends of the fit-method ranking are conserved across the two problems. The cyclic-skew-only fit completes every trajectory on both problems at the lowest error among the completing configurations, and the plain least-squares fit never completes every trajectory, while the two intermediate constraints exchange order between the two problems. The conserved content of the hierarchy is therefore its endpoints. Both problems share an advective quadratic nonlinearity with an energy-preserving structure under the boundary conditions considered, and Burgers provides a one-field scalar analogue of the plume's coupled convective dynamics. The finding is scoped to parametric PDEs with an energy-preserving convective nonlinearity. Extending the test to a reaction-dominated quadratic, whose nonlinearity injects or removes energy, is a natural next step.

\begin{figure}[htbp]
\centering
\includegraphics[width=\linewidth]{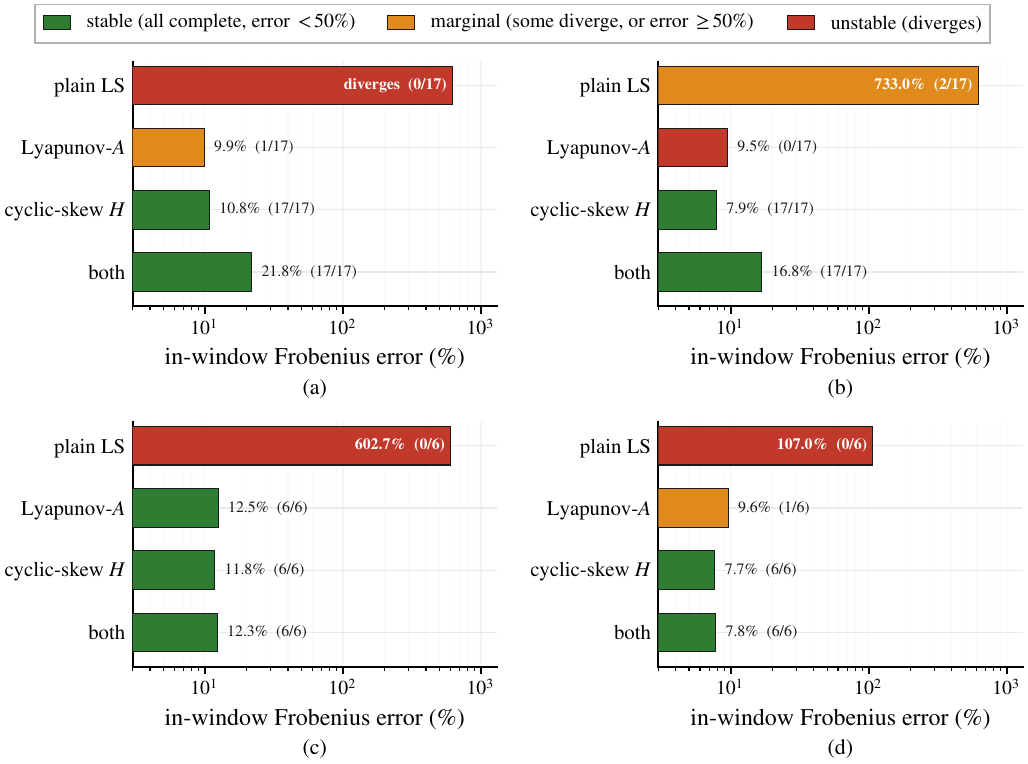}
\caption{In-window Frobenius error of the four fit methods (plain least squares, Lyapunov-stable linear part, cyclic-skew quadratic part, both) for (a) plume single basis, (b) plume regime-split basis, (c) Burgers single basis, (d) Burgers regime-split basis. Bars are annotated with the in-window Frobenius error and the number of trajectories reaching the test horizon (five training horizons on both problems), and are colored by grade: all trajectories complete with error below $50\%$; partial completion or error above $50\%$; none reach the horizon (or the fit diverges in window, where the error is omitted).}
\label{fig:transfer_hierarchy}
\end{figure}

\subsection{Sensitivity to the spectrum-aware Tikhonov exponent}
\label{sec:burgers:alpha}

The standard regularized regression \citep{mcquarrie2021combustion} admits a spectrum-aware Tikhonov variant (\cref{sec:background:fit}), in which the per-mode prior strength scales as the $k$-th singular value of the snapshot matrix raised to a power, analogous to the mode-dependent weights of Bayesian OpInf \citep{guo2022bayesian}. The production regression is the power-zero member of this family, the uniform weight. The variant is not part of the proposed methodology, and we vary the exponent only to confirm the configuration is robust to it and to document that no exponent is a tuned contribution. Across powers from zero to four the in-window error varies by about one percentage point on the plume ($7.8\%$ to $8.7\%$) and by about $0.6$ points on Burgers ($7.05\%$ to $7.66\%$), with no monotone trend. Once the cyclic-skew constraint is in place, the tail modes are already structurally constrained and the per-mode weighting adds no effective budget. The uniform weight remains the default, and no exponent contributes to the reported accuracy. We also ran a noise study on the plume with $\sigma_{\rm n} \in \{0, 1, 5, 10\}\%$, where $\sigma_{\rm n}$ denotes the relative amplitude of Gaussian white noise added to the snapshot columns as a percentage of the Frobenius norm of the noiseless snapshot matrix, and it confirms that the exponent gives no noise robustness either. For $\sigma_{\rm n} \geq 1\%$, the spectrum-weighted plain least-squares fit with no structural constraint diverges across the full range of exponents. The cyclic-skew model refit on noisy snapshots at $\sigma_{\rm n} \in \{1, 5\}\%$ separates the two failure channels once more. Every trajectory stays bounded at both noise levels, and the fit completes $17$ of $17$ to $t = 60$, even though five-percent noise inflates the inferred worst-case growth rate beyond $+1.0$. The accuracy of the bounded fit nonetheless degrades with noise, as the per-parameter in-sample error rises from $1.6\%$ noiseless to $4.5\%$ at one percent and $66\%$ at five percent. Recovering accuracy under noise would require a denoising front end, which boundedness alone does not provide.

The single-step least-squares residual of \cref{sec:methods:loi} belongs to the family of multi-step rollout residuals introduced there. The Burgers transition admits no analogue of the Morton--Taylor--Turner plume rise exponent to serve as an external scaling target for a rollout objective. We therefore did not execute the multi-step study on Burgers and make no claim about it here. The structural form of the per-parameter regression of \cref{sec:methods:loi} is identical across both problems. The same multi-step family is therefore available on Burgers for any application that introduces an integrated physics target as a secondary fit objective.

As on the plume, the structural constraints and the construction of \cref{sec:methods:composition} transfer unchanged to Burgers, with the problem-scale settings (reduced dimension, kernel coordinate, and regime midpoint) fixed once from the Burgers data without tuning against test error. The transferred configuration achieves $7.7\%$ in-window error with all six trajectories bounded to five training horizons. As on the plume, its forecast error grows beyond the training window. The model is trained on short-horizon dynamics and extrapolates the later decay of the viscous solution, while the trajectory stays bounded.


\section{Parameter space extrapolation}
\label{sec:extrap}

\cref{sec:plume} and \cref{sec:burgers} evaluate the reduced model primarily in leave-one-out settings, where the test parameters lie within the sampled parameter range and are bracketed by nearby training values except at the endpoints. This section addresses the harder question of how the model behaves when the test parameter falls outside the convex hull of the training set, or in a dynamical regime entirely held out from training. We report three scenarios on the plume. A boundary holdout removes the four extreme parameters from the training set. A super-to-sub regime extrapolation holds out the entire sub-regime. The reverse sub-to-super regime extrapolation holds out the entire super-regime.

\subsection{Three extrapolation scenarios}
\label{sec:extrap:scenarios}

Each scenario retrains the proposed framework on a strict subset of the $17$ training parameters and evaluates on the held-out parameters. The basis is recomputed on the reduced training set to avoid information leakage. Because the in-distribution study showed weak sensitivity to the seed of the randomized SVD, we report the extrapolation trends for one representative seed. \Tabref{tab:extrap} summarizes the three scenarios and the resulting in-window accuracy and stability statistics, where $n_{\rm stable}$ denotes the number of test trajectories that integrate to $t = 60$ without numerical divergence and $n_{\rm test}$ the total number of held-out test points.

\begin{table}[t]
\centering
\caption{Parameter space extrapolation scenarios on the plume, with the size- and distance-matched control pair B$'$/C$'$ and the in-distribution baseline; parameter-set sizes in parentheses. The $z_{\max}$ column is the mean in-window temporal correlation of the plume top trajectory $z_{\max}(t)$ with the full-order data.}
\label{tab:extrap}
\footnotesize
\setlength{\tabcolsep}{3.5pt}
\begin{tabular}{lccccc}
\toprule
Scenario & Training $N$ & Held-out $N$ & In-window error & $z_{\max}$ corr. & $n_{\rm stable}/n_{\rm test}$ \\
\midrule
A: Boundary holdout & $\{0.3, 0.4, \ldots, 1.5\}$ (13) & $\{0.1, 0.2, 1.6, 1.7\}$ (4) & $16.24\%$ & $0.89$ & $4/4$ \\
B: Super-to-sub & $\{0.6, 0.7, \ldots, 1.7\}$ (12) & $\{0.1, 0.2, \ldots, 0.5\}$ (5) & $51.77\%$ & $0.90$ & $5/5$ \\
C: Sub-to-super & $\{0.1, 0.2, \ldots, 0.5\}$ (5) & $\{0.6, 0.7, \ldots, 1.7\}$ (12) & $85.10\%$ & $0.91$ & $12/12$ \\
B$'$: matched control & $\{0.6, 0.7, \ldots, 1.0\}$ (5) & $\{0.1, 0.2, \ldots, 0.5\}$ (5) & $51.83\%$ & $0.96$ & $5/5$ \\
C$'$: matched control & $\{0.1, 0.2, \ldots, 0.5\}$ (5) & $\{0.6, 0.7, \ldots, 1.0\}$ (5) & $61.69\%$ & $0.93$ & $5/5$ \\
\midrule
In-distribution baseline & all 17 $N$ & (leave-one-out) & $7.91\%$ & $0.90$ & $17/17$ \\
\bottomrule
\end{tabular}
\end{table}

The cyclic-skew quadratic constraint provides stability under parameter space extrapolation even when the training set covers only a single regime, and all $21$ extrapolated test trajectories remain non-divergent at $t = 60$. By contrast, the unconstrained plain-LS fit diverges on $15$ of the $17$ test points even in-distribution (see \tabref{tab:transfer_matrix}). The bounded behavior under cross-regime extrapolation is therefore a property of the constraint beyond what a generic ROM provides.

\Figref{fig:k9_panel}(a) shows that in-window accuracy degrades smoothly and predictably with extrapolation distance, with no abrupt collapse at the training boundary. Within each boundary direction, the in-window error grows with the distance from the test parameter $\mu^\ast$ to the nearest training parameter. In Scenario A, the four boundary test points produce in-window errors of $23.9\%, 16.2\%, 9.0\%, 15.9\%$ at $N = 0.1, 0.2, 1.6, 1.7$ respectively. Their mean of $16.2\%$ is roughly twice the in-distribution mean of $7.91\%$. In Scenario B, the error grows from $24.5\%$ at the boundary $N = 0.5$ to $69.1\%$ at the most distant $N = 0.1$.

\begin{figure}[htbp]
\centering
\includegraphics[width=\linewidth]{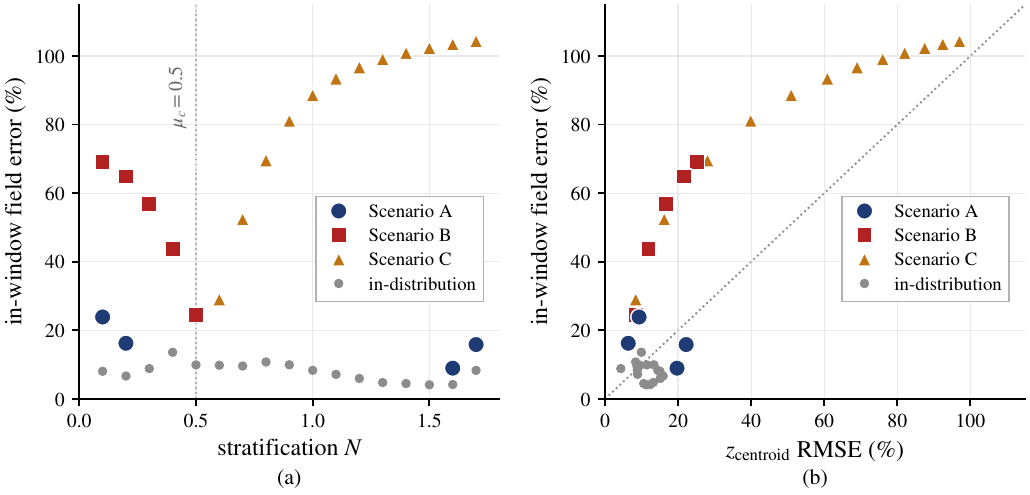}
\caption{(a) In-window field error versus $N$ for the three extrapolation scenarios A, B, and C and the in-distribution leave-one-out baseline, with the dotted vertical line at $\mu_c = 0.5$. (b) In-window field error versus the relative RMSE of the buoyancy-weighted vertical centroid $z_{\rm centroid}$ for the same points, with the $y = x$ diagonal.}
\label{fig:k9_panel}
\end{figure}

Scenarios B and C each hold out one entire regime, and a comparison of their mean in-window errors initially suggests a strong directional asymmetry. The sub-to-super direction in Scenario C has mean $85.1\%$ and appears worse than the super-to-sub direction in Scenario B at mean $51.8\%$. This raw comparison is confounded, because the two scenarios differ in training set size (twelve nodes versus five) and in test coverage. Scenario C extends $1.2$ beyond its training edge and B only $0.5$. A matched control with five training nodes on each side of the boundary and distance-matched test pairs removes most of the in-window gap. The controlled means are $51.8\%$ against $61.7\%$, so about ten points of the $33$-point raw gap survive the control, and all control trajectories remain non-divergent. The surviving gap grows with distance from the training edge, from under two percentage points at the nearest matched pair to seventeen at the most distant. Training set size and extrapolation distance therefore explain the bulk of the raw asymmetry, while the residual is a genuine directional effect concentrated at the far nodes, which the next subsection interprets.

\subsection{Physical interpretation of the residual asymmetry}
\label{sec:extrap:asymmetry}

Projecting each held-out in-window trajectory onto the rank-$20$ basis of the opposite training set captures comparable energy fractions, $80\%$ under B and $87\%$ under C for the raw energy and a nearly identical $76\%$ and $75\%$ for the fluctuation energy after removing each trajectory's time mean. Basis representation therefore introduces no directional bias, and the matched control above shows that training set size and extrapolation distance account for the bulk, though not all, of the in-window gap. The remaining question concerns the most distant extrapolations. Super-regime training histories contain a ballistic-rise phase alongside the oscillation, while sub-regime histories never contain a restoring phase. Scenario C's worst errors occur at the far nodes that dominate the raw gap, exactly where its operators must generate an oscillation they have never seen. The matched control preserves this ranking beyond the training window, with the sub-to-super direction worse at both forecast horizons ($86\%$ against $82\%$ at the first horizon, $106\%$ against $86\%$ at the deep horizon). These controls lead to a more specific interpretation. In-window cross-regime accuracy is governed primarily by training set size and extrapolation distance. The residual asymmetry that survives the control grows with extrapolation distance and points in one direction, and the failure mode is the generation of dynamics absent from the training history, in this case the trapped oscillation that sub-regime training data do not contain.

\subsection{Integrated dynamics are preserved while pointwise field fidelity degrades}
\label{sec:extrap:correlation}

The Frobenius error of \tabref{tab:extrap} is a pointwise measure that is sensitive to phase. In an advection- and oscillation-dominated unconfined flow it is governed by small displacements of coherent structures. This error therefore overstates the loss of the quantities that drive decisions, since the many-query tasks that motivate this class of surrogates depend on integrated rise quantities. The reduced model tracks these quantities far more robustly, and the temporal correlation of the integrated plume rise trajectories with the full-order data stays high under both cross-regime directions. The centroid-height trajectory $z_{\rm centroid}(t)$ is tracked with Pearson correlation $0.96$ under both extrapolation directions and the trajectory of the plume top $z_{\max}(t)$ with correlations of $0.90$ and $0.91$ (\tabref{tab:extrap}). As in \tabref{tab:plume_accuracy}, the correlation is reported on the plume top, while the error comparison of \figref{fig:k9_panel}(b) uses the smoother centroid. Across all five scenarios this correlation stays between $0.89$ and $0.96$ while the field error varies more than fivefold, which quantifies how much of the decision-relevant content survives extrapolation. These correlations stay close to the in-window values of \cref{sec:plume:accuracy} even as the field Frobenius error rises from $8\%$ in-window to $52\%$ and $85\%$ under extrapolation. \Figref{fig:k9_panel}(b) confirms this separation directly through the error magnitudes. For the cross-regime scenarios every test point lies above the $y = x$ diagonal, so the integrated-rise centroid RMSE stays below the pointwise field Frobenius error, by a wide margin at moderate errors and by a narrowing one at the far sub-to-super nodes. Only the two boundary points at $N = 1.6$ and $1.7$, where both errors are small, fall below the diagonal, with the centroid RMSE exceeding the field error by a factor of up to $2.2$ there. The low-order rise quantities therefore degrade more slowly than the full field over most of the tested range. The per-point spatial correlation also decreases as the Frobenius error increases, reinforcing the conclusion that the main degradation is in pointwise field fidelity. For uncertainty quantification and scenario analysis of plume rise, the integrated rise is the target quantity and the pointwise field fidelity is secondary.

Across all three scenarios, the construction provides \emph{stability under extrapolation}, with bounded trajectories on all $21$ held-out test cases and no a priori knowledge that the test parameter lies outside the training range or in an unseen regime. The in-window field accuracy degrades substantially, while the integrated rise dynamics remain well tracked (correlation near or above $0.90$, \cref{sec:extrap:correlation}). Outside the training range the model is best read as a qualitative full-field predictor and a reliable predictor of integrated rise. Within it, field reconstruction is markedly more accurate.


\section{Discussion}
\label{sec:discussion}
\label{sec:discussion:synthesis}

We synthesize the empirical results around the cyclic-skew energy-channel mechanism and use the data-weighted growth rate to interpret the remaining accuracy--stability tradeoffs. The cyclic-skew quadratic constraint is the structural feature most strongly associated with trajectory completion, and the proposed configuration avoids the abrupt rank-dependent divergence observed in the unconstrained plain-LS fit.

On both problems the cyclic-skew quadratic constraint yields full stable trajectory completion in every cell of the cross-problem comparison and achieves the lowest in-window error among the fully stable configurations in every cell. The Lyapunov-stable linear operator constraint of \citep{goyal2024guaranteed} is neither necessary nor beneficial under the cyclic-skew constraint, and on the plume it is also insufficient on its own, while on Burgers it yields stable trajectories on its own only under the single global basis. Among the constraints tested, the cyclic-skew constraint alone secures finite-horizon trajectory completion, because it closes the cubic finite-time blowup channel associated with the learned quadratic term (\propref{prop:energy_neutral}). Its \emph{accuracy} advantage over the Lyapunov-stable parameterization is a separate effect, traceable to the excess dissipation that this parameterization imposes. This interpretation is consistent with the energy-preserving formulation of \citep{gkimisis2025representation}. Once the cubic energy contribution of the inferred quadratic operator is constrained to vanish, the inferred linear part is free to take the values that the data require for fidelity, without that excess dissipation.

Varying the reduced dimension produces no abrupt divergence under the proposed configuration. A baseline at $\mu_c = 0.9$ with plain least-squares fitting exhibited an abrupt divergence at $r = 25$ where the reduced model diverged for every test point. At the adopted configuration with $\mu_c = 0.5$ and the cyclic-skew quadratic constraint, the reduced model is non-divergent across the full rank range $r \in \{10, 15, 20, 25, 30\}$, with the in-window error decreasing monotonically. The bases themselves support substantially higher ranks ($r \approx 60$ to $80$, \cref{sec:plume:robustness}) at a tolerance of $10^{-3}$ on the relative singular values. The abrupt divergence in that baseline therefore reflected the conditioning of the unconstrained least-squares fit, which does not recur under the cyclic-skew constraint at any tested rank, consistent with the reduction in free quadratic entries that the constraint imposes.

Both behaviors point to the data-weighted growth rate, which measures the inferred linear part through the energy channel. The proposed reduced operators are linearly unstable at every parameter. The worst-case growth rate $a_{\max} = \lambda_{\max}(\operatorname{sym}\hat{\bm A})$ is the largest eigenvalue of the Euclidean symmetric part of the inferred linear operator, and it is positive at every parameter, spanning $+0.11$ to $+0.46$ across the seventeen training fits (\figref{fig:energy_channel}(b)) and $[+0.20, +0.46]$ for the blended leave-one-out operators (\appref{sec:app_certified}). It bounds the linear-operator contribution to the instantaneous energy rate and is distinct from the spectral growth rate $\max\operatorname{Re}\lambda(\hat{\bm A})$ of \cref{sec:methods:diagnosis}, which governs asymptotic modal growth. \propref{prop:energy_neutral} therefore guarantees only freedom from finite-time blowup. The long-horizon boundedness we report ($17/17$ to five training horizons) is empirical. A proof remains open, and $a_{\max}$ does not explain the observation, since it is positive at every parameter. The data-weighted growth rate
\begin{equation}
a_{\rm data}
=\frac{\sum_k \hat{\bm q}_k^\top \operatorname{sym}_{\bm M}(\hat{\bm A})\,\hat{\bm q}_k}
      {\sum_k \hat{\bm q}_k^\top \bm M\,\hat{\bm q}_k},
\qquad
\operatorname{sym}_{\bm M}(\hat{\bm A})=\tfrac12(\bm M\hat{\bm A}+\hat{\bm A}^\top\bm M),
\label{eq:a_data}
\end{equation}
is the growth rate of the inferred linear part weighted by the directions the training snapshots occupy. Here $\hat{\bm q}_k$ are the reduced training snapshots and $\bm M$ is the symmetric positive-definite weight matrix of the energy metric in which the growth rate is evaluated. It is several times smaller than $a_{\max}$ and negative at twelve of the seventeen parameters, with two of the five positive values within $+0.01$ of zero (\appref{sec:cert_cost_law}). The exponential envelope permitted by the Gr\"onwall bound is therefore never approached, and the trajectories grow far more slowly than the worst case allows. This is consistent with the observed boundedness while falling short of a proof.

At the three weakly stratified free rise parameters ($N \le 0.3$), the data-weighted growth rate $a_{\rm data}$ is \emph{positive} at all three nodes, and the trajectory samples net-amplifying directions of the linear part with $a_{\max}$ up to $+0.28$. The reduced model still stays bounded to $t = 60$, where the spectrum alone would permit growth. Its boundedness at those parameters therefore cannot be attributed to effective linear dissipation, and it isolates the closed cubic channel as the operative mechanism that \figref{fig:energy_channel} shows at its clearest. At these parameters the buoyancy source drives the reduced system to a bounded quasi-steady state. The realized energy-growth rate $a_{\rm eff} = \dot E/(2E)$ along the integrated trajectory is therefore near zero at every parameter, within $[-0.003, +0.016]$ across the seventeen training values (\appref{sec:app_energy}). The cubic-channel closure secures that forced relaxation against finite-time blowup. The data-weighted growth rate is also computable from the fitted blended operator before any long integration, so it is available at a new deployment parameter before a variant is chosen. We do not claim it as a validated screen between bounded and unbounded deployments, because every cyclic-skew case in this study stayed bounded, including every node with $a_{\rm data} > 0$. The weakly stratified free rise nodes, where $a_{\rm data}$ is largest, are the cases where the stability-constrained variant of \appref{sec:app_certified} is most relevant when an a priori guarantee is required.

The same quantity grades the stability-accuracy tradeoff of the dissipativity constraint. The stability-constrained fit is inexpensive in the strongly stratified trapping regime, where $a_{\rm data}$ is small or negative, and grows expensive toward the weakly stratified free rise end where buoyancy drives the energy. \Figref{fig:certcost_law} shows that this dependence is smooth across the parameter range, with no threshold behavior. The correlation of the cost with $a_{\rm data}$ is $+0.90$, while the worst-case growth rate $a_{\max}$ remains the weaker predictor (\appref{sec:cert_cost_law}). \Figref{fig:tradeoff} isolates what the linear-part choice costs by holding the cyclic-skew $\hat{\bm H}$ fixed across all three constructions and varying only $\hat{\bm A}$, with each of the $17$ leave-one-out parameters resolved by $N$. The free model is the unconstrained-$\hat{\bm A}$ \CROI{} of \cref{sec:plume:accuracy}. It leaves its worst-case growth rate at the data-optimal value, positive at a mean of $+0.35$. The mean in-window error there is $7.91\%$. The Lyapunov-stable construction constrains the growth rate near zero at a mean error near $17\%$, while the stability-constrained construction enforces $\operatorname{sym}\hat{\bm{A}} \preceq -\eta\bm{I}$. The worst-case growth rate is therefore bounded above by $-\eta$, and the model carries a uniform-in-time absorbing-ball bound. Over $\eta \in \{0.05, 0.1, 0.2, 0.35, 0.5\}$ the mean in-window error rises from $37.3\%$ and approaches $49\%$ by $\eta = 0.5$, with the marginal cost of each further increment shrinking. A tighter decay rate therefore strengthens the guarantee at diminishing accuracy cost. The color shows the cost is strongly $N$-dependent. The lowest-error Lyapunov and stability-constrained parameters are the strongly stratified ones, where the inferred operator is already near-dissipative. The weakly stratified free rise parameters incur the largest accuracy penalty. This is the same ranking as in the stability-accuracy tradeoff of \figref{fig:certcost_law}. The free model therefore attains the reported in-window accuracy without a uniform-in-time bound, while the stability-constrained family supplies one at a graded, $N$-dependent cost. The per-metric accuracy and guaranteed bound at each margin are tabulated in \tabref{tab:certified} (\appref{sec:app_certified}). The data-weighted growth rate links the accuracy and the guarantee. It identifies where the free model's worst-case envelope is loose and where a uniform-in-time guarantee is inexpensive should one be required.

\begin{figure}[htbp]
\centering
\includegraphics[width=\linewidth]{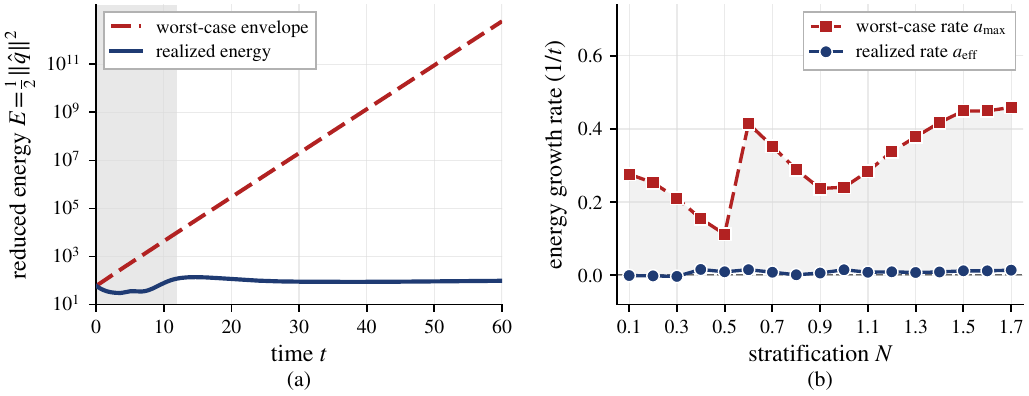}
\caption{(a) Reduced energy $E(t) = \tfrac{1}{2}\|\hat{\bm q}\|^2$ at $N = 0.3$ for the realized model and for the Gr\"onwall envelope permitted by the worst-case growth rate $a_{\max}$; the shaded band marks the training window $t \leq 12$. (b) Worst-case growth rate $a_{\max}$ and realized energy growth rate $a_{\rm eff}$ versus $N$, with the gray fill spanning the gap between them.}
\label{fig:energy_channel}
\end{figure}

\begin{figure}[htbp]
\centering
\includegraphics[width=0.72\linewidth]{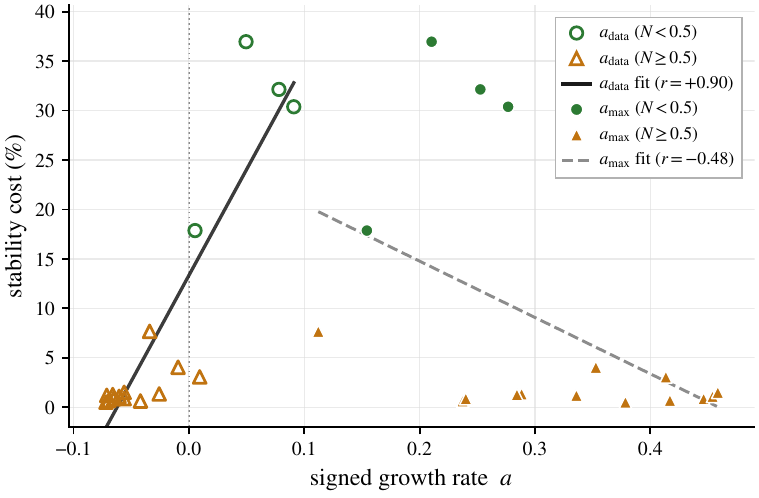}
\caption{Stability-constraint cost versus signed growth rate. The cost is the increase in mean in-window Frobenius error incurred by enforcing the dissipativity constraint. Each model appears twice, as an open marker at its data-weighted growth rate $a_{\rm data}$ (solid least-squares fit) and a filled marker at its worst-case growth rate $a_{\max}$ (dashed fit), sharing one vertical coordinate. Color encodes the regime side (green sub, orange super). Correlations use the Euclidean energy metric.}
\label{fig:certcost_law}
\end{figure}

\begin{figure}[htbp]
\centering
\includegraphics[width=0.74\linewidth]{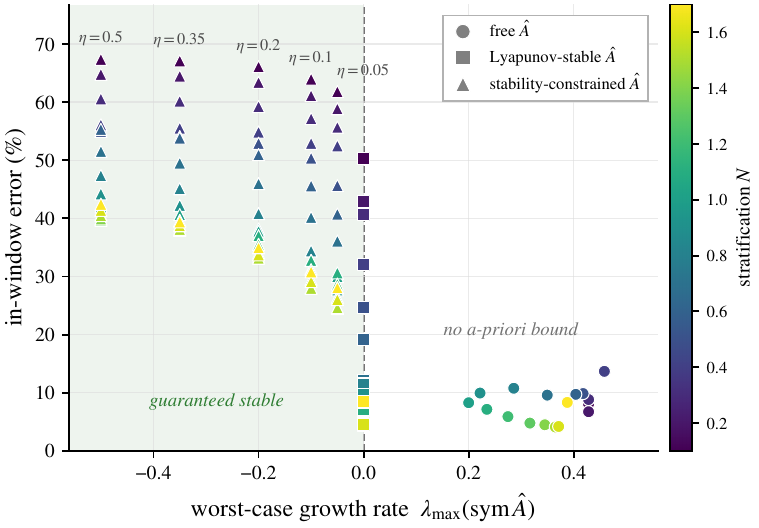}
\caption{In-window error versus the worst-case growth rate $a_{\max} = \lambda_{\max}(\operatorname{sym}\hat{\bm{A}})$. Each marker is one of the $17$ leave-one-out parameters. Color encodes $N$. Marker shape encodes the construction of the linear part $\hat{\bm{A}}$, free, Lyapunov-stable, or stability-constrained at margin $\eta \in \{0.05, 0.1, 0.2, 0.35, 0.5\}$. All three constructions carry the same cyclic-skew $\hat{\bm{H}}$ and differ only in $\hat{\bm{A}}$. The shaded half-plane $x < 0$ marks the region carrying the stability-constrained decay bound.}
\label{fig:tradeoff}
\end{figure}

Accuracy and stability are claims about different horizons, and separating them clarifies the remaining results. The reduced model is an accurate in-window surrogate (\cref{sec:plume:accuracy}). Beyond the window it is a stable bounded extrapolator whose forecast-accuracy degradation is characterized in \cref{sec:extrap}. This distinction also organizes the asymmetry of cross-regime extrapolation and the sample complexity scaling below.

The three extrapolation scenarios reported in \cref{sec:extrap} produce non-divergent trajectories in all $21$ test cases. The raw super-to-sub versus sub-to-super accuracy gap ($51.8\%$ against $85.1\%$ mean in-window error, a $64\%$ relative degradation) is largely a design effect of training set size and test coverage: a matched control shrinks it to about ten points, and the residual is a directional effect concentrated at the far sub-to-super nodes, with the basis representation itself nearly symmetric (\cref{sec:extrap:asymmetry}). The boundedness itself is the result, since the unconstrained plain-LS fit diverges on $15$ of the $17$ test points even in-distribution. Stability under cross-regime extrapolation is therefore a non-trivial property of the cyclic-skew quadratic constraint. Beneath that stability, in-window accuracy declines monotonically with distance from the training range.

The term regime-bridging refers to bridging the regime transition inside the training hull, through the regime-split basis and the per-parameter interpolation. The $7.91\%$ leave-one-out accuracy across the sub-to-super transition demonstrates it. It does not denote accurate field prediction of an entirely held-out regime, for which we report only bounded trajectories and a tracked integrated rise. The $51$ to $85\%$ in-window field error under whole-regime holdout is the accuracy cost of true parameter space extrapolation.

The sample complexity study of \cref{sec:plume:robustness} draws a clean contrast between the local interpolation and the affine form. The median in-window error grows smoothly from $8\%$ at the full $17$-parameter grid to $33\%$ at a three-parameter grid, and all trajectories stay bounded. Where the affine fit imposes a compromise across the regime transition that no amount of training data removes (the $\rho_{\rm aff}$ diagnostic of \cref{sec:methods:diagnosis}), the local interpolation keeps the per-parameter operator quality independent of the training set size. Its error growth is dominated by the kernel-interpolation gap, which shrinks with finer sampling, while the affine eigenvalue compromise diagnosed in \cref{sec:methods:diagnosis} plays no part. The reduced model therefore stays stable at a usable accuracy even at the sparsest grid. Across the extrapolation and sample complexity studies, the construction consistently prevents numerical divergence over the tested horizons, while accuracy degrades gradually with distance and data sparsity.


\section{Conclusion}
\label{sec:conclusion}

We developed cyclic-skew regime-bridged Operator Inference (\CROI), a non-intrusive parametric reduced-order modeling framework for PDEs that span a qualitative regime transition as a parameter varies. The framework infers reduced operators at each training parameter independently and combines them at any queried parameter through a Gaussian kernel in the parameter. All operators are expressed on a single reference reduced basis shared across the training set, and that basis is split at a physical indicator of the dynamical transition. The inferred quadratic operator is constrained to be cyclic-skew. The central result is that the cyclic-skew constraint is the structural prior most strongly associated with trajectory completion in both problems studied, when isolated from the Lyapunov-stable linear parameterization with which it is usually paired. On the buoyancy-source-forced plume it is the only single constraint under which every trajectory completes the five-horizon test, whereas the Lyapunov-stable parameterization alone fails. Because a joint multi-parameter least-squares fit across a regime transition displaces the inferred spectrum into the unstable half-plane and incurs a single-basis compromise the standard affine parametric form cannot accommodate, the per-parameter inference is needed to avoid the affine-spectrum instability, and the basis split further reduces the in-window error across the transition.

The cyclic-skew constraint succeeds because forcing the reduced nonlinearity to make no net contribution to the energy budget closes the finite-time blowup channel. The constraint survives the convex operator blend unchanged, and this yields boundedness at every evaluated parameter. When a uniform-in-time bound is required, a stability-constrained variant that enforces a dissipative linear part supplies such a guarantee, with an accuracy penalty that depends on how the trajectory samples the energy-growth directions of the reduced linear operator. The data-weighted growth rate provides a diagnostic for when enforcing the dissipative constraint is likely to carry a small accuracy penalty.

We evaluated the framework on two parametric problems. The principal application is a two-dimensional Boussinesq buoyant plume in a linearly stratified ambient, with the stratification varied across the subcritical-to-supercritical transition. \CROI{} stays bounded at all $17$ test parameters to five times the training horizon. In what is to our knowledge the first parametric OpInf study of an unconfined plume in a stratified ambient, the full-order model recovers the plume rise stratification scaling, and \CROI{} reproduces that scaling at the moderate forecast horizon and retains the correct negative trend at the deep horizon. The recovered plume rise slope changes from $-0.70$ at the moderate forecast horizon to $-0.40$ at the deep horizon as trajectory error accumulates, and it reproduces the internal-wave oscillation at roughly half the full-order frequency. Improving the long-horizon scaling slope and the oscillation frequency would likely require a regression that penalizes error along the integrated trajectory while remaining compatible with the cyclic-skew constraint. The cross-problem comparison uses the one-dimensional viscous Burgers equation across the smooth-to-shock transition, where the structural constraints transfer unchanged, with only problem-scale settings re-fixed from the Burgers data, and the transferred configuration reaches comparable in-window accuracy with full long-time stability. That same boundedness carries over to parameter space extrapolation. All $21$ boundary holdout and cross-regime trajectories remain bounded to five times the training horizon, with in-window accuracy degrading smoothly with extrapolation distance, and a size- and distance-matched control traces most of the raw directional gap to training set size and test coverage, leaving a residual asymmetry concentrated at the largest extrapolation distances. The hardest case is the prediction of trapped oscillations from sub-regime-only training data.




\phantomsection
\addcontentsline{toc}{section}{Declaration of generative AI and AI-assisted technologies in the manuscript preparation process}
\section*{Declaration of generative AI and AI-assisted technologies in the manuscript preparation process}

During the preparation of this work the authors used AI-assisted tools to edit the language and improve readability. After using these tools, the authors reviewed and edited the content as needed and take full responsibility for the content of the published article.

\phantomsection
\addcontentsline{toc}{section}{CRediT authorship contribution statement}
\section*{CRediT authorship contribution statement}

\textbf{Youngjun Kwon:} Conceptualization, Methodology, Software, Validation, Formal analysis, Investigation, Data curation, Visualization, Writing -- original draft. \textbf{Abdollah Shafieezadeh:} Conceptualization, Supervision, Resources, Funding acquisition, Writing -- review \& editing.

\phantomsection
\addcontentsline{toc}{section}{Declaration of competing interest}
\section*{Declaration of competing interest}

The authors declare that they have no known competing financial interests or personal relationships that could have appeared to influence the work reported in this paper.

\phantomsection
\addcontentsline{toc}{section}{Acknowledgments}
\section*{Acknowledgments}

This work was supported in part by the Lichtenstein endowment at The Ohio State University. Any opinions, findings, and conclusions expressed are those of the authors and do not necessarily reflect those of the sponsoring organization.


\appendix
\makeatletter
\@addtoreset{figure}{section}
\@addtoreset{table}{section}
\makeatother

\section{Full-order model details}
\label{sec:app_fom}

\subsection{Numerical scheme}
\label{sec:app_fom_numerics}

The 2D Boussinesq solver for Eq.~\eqref{eq:plume_pde} uses the following discretization:

\begin{itemize}[leftmargin=*]
\item \textbf{Staggered marker-and-cell (MAC) grid}: velocity components
$\bm{u} = (u, w)$ stored at the respective face centers. Pressure $p$
and buoyancy $b$ are stored at cell centers. The cell-centered buoyancy is
linearly interpolated to vertical faces for use in the buoyancy term
$b\,\hat{\bm{z}}$ of the governing equations.

\item \textbf{Pressure projection}: a Helmholtz--Hodge decomposition
enforces $\nabla \cdot \bm{u} = 0$ at every time step. The pressure
Poisson solve uses fast Fourier transforms along the periodic
horizontal direction $x$ and tridiagonal solves in the vertical
direction $z$.

\item \textbf{Time integration}: second-order Runge--Kutta (RK2) with
adaptive time step $\Delta t$ from the Courant--Friedrichs--Lewy (CFL)
condition with CFL $= 0.4$.

\item \textbf{Sponge layer}: an upper damping region
$z \in [L_z - d, L_z]$ absorbs upward-moving disturbances. The sponge
profile is a smooth Gaussian-like ramp with characteristic strength
$\sigma_s = 2.0$ and thickness $d = 0.15$.
\end{itemize}

\paragraph{Boundary conditions} All fields are periodic in the
horizontal direction $x \in [0, L_x]$. The velocity is no-slip at the bottom ($z = 0$) and free-slip with the sponge damping at the top ($z = L_z$). The pressure Poisson solve uses homogeneous Neumann conditions in $z$. The buoyancy uses zero-flux (Neumann) conditions at both walls, so that the only buoyancy input is the interior source of Eq.~\eqref{eq:source}.

\paragraph{Buoyancy source} A steady source is added to the buoyancy
equation at every time step,
\begin{equation}
S(x, z) \;=\; S_0 \, \exp\!\Big(-\frac{(x - L_x/2)^2}{2 \sigma_x^2}\Big)
\, \exp\!\Big(-\frac{z^2}{h_z^2}\Big),
\label{eq:source}
\end{equation}
a horizontal Gaussian of width $\sigma_x = 0.125$ ($\approx 12$ cells on the production grid) centered at mid-domain and falling off as a Gaussian in $z$ from the bottom boundary with vertical scale $h_z = 0.03$ ($\approx 3$ cells). The amplitude is $S_0 = 0.8$. The source is maintained for the full simulation and models a continuous buoyant release such as a vent or a fire, so that the plume settles into a sustained quasi-steady column.

\paragraph{Production parameters} The production domain, grid, horizon, snapshot count, and diffusivities are as given in \cref{sec:problem:plume}. With characteristic velocity and length of order unity, the Reynolds number is $\mathcal{O}(10^3)$. The simulation is two-dimensional and unsteady with no turbulence closure, consistent with the grid and CFL convergence reported below.

\subsection{Grid convergence}
\label{sec:app_grid_conv}

A six-grid convergence study at $N \in \{0.5, 1.5\}$ yields the following.

\begin{itemize}[leftmargin=*]
\item Plume rise diagnostics $z_{\max}$ and $b_{\max}$ are converged
at $128 \times 96$, with $< 1\%$ change between $128 \times 96$ and
$192 \times 144$ resolutions.

\item Kinetic energy KE${}_{\max}$ shows a $\sim 5\%$ drift between
$160 \times 120$ and $192 \times 144$ before stabilizing at finer
grids. We accept this as the FOM uncertainty budget. The ROM is
evaluated against the production-grid FOM. Therefore this drift does not
affect ROM error reporting.

\item A spot check at the strongly stratified $N = 1.5$ shows the same
convergence pattern as $N = 0.5$, so no separate refinement is
required for the supercritical regime.
\end{itemize}

The production grid $240 \times 144$ uses $\Delta x = \Delta z \approx 0.0104$ to keep cells square. This extends the convergence-validated $192 \times 144$ ($L_x = 2.0$) configuration laterally at the same grid spacing to accommodate the lateral intrusion at strong $N$, which pushes the plume cap further outward in $x$.

\subsection{Temporal convergence}
\label{sec:app_temporal_conv}

The time integrator is a second-order Runge--Kutta scheme paired with the Helmholtz--Hodge pressure projection of the numerical scheme above, a standard combination for incompressible Boussinesq flow. The per-step pressure projection and the second-order spatial stencil set the accuracy floor. A higher-order integrator therefore offers little benefit, because the measured grid and CFL convergence budget below is already dominated by the spatial discretization. The adaptive step is set by the Courant number, and the bulk quantity most sensitive to it is the peak kinetic energy. Refining the step (CFL${} \in \{0.2, 0.3, 0.4, 0.5\}$ at $N = 0.5$) moves this peak kinetic energy by $2.9\%$ between the production CFL${} = 0.4$ and the finest CFL${} = 0.2$, and the trapping-height centroid by $0.7\%$. Both shifts lie within the $\sim 5\%$ full-order uncertainty budget of the grid study and well below the $7.9\%$ in-window reduced-model error. The temporal discretization is therefore not a limiting factor in the operator regression, and this numerical uncertainty is too small to confound the reduced-model assessment. We repeated the check at the strongly stratified case $N = 1.5$ and reached the same conclusion. At this $N$, the post-transient internal-wave oscillation makes the trajectory most temporally sensitive. Peak kinetic energy changes by $3.2\%$ and the trapping-height centroid by $0.7\%$ between CFL${} = 0.4$ and CFL${} = 0.2$. The dominant oscillation frequency of the trajectory of the plume top is unchanged across the four step sizes. The most temporally sensitive quantity is therefore resolved already at the production setting.

\subsection{Sponge layer transparency}
\label{sec:app_sponge}

The sponge layer is included to absorb upward-propagating internal-wave radiation at the top of the computational domain so that the truncated finite vertical extent does not produce spurious wave reflections. To ensure that the sponge does not contaminate the bulk plume dynamics the reduced model is intended to learn, we ran a two-axis sensitivity study over the sponge thickness $d$ and the sponge strength $\sigma_s$ at $N = 0.5$, the case in which the plume reaches the highest vertical extent and is therefore the most susceptible to sponge interaction.

We varied the thickness over the five values $d \in \{0.05, 0.10, 0.15, 0.20, 0.25\}$ in steps of $0.05$ with $\sigma_s$ held at its default, and the strength over the doubling sequence $\sigma_s \in \{0.5, 1.0, 2.0, 4.0, 8.0\}$ with $d$ held at its default. The bulk QoIs $z_{\max}(t)$, $z_{\rm centroid}(t)$, the kinetic energy $\mathrm{KE}(t)$, the buoyancy maximum $b_{\max}(t)$, and the plume cross-sectional area all change by less than $0.5\%$ relative to the default. The default $(\sigma_s, d) = (2.0, 0.15)$ therefore lies on a wide plateau of insensitivity, and this plateau justifies the choice. Outside this plateau, for thin and weak sponges ($d \lesssim 0.03$ and $\sigma_s \lesssim 0.2$), we observe wave reflection that contaminates the $z_{\rm centroid}$ trajectory, but the default lies well inside the insensitive plateau and away from this regime. The same conclusion holds at the strongly stratified end ($N = 1.5$), where the plume is trapped low and the sponge interaction is even smaller. A spot check there shows no change in the bulk QoIs across the same parameter ranges.

These checks verify the numerics, and no comparison against laboratory plume data is attempted, since the FOM serves as the reference solution for every reduced-model error reported here. An energy-budget check provides production-grade verification that internal-wave reflection is negligible inside the training window $t \leq 12$. The kinetic energy in the upper half of the domain, above the trapping height of each $N$, shows no spurious revivals of the kind a reflected wave packet would produce. The default sponge configuration is used throughout the empirical sections without further modification.

\subsection{State-vector assembly for the OpInf regression}
\label{sec:app_lifting}

The full-order state vector $\bm{q}$ supplied to the POD and the OpInf regression is assembled from the cell-centered buoyancy $b$ and the staggered velocity components $u$, $w$ in three steps. First, each field is interpolated to cell centers if not already there. This yields three $N_x \times N_z$ arrays per snapshot, where $N_x \times N_z = 240 \times 144$ is the production grid of \cref{sec:problem}. Second, the three arrays are flattened (row-major) and concatenated into a single vector $\bm{q} \in \mathbb{R}^{3 N_x N_z}$. Third, the resulting per-parameter snapshot matrix $\bm{Q}^{(i)} \in \mathbb{R}^{3 N_x N_z \times K}$ is centered and scaled per variable using the standard shift--scale preprocessing \citep{mcquarrie2021combustion}, and the concatenation of these matrices over the training parameters furnishes the input to the POD truncation and the OpInf regression of \cref{sec:background:fit}.

\subsection{Computational cost}
\label{sec:app_compute}
\label{sec:app_cost}

All timings were measured on a single workstation with two Intel Xeon Gold 6248R processors ($24$ cores used via thread pinning), $256$\,GB of memory, and no graphics processor. They are per-process CPU wall-clock times and are expected to be reproducible to within a small constant factor on comparable $24$-core hardware. \Tabref{tab:cost} summarizes the representative costs.

\begin{table}[htbp]
\centering
\caption{Representative wall-clock costs for the production setup.}
\label{tab:cost}
\small
\begin{tabular}{lr}
\toprule
Operation & Wall-clock \\
\midrule
Full-order plume simulation (per parameter, $240\times144$, $t{=}20$) & $\sim 120$\,s \\
Reduced-operator fit (cyclic-skew, per parameter) & $\sim 1$\,s \\
Reduced-model integration to $t = 60$ (adaptive Runge--Kutta) & $\sim 0.3$\,s \\
\bottomrule
\end{tabular}
\end{table}

The reduced model integrates about $400\times$ faster than the corresponding full-order simulation, taking the $\sim 120$\,s full-order run to $t = 20$ as the numerator and the $\sim 0.3$\,s reduced-model integration to $t = 60$ as the denominator. At a matched horizon, the reduced-model speedup is larger, because the reduced cost in the denominator covers a three-times-longer integration than the full-order cost in the numerator. The offline training cost of the cyclic-skew configuration is several minutes over the training set and is therefore repaid after a small number of reduced-model queries. Once the offline fit is amortized, the end-to-end speedup approaches the per-query factor of about $400\times$. Fitting the structure-preserving operators adds negligible cost relative to a single full-order simulation, and we make no claim about relative fitting speed among the structural variants.


\section{Spectral diagnosis and energy boundedness}
\label{sec:app_bauer_fike}
\label{sec:app_theory}

This appendix supports two claims made in the main text. The first claim of \cref{sec:methods:diagnosis} is that the affine joint fit displaces the inferred linear spectrum toward instability on a problem with a regime transition, and that no regularization strength repairs it. The second claim of \cref{sec:methods:cyclicskew} and \cref{sec:burgers} is that the cyclic-skew constraint on the reduced quadratic operator is the structural prior that secures stability. It rules out finite-time blowup of the reduced trajectory by closing the nonlinear energy injection channel, and the convex operator blend of \cref{sec:methods:loi} preserves that property at every evaluated parameter.

\subsection{Setup}

Fix a reduced dimension $r$ and a POD basis $\bm{\Phi} \in \mathbb{R}^{n \times r}$ common to all training parameters $\mu \in \mathcal{M}_{\rm train} = \{\mu_1, \dots, \mu_M\}$. Let $\bm{X}^{(i)} = \bm{\Phi}^\top (\bm{Q}^{(i)} - \bar{\bm{q}}\bm{1}_K^\top) \in \mathbb{R}^{r \times K}$ be the centered reduced snapshot matrix at $\mu_i$, as in \cref{sec:background:fit}, with $K$ the number of snapshots per training parameter, and $\bm{Y}^{(i)} = \bm{\Phi}^\top \dot{\bm{Q}}^{(i)}$ the corresponding projected time-derivative matrix. Adopt the affine form $\Ab(\mu) = \Ab_0 + g(\mu)\,\Ab_1$ ($g(\mu) = \mu^2$ for the plume, $g(\mu) = \mu$ for Burgers), and set $g_i := g(\mu_i)$. The single-parameter fit at $\mu_i$ is the $\Ab$-block of the regularized regression of Eq.~\eqref{eq:opinf_regression} restricted to the snapshots of $\mu_i$. In the unregularized linear-only limit it reduces to $\Ab^{\rm single}(\mu_i) = \bm{Y}^{(i)} (\bm{X}^{(i)})^{+}$, the form used in the analysis below. The multi-parameter affine fit $(\Ab_0^{\rm joint}, \Ab_1^{\rm joint})$ minimizes $\sum_{i} \| (\Ab_0 + g_i \Ab_1)\bm{X}^{(i)} - \bm{Y}^{(i)} \|_F^2$, with prediction $\Ab^{\rm joint}(\mu_i) = \Ab_0^{\rm joint} + g_i \Ab_1^{\rm joint}$. The per-parameter perturbation that the joint constraint incurs is
\begin{equation}
\bm{E}_i \;\equiv\; \Ab^{\rm joint}(\mu_i) - \Ab^{\rm single}(\mu_i).
\label{eq:app_perturbation}
\end{equation}

\subsection{Spectral displacement under the affine joint fit}

\paragraph{Directional bound} The single-parameter fit minimizes
the residual at $\mu_i$ exactly. Its spectrum is therefore the reference, and the joint affine constraint can only perturb it. By the Bauer--Fike theorem \citep{bauer1960norms}, every eigenvalue $\lambda$ of $\Ab^{\rm joint}(\mu_i)$ lies within $\kappa_2(\bm{V}^{(i)})\,\|\bm{E}_i\|_2$ of an eigenvalue of $\Ab^{\rm single}(\mu_i)$, where $\bm{V}^{(i)}$ is the eigenvector matrix of $\Ab^{\rm single}(\mu_i)$ and $\kappa_2$ the spectral condition number. The bound is sign-agnostic, constraining only the magnitude of the displacement of the maximum real part while leaving its direction relative to the stable half-plane undetermined. Empirically, the displacement is toward instability. The eigenvector conditioning is in fact mild on the plume, with $\kappa_2(\bm{V}^{(i)}) \in [8, 72]$ across the seventeen training values and no monotone trend in $\mu$. The drift is therefore carried by the perturbation magnitude $\|\bm{E}_i\|_2$ (of order $7$ to $23$ in spectral norm) and is uncoupled from the conditioning. Evaluated at these values the bound exceeds the observed displacement of the maximum real part by one to three orders of magnitude, so its role is qualitative. It identifies $\|\bm{E}_i\|_2$ as the driver and caps the possible drift, and the quantitative content of the diagnosis rests on the measured residual $\rho_{\rm aff}$ and the measured spectra.

\paragraph{Affine-inadequacy measure} The perturbation of Eq.~\eqref{eq:app_perturbation} is large precisely when the
per-parameter operators do not lie on an affine $g(\mu)$-family. We quantify this directly. Let $\{\Ab^{\rm single}(\mu_i)\}$ be the per-parameter operators and let $(\Ab_0^\star, \Ab_1^\star)$ be their best affine fit. The relative residual
\begin{equation}
\rho_{\rm aff}
\;=\;
\frac{\bigl(\sum_i \| \Ab^{\rm single}(\mu_i)
       - (\Ab_0^\star + g_i \Ab_1^\star) \|_F^2 \bigr)^{1/2}}
     {\bigl(\sum_i \| \Ab^{\rm single}(\mu_i) \|_F^2 \bigr)^{1/2}}
\label{eq:rho_affine}
\end{equation}
is a computable a priori indicator of how far the operator family departs from any affine form. On a representative design spanning the full parameter range, $\rho_{\rm aff} \approx 0.26$. About a quarter of the total operator magnitude therefore lies outside the best affine fit in $g(\mu)$. The joint fit must compromise on this variation. Restricting the fit to one regime side reduces the departure unevenly. On the sub-regime training values the residual falls to $\rho_{\rm aff} \approx 0.10$, while on the super-regime values it falls only to $\rho_{\rm aff} \approx 0.21$, so an appreciable affine inadequacy persists within the super-regime itself. The restricted fits also involve fewer operators, which eases the regression irrespective of any regime structure. This is the diagnostic counterpart of two constructions. The regime split of \cref{sec:methods:regimebridge} confines each affine fit to a range where $\rho_{\rm aff}$ is smaller. The local operator interpolation of \cref{sec:methods:loi} removes the affine constraint altogether.

\subsection{Empirical verification on the production plume problem}
\label{sec:app_bauer_fike_empirical}

At the production setup $M = 17$, with matched Tikhonov regularization ($\beta = 10^{-3}$), the affine joint fit of \citep{mcquarrie2023nonintrusive} produces an operator that is unstable at every training node. The spectral growth rate $\max\operatorname{Re}(\lambda(\hat{\bm{A}}(N)))$ ranges from $+1.78$ at the smallest $N$ to $+0.97$ at the largest. The single-parameter least-squares fit at the same nodes and with the same regularization peaks between $+0.04$ and $+0.10$. This is smaller by a factor of $18$ to $38$ (\figref{fig:bauer_fike}). The instability is a property of the joint affine constraint, while the per-parameter data itself is well-conditioned. The per-parameter construction of \cref{sec:methods:loi} avoids it by fitting each operator independently.

The joint-fit instability cannot be regularized away. Varying the Tikhonov weight across the twelve decades $\beta \in [10^{-6}, 10^{6}]$ monotonically shrinks the operator norm ($\|\hat{\bm{A}}_0\|_2$ from $64$ down to $5\times10^{-8}$) but leaves $\max\operatorname{Re}(\lambda)$ strictly positive at every $\beta$. It falls from $+14$ at $\beta = 10^{-6}$ down to $+3\times10^{-4}$ at $\beta = 10^{4}$ and approaches zero only as the operator itself vanishes, but it never crosses into the stable half-plane. No regularization strength therefore makes the affine joint fit stable while leaving a usable operator, and the fit-time treatment of the parameter dependence must change.

\subsection{Energy boundedness of the cyclic-skew reduced model}
\label{sec:app_energy}

The reduced model integrated in this paper is the quadratic system $\dot{\hat{\bm{q}}} = \hat{\bm{c}} + \hat{\bm{A}}\hat{\bm{q}} + \hat{\bm{H}}[\hat{\bm{q}} \otimes \hat{\bm{q}}]$, with $(\hat{\bm{H}}[\hat{\bm{q}}\otimes\hat{\bm{q}}])_i = \sum_{j,k} \hat F_{ijk}\, \hat{q}_j \hat{q}_k$ for the third-order tensor $\hat F$ of the reduced quadratic operator (\cref{sec:background:structure}). The cyclic-skew constraint of \cref{sec:methods:cyclicskew} is the linear condition
\begin{equation}
\hat F_{ijk} + \hat F_{jki} + \hat F_{kij} = 0
\qquad \text{for all } i, j, k.
\label{eq:cyclic_skew}
\end{equation}

\begin{proposition}[The constraint closes the nonlinear energy channel]
\label{prop:energy_neutral}
Let $E(\hat{\bm{q}}) = \tfrac{1}{2}\|\hat{\bm{q}}\|_2^2$. If $\hat F$ satisfies
Eq.~\eqref{eq:cyclic_skew}, then
$\hat{\bm{q}}^\top \hat{\bm{H}}[\hat{\bm{q}}\otimes\hat{\bm{q}}] = 0$ for all
$\hat{\bm{q}}$. It follows that
\begin{equation}
\frac{\mathrm{d}E}{\mathrm{d}t}
= \hat{\bm{q}}^\top \hat{\bm{c}} + \hat{\bm{q}}^\top \hat{\bm{A}}\hat{\bm{q}}
\;\le\; \|\hat{\bm{c}}\|_2 \sqrt{2E}
        + 2\max(0,a_{\max})\,E,
\qquad
a_{\max} = \lambda_{\max}\!\bigl(\tfrac{1}{2}(\hat{\bm{A}}+\hat{\bm{A}}^\top)\bigr).
\label{eq:energy_rate}
\end{equation}
Consequently $E(t)$ is finite for every finite $t$ and every $\hat{\bm{A}}$, which prevents the reduced trajectory from diverging in finite time.
\end{proposition}

Here $a_{\max} = \lambda_{\max}(\operatorname{sym}\hat{\bm{A}})$ is the worst-case growth rate in the Euclidean metric and is used in \cref{sec:discussion}. The bound $\hat{\bm{q}}^\top \hat{\bm{A}}\hat{\bm{q}} \le a_{\max}\|\hat{\bm{q}}\|^2$ holds for any $\hat{\bm{A}}$, and the $\max(0,\cdot)$ truncation in Eq.~\eqref{eq:energy_rate} keeps the envelope growth rate nonnegative when $a_{\max}$ is itself negative. Dropping the constraint reinstates the cubic term $\hat{\bm{q}}^\top \hat{\bm{H}}[\hat{\bm{q}}\otimes\hat{\bm{q}}]$, of order $\|\hat{\bm{q}}\|^3$, which admits finite-time blowup.

\begin{proof}
Write $\hat{\bm{q}}^\top \hat{\bm{H}}[\hat{\bm{q}}\otimes\hat{\bm{q}}]
= \sum_{i,j,k} \hat F_{ijk}\, \hat{q}_i \hat{q}_j \hat{q}_k$. Relabeling the summation
index triples by the two cyclic rotations
$(i,j,k)\!\to\!(j,k,i)\!\to\!(k,i,j)$ leaves the monomial
$\hat{q}_i \hat{q}_j \hat{q}_k$ unchanged, giving
$3 \sum_{i,j,k} \hat F_{ijk} \hat{q}_i \hat{q}_j \hat{q}_k
= \sum_{i,j,k} (\hat F_{ijk}+\hat F_{jki}+\hat F_{kij})\, \hat{q}_i \hat{q}_j \hat{q}_k = 0$
by Eq.~\eqref{eq:cyclic_skew}. The cubic form therefore vanishes
identically. Then
$\mathrm{d}E/\mathrm{d}t = \hat{\bm{q}}^\top \dot{\hat{\bm{q}}}
= \hat{\bm{q}}^\top\hat{\bm{c}} + \hat{\bm{q}}^\top\hat{\bm{A}}\hat{\bm{q}}$,
since the quadratic contribution is gone. The estimate follows from
$\hat{\bm{q}}^\top\hat{\bm{A}}\hat{\bm{q}}
= \hat{\bm{q}}^\top \tfrac{1}{2}(\hat{\bm{A}}+\hat{\bm{A}}^\top)\hat{\bm{q}}
\le \max(0,a_{\max})\|\hat{\bm{q}}\|^2 = 2\max(0,a_{\max})E$ and
$\hat{\bm{q}}^\top\hat{\bm{c}} \le \|\hat{\bm{c}}\|\,\|\hat{\bm{q}}\|
= \|\hat{\bm{c}}\|\sqrt{2E}$. A comparison estimate then bounds
$E(t)$ by the solution of the corresponding scalar differential
inequality, which is linear in $u = \|\hat{\bm{q}}\| = \sqrt{2E}$ and is
finite on any bounded interval. The unconstrained
case is the standard cubic blowup of $\dot{E} \sim E^{3/2}$.
\end{proof}

\propref{prop:energy_neutral} is the discrete-energy boundedness mechanism known for energy-preserving quadratic Galerkin systems \citep{schlegel2015long}. A linearly unstable $\hat{\bm{A}}$ ($a_{\max} > 0$) produces at most exponential growth. This rules out the finite-time divergence that an energy-injecting nonlinearity admits. Over a finite horizon, closing the cubic channel removes finite-time blowup. It does not imply global boundedness when $\hat{\bm{A}}$ is linearly unstable, which would require the dissipative constraint of the stability-constrained variant of \appref{sec:app_certified}. The same caveat appears in a prior boundedness analysis \citep{liao2025boundedness}: an energy-preserving quadratic nonlinearity does not by itself confer global boundedness, and in more than two dimensions boundedness is not characterized by the shifted linear part at all, since that analysis constructs a bounded three-dimensional system whose shifted linear part is nowhere dissipative. The long-time stability we report is therefore the finite-horizon, empirically non-divergent behavior over five training horizons, and it stops short of a global bound.

Integrating the seventeen blended reduced models of the proposed framework at the evaluation parameters (\tabref{tab:certified}) confirms this quantitatively. Every trajectory stays bounded out to $t = 60$, five times the training horizon, with the discrete energy held inside the envelope of Eq.~\eqref{eq:energy_rate}. Along these integrated trajectories the cubic energy form $\hat{\bm{q}}^\top \hat{\bm{H}}[\hat{\bm{q}}\otimes\hat{\bm{q}}]$ vanishes to machine precision. The closure of the cubic channel therefore extends from the training snapshots to the full integrated trajectory. The exponential growth the energy envelope permits is never approached, as the following estimate quantifies. At the fitted growth rates and forcing norms the rate estimate of Eq.~\eqref{eq:energy_rate} already allows the reduced amplitude to double within $1.5$ to $6.2$ time units depending on the parameter, far inside the $t = 60$ horizon. The realized energy-growth rate stays far below the envelope's exponential rate at every parameter (\cref{sec:discussion}), so the envelope is never tight. The proposition therefore excludes finite-time blowup, while the observed long-horizon boundedness is an empirical saturation the bound does not by itself explain. The bounded behavior is also robust to the initial condition. We integrated the free (non-stability-constrained) cyclic-skew model from the data initial condition and from its amplifications by factors of two, five, and ten. We also integrated it from three random directions at matched norm. None of these $119$ runs across the seventeen parameters diverged by $t = 60$. Most runs relax back toward the data amplitude, with a median terminal amplitude of $1.4$ times the data scale, while the largest amplifications at weak stratification settle on bounded states above it. Therefore the data-scale state is locally but not globally attracting, consistent with the finite stability domains characterized for quadratic reduced models \citep{kramer2021stability}. The worst-case growth rate $a_{\max}$, the realized energy-growth rates, and the data-weighted growth rate $a_{\rm data}$ of \appref{sec:cert_cost_law} are reported and discussed in \cref{sec:discussion}.

The bound becomes operative at every interpolated parameter because the convex operator blend preserves it, as \propref{prop:blend_preserves} shows.

\begin{proposition}[The convex blend preserves the property at every $\mu^\ast$]
\label{prop:blend_preserves}
Let $\{\hat{\bm{H}}_i\}$ each satisfy Eq.~\eqref{eq:cyclic_skew} and
let $\{w_i(\mu^\ast)\}$ be the Gaussian-kernel weights of
\cref{sec:methods:loi}, which are non-negative and sum to one. Then
the blended operator
$\hat{\bm{H}}(\mu^\ast) = \sum_i w_i(\mu^\ast)\,\hat{\bm{H}}_i$
also satisfies Eq.~\eqref{eq:cyclic_skew}, and
$\hat{\bm{q}}^\top \hat{\bm{H}}(\mu^\ast)[\hat{\bm{q}}\otimes\hat{\bm{q}}] = 0$ for
all $\hat{\bm{q}}$. The boundedness of \propref{prop:energy_neutral} therefore
holds at every evaluation parameter $\mu^\ast$, with the linear part
$\hat{\bm{A}}(\mu^\ast) = \sum_i w_i(\mu^\ast)\,\hat{\bm{A}}_i$.
\end{proposition}

\begin{proof}
Condition \eqref{eq:cyclic_skew} is linear in the tensor entries.
Any linear combination of solutions is therefore a solution, and the
convex blend is one such combination. Equivalently,
$\hat{\bm{q}}^\top \hat{\bm{H}}(\mu^\ast)[\hat{\bm{q}}\otimes\hat{\bm{q}}]
= \sum_i w_i\,\hat{\bm{q}}^\top \hat{\bm{H}}_i[\hat{\bm{q}}\otimes\hat{\bm{q}}]
= \sum_i w_i \cdot 0 = 0$ by \propref{prop:energy_neutral} applied to
each node.
\end{proof}

The constraint set of Eq.~\eqref{eq:cyclic_skew} is linear and convex. The fit is therefore a well-posed projection onto the constraint subspace.

\paragraph{Empirical confirmation} Both propositions are borne out to machine
precision on the fitted operators by the normalized cubic energy form $\mathcal{E}_{\mathrm{cub}} = |\hat{\bm{q}}^\top \hat{\bm{H}}[\hat{\bm{q}}\otimes\hat{\bm{q}}]| / \|\hat{\bm{H}}\|_F$, the rate at which the reduced nonlinearity injects energy, normalized by the operator magnitude. Maximized over random unit vectors and over the reduced trajectory, it is $\approx 7 \times 10^{-17}$ for the cyclic-skew fits and $\approx 9 \times 10^{-2}$ for the unconstrained least-squares fit. For the kernel-blended operator at interpolated $\mu^\ast$ it is $\approx 8 \times 10^{-17}$ (\figref{fig:energy_convexity}(b)). The structured comparison of \cref{sec:burgers} confirms the mechanism. The cyclic-skew model is long-time stable on all $17$ plume parameters out to five training horizons even though its inferred linear operator is itself mildly unstable at every super-regime node (\cref{sec:plume:longtime}). The Lyapunov-stable-%
$\hat{\bm{A}}$ model has a stable linear part but an unconstrained quadratic operator, and on the plume it diverges at every node under the regime-split basis and completes at most one of $17$ under the single basis. Long-time boundedness on these problems is governed by the cubic channel, with the linear spectrum the weaker predictor of the two.

\begin{figure}[htbp]
\centering
\includegraphics[width=\linewidth]{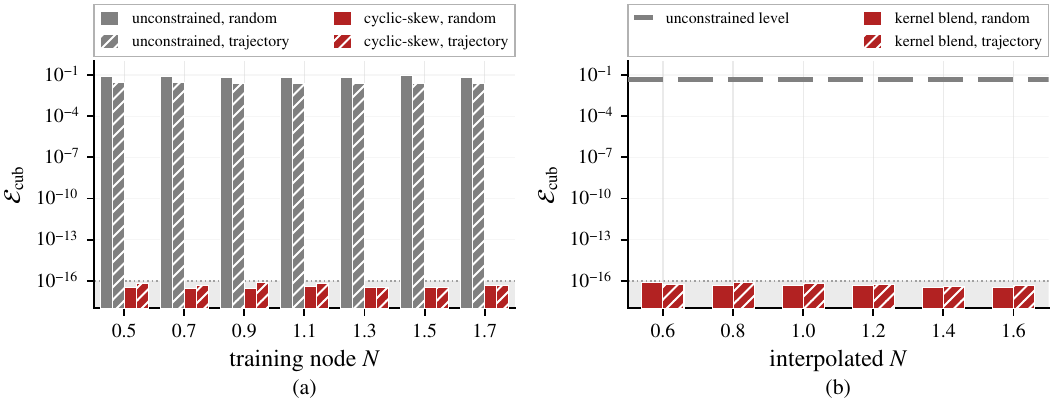}
\caption{Normalized cubic energy form $\mathcal{E}_{\mathrm{cub}}$ (maximized over random unit vectors and over the trajectory) for (a) the unconstrained and cyclic-skew fits at the seven super-side training nodes $N \in \{0.5, 0.7, \ldots, 1.7\}$ of the nine-value representative design of the $\rho_{\rm aff}$ study and (b) the kernel-blended operator at interpolated $N$, with the unconstrained level marked.}
\label{fig:energy_convexity}
\end{figure}

\subsection{Basis-sharing perturbation}
\label{sec:app_davis_kahan}

The local operator interpolation of \cref{sec:methods:loi} fits all per-parameter operators in a single shared reduced basis. The following standard bound establishes that this shared basis is a controlled perturbation of the basis a per-parameter fit would use.

\begin{proposition}[Davis--Kahan bound for the shared reduced basis]
\label{prop:davis_kahan}
Let $\bm{\Phi}^{(i)}$ be the rank-$r$ truncated SVD basis of the
joint-mean-centered block $\bm{Q}^{(i)} - \bar{\bm{q}}\bm{1}_K^\top$ at $\mu_i$, and let
$\bm{\Phi}_{\rm joint}$ be the rank-$r$ basis of the concatenated
snapshot matrix $\bm{Q}_{\rm joint}$. If the Gram-matrix eigenvalue gap
$\sigma_r^2(\bm{Q}_{\rm joint}) - \sigma_{r+1}^2(\bm{Q}_{\rm joint})
\geq \gamma > 0$, then the principal-angle distance between the two
subspaces satisfies
\begin{equation}
\|\sin \bm{\Theta}(\bm{\Phi}^{(i)}, \bm{\Phi}_{\rm joint})\|_2
\;\leq\;
\frac{\|\bm{Q}_{\rm joint}\bm{Q}_{\rm joint}^\top
  - \bm{Q}^{(i)}{\bm{Q}^{(i)}}^\top\|_2}{\gamma}.
\end{equation}
\end{proposition}

\noindent The bound is the symmetric Davis--Kahan theorem \citep{davis1970rotation, stewart1990matrix} applied to the symmetric Gram matrices $\bm{Q}\bm{Q}^\top$ of the two snapshot sets, whose eigenvalues are the squared singular values. The self-gap of the joint Gram matrix suffices here because the columns of $\bm{Q}^{(i)}$ are columns of $\bm{Q}_{\rm joint}$ (both matrices carry the same joint-mean centering), so $\bm{Q}_{\rm joint}\bm{Q}_{\rm joint}^\top - \bm{Q}^{(i)}\bm{Q}^{(i)\top}$ is the sum of the remaining blocks' Gram contributions and is positive semidefinite, and eigenvalue interlacing converts the joint-Gram gap into the mixed gap the theorem requires. The Gram eigenvalue gap $\gamma = \sigma_r^2 - \sigma_{r+1}^2$ is evaluated at the truncation rank $r$ and is positive at the ranks used here. As with any sin-$\bm{\Theta}$ estimate, the bound is informative only when the perturbation 2-norm in the numerator is below the gap $\gamma$, since beyond that the right-hand side exceeds one and the subspace distance is no longer constrained. The shared basis therefore incurs a perturbation of the per-parameter operators that is bounded in terms of the gap. That bounded perturbation is the price of the per-parameter operator quality that the local interpolation delivers independently of $M$.


\section{Stability-constrained parametric interpolation}
\label{sec:app_certified}

The free construction of \appref{sec:app_energy} closes the cubic energy injection channel and is preserved by the convex blend, but it leaves the linear part $\hat{\bm{A}}(\mu^\ast)$ uncontrolled. The blended operators used at evaluation are mildly unstable, with $\lambda_{\max}(\operatorname{sym}\hat{\bm{A}}(\mu^\ast))$ between $+0.20$ and $+0.46$ across the evaluated parameters. The energy bound of \propref{prop:energy_neutral} therefore guarantees only freedom from finite-time blowup, without a uniform-in-time bound. This appendix gives a variant that is uniformly bounded in time at every interpolated parameter, obtained by matching the interpolation rule to the geometry of each constraint.

Stability-preserving operator interpolation has a substantial literature, almost all of it linear. Matrix-manifold constructions interpolate stable linear time-invariant operators \citep{amsallem2009interpolation,eid2011stability}, and a separate line computes the nearest stable linear operator to a given unstable one \citep{gillis2017distance}, the projection underlying our dissipative clip. None of these addresses a quadratic nonlinearity. For nonlinear reduced models, one line of work constrains a single non-parametric quadratic operator inside an absorbing ball \citep{pontesduff2024stability}, and another infers a single locally stable quadratic model \citep{goyal2024guaranteed}. A third preserves a gradient structure in the inferred operators \citep{geng2024gradient}, with a forecast-horizon-dependent error constant. A roadmap states the preserve-structure-under-interpolation principle without a quadratic construction \citep{sreekumar2026roadmap}. The new element is the pairing of a \emph{parameter-uniform} absorbing ball with a uniform-in-time error bound for a \emph{data-driven quadratic} reduced model. This combination follows because both structural constraints are convex and the interpolation is a single convex blend.

\subsection{Geometry-matched interpolation}

The two structural properties we use live in convex sets of operator space. The cyclic-skew condition of Eq.~\eqref{eq:cyclic_skew} is a linear subspace, and dissipativity with margin $\eta>0$, $\operatorname{sym}(\hat{\bm{A}}) \preceq -\eta\bm{I}$, is a convex set, a translate of the dissipative cone. Both are closed under convex combinations, and a single convex kernel blend therefore preserves both.

\begin{proposition}[Geometry-matched preservation]
\label{prop:cert_preserve}
Let $\{(\hat{\bm{c}}_i, \tilde{\bm{A}}_i, \hat{\bm{H}}_i)\}$
satisfy, for every $i$, the cyclic-skew condition on
$\hat{\bm{H}}_i$ and
$\operatorname{sym}(\tilde{\bm{A}}_i)\preceq-\eta\bm{I}$. Then
for any convex weights $\{w_i\}$ the blended operators
$\hat{\bm{c}}(\mu^\ast)=\sum_i w_i\hat{\bm{c}}_i$,
$\hat{\bm{H}}(\mu^\ast)=\sum_i w_i\hat{\bm{H}}_i$ and
$\hat{\bm{A}}(\mu^\ast)=\sum_i w_i\tilde{\bm{A}}_i$ satisfy
$\hat{\bm{q}}^\top\hat{\bm{H}}(\mu^\ast)[\hat{\bm{q}}\otimes\hat{\bm{q}}]=0$ and
$\operatorname{sym}(\hat{\bm{A}}(\mu^\ast))\preceq-\eta\bm{I}$.
\end{proposition}
\begin{proof}
The first claim is \propref{prop:blend_preserves}. For the second,
$\operatorname{sym}(\hat{\bm{A}}(\mu^\ast))
=\sum_i w_i\operatorname{sym}(\tilde{\bm{A}}_i)
\preceq\sum_i w_i(-\eta\bm{I})=-\eta\bm{I}$,
since the matrix inequality is preserved under convex combination.
\end{proof}

A dissipative version $\tilde{\bm{A}}_i$ of each fitted operator is obtained by clipping the symmetric part. With $\hat{\bm{A}}_i=\bm{G}_i+\bm{K}_i$ for $\bm{G}_i$ symmetric and $\bm{K}_i$ skew, and the eigendecomposition $\bm{G}_i=\bm{U}\operatorname{diag}(s)\bm{U}^\top$, set $\tilde{\bm{A}}_i=\bm{U}\operatorname{diag}(\min(s,-\eta))\bm{U}^\top+\bm{K}_i$. This is the nearest operator in Frobenius norm with $\operatorname{sym}\preceq-\eta\bm{I}$ and unchanged skew part. No manifold interpolation is required, because dissipativity is convex and the plain convex blend already preserves it. The cyclic-skew $\hat{\bm{H}}_i$ are left unmodified. The clip therefore touches only the linear part, and the cubic energy injection channel stays closed.

\subsection{Uniform-in-time stability and a priori error}

\begin{proposition}[Uniform-in-time energy bound at every parameter]
\label{prop:cert_uniform}
For the stability-constrained reduced model
$\dot{\hat{\bm{q}}}=\hat{\bm{c}}(\mu^\ast)+\hat{\bm{A}}(\mu^\ast)\hat{\bm{q}}
+\hat{\bm{H}}(\mu^\ast)[\hat{\bm{q}}\otimes\hat{\bm{q}}]$ of
\propref{prop:cert_preserve}, the energy $E=\tfrac12\|\hat{\bm{q}}\|^2$ obeys
$\dot{E}\le\|\hat{\bm{c}}\|\sqrt{2E}-2\eta E$, and therefore
\begin{equation}
\|\hat{\bm{q}}(t)\|\;\le\;\max\!\Bigl(\|\hat{\bm{q}}(0)\|,\;R\Bigr),
\qquad R:=\frac{\|\hat{\bm{c}}(\mu^\ast)\|}{\eta},
\label{eq:cert_ball}
\end{equation}
for all $t\ge0$ and all $\mu^\ast$, a parameter-uniform absorbing ball.
\end{proposition}
\begin{proof}
The cubic term vanishes (\propref{prop:cert_preserve}). With
$\operatorname{sym}\hat{\bm{A}}\preceq-\eta\bm{I}$,
$\dot E=\hat{\bm{q}}^\top\hat{\bm{c}}+\hat{\bm{q}}^\top\hat{\bm{A}}\hat{\bm{q}}
\le\|\hat{\bm{c}}\|\sqrt{2E}-2\eta E$. Writing $u=\|\hat{\bm{q}}\|$ gives
$\dot u\le\|\hat{\bm{c}}\|-\eta u$, which integrates to
$u(t)\le R+(u(0)-R)e^{-\eta t}$.
\end{proof}

Under the convex blend $\|\hat{\bm c}(\mu^\ast)\| \le \max_i \|\hat{\bm c}_i\|$, so the radius $R$ is bounded uniformly in $\mu^\ast$. The estimate in \propref{prop:cert_uniform} is the classical absorbing-ball bound for a quadratic system with an energy-preserving nonlinearity and a dissipative linear part \citep{schlegel2015long}. New here is its parameter-uniform preservation under the convex operator blend. The uniform amplitude bound of \propref{prop:cert_uniform} converts the operator-interpolation error into a uniform bound on the trajectory error, a conversion that neither the free construction nor a single-operator parametric fit admits.

\begin{remark}[Conditional uniform error bound]
\label{rem:cert_error}
When the enforced margin $\eta$ dominates the nonlinear amplification at the
absorbing-ball radius $R$ of Eq.~\eqref{eq:cert_ball}, a scalar Gr\"onwall
estimate converts the operator-interpolation defect into a uniform-in-$t$
bound on the trajectory error, the uniform-in-time counterpart of the
horizon-dependent error constant of the single-global-operator construction
\citep{geng2024gradient}. On the advection-dominated plume, that premise fails
at the small attainable margin ($\eta = 0.05$), and the operative guarantee therefore
remains the unconditional absorbing ball of \propref{prop:cert_uniform}.
\end{remark}

\subsection{The accuracy-versus-stability tradeoff}

\Tabref{tab:certified} reports the leave-one-out comparison on the plume across the free, Lyapunov-stable, and stability-constrained constructions. The comparison is averaged over five basis seeds. The free construction carries no uniform-in-time guarantee because its blended linear operator is not dissipative. The stability-constrained construction guarantees $\operatorname{sym}\hat{\bm{A}}(\mu^\ast)\preceq-\eta\bm{I}$ at every parameter (\propref{prop:cert_preserve}) at a substantial accuracy cost, because the data-optimal operators of this advection-driven flow are far from dissipative and the projection that enforces the constraint overdamps them. The Lyapunov-stable parameterization lands at the dissipativity boundary, and that placement is guaranteed because Eq.~\eqref{eq:lyap_A} gives $\operatorname{sym}(\hat{\bm{A}}) \preceq -\varepsilon\bm{I}$ with the small fitting slack $\varepsilon$ and the convex blend preserves that constraint set. This slack $\varepsilon$ is distinct from the prescribed stability margin $\eta$. The guarantee it inherits is only the $\varepsilon$-level one, whose absorbing ball $\|\hat{\bm{c}}\|/\varepsilon$ is too large to be informative. The dissipative constraint set of \propref{prop:cert_preserve} at a finite margin $\eta$ therefore carries the quantitative bound. The cyclic-skew constraint already secures empirical long-time stability at the best accuracy, and the stability constraint provides a provable bound when one is required, at a quantified cost. On a problem in which dissipation dominates advection, that cost would be smaller.

\begin{table}[htbp]
\centering
\caption{Accuracy versus stability guarantee on the plume, leave-one-out over five basis seeds. The growth-rate column reports the Euclidean-symmetric value $\lambda_{\max}(\operatorname{sym}\hat{\bm{A}})$ as the worst-case statistic, the maximum over the $17$ leave-one-out parameters. The $\bm M$-weighted growth rate of \appref{sec:cert_cost_law} is a separate quantity.}
\label{tab:certified}
\small
\begin{tabular}{lccc}
\toprule
Construction & In-window & Stable $t{=}60$ & $\lambda_{\max}(\operatorname{sym}\hat{\bm{A}}(\mu^\ast))$ \\
\midrule
Free (cyclic-skew $\bm{H}$)        & $7.9\%$  & $17/17$ & $+0.46$ \\
Lyapunov-stable $\bm{A}$ + skew    & $16.9\%$ & $17/17$ & $-0.00$ \\
\textbf{Stability-constrained} ($\eta = 0.05$)      & $37.3\%$ & $17/17$ & $\mathbf{-0.05}$ \\
\bottomrule
\end{tabular}
\end{table}

\subsection{The data-weighted growth rate governs the stability cost}
\label{sec:cert_cost_law}

The $7.9\%\!\to\!37.3\%$ accuracy penalty of \tabref{tab:certified} is regime-dependent and far from a fixed price. The first candidate predictor is the worst-case growth rate $a_{\max} := \lambda_{\max}(\operatorname{sym}\hat{\bm A})$, the Euclidean value quoted in \tabref{tab:certified} and the synthesis of \cref{sec:discussion:synthesis}. This is the growth rate of the most unstable mode the constraint must damp, yet it proves the weaker predictor. Its $\bm M$-weighted counterpart $\lambda_{\max}(\operatorname{sym}_{\bm M}\hat{\bm A})$ uses $\bm M$, the symmetric positive-definite weight matrix of the energy metric. That counterpart is reported only as a separate metric-weighted diagnostic. The reliable one is the \emph{data-weighted} growth rate $a_{\rm data}$ of Eq.~\eqref{eq:a_data}. It is the energy-weighted growth rate of the inferred linear part over the directions the training snapshots occupy. Results are aggregated per training node, averaging three basis seeds at each $N$, which yields a smaller dataset distinct from the five-seed leave-one-out table of \tabref{tab:certified}. On it the stability cost correlates with $a_{\rm data}$ at $+0.90$, $+0.81$, and $+0.72$ under the $L_2$, kinetic-energy, and physical kinetic-plus-potential energy metrics. All three are positive and substantial. A Spearman rank correlation of $+0.84$ on the $L_2$ metric rules out a single high-leverage cluster, and the cluster-bootstrap-by-$N$ $95\%$ confidence interval is $[+0.70, +0.97]$. The worst-case growth rate $a_{\max}$ is the weaker predictor under every metric, at $+0.06$ and $+0.14$ under the kinetic-energy and physical metrics. On the $L_2$ metric it is even anti-correlated ($-0.48$), because it ignores whether the trajectory visits the unstable mode it measures. The data-weighted quantity also stays predictive after partialing out the parameter $N$ at partial correlation $+0.72$ and is therefore not a proxy for the regime label.

The asymmetry follows from a short derivation. With $\hat{\bm A}_{\rm cert}$ the constrained minimizer, orthogonality of the unconstrained least-squares residual gives the cost as $\|(\hat{\bm A}_{\rm cert}-\hat{\bm A})\,\hat{\bm Q}\|_F^2$, with $\hat{\bm Q}$ the matrix of reduced training snapshots (the projected centered snapshots written $\bm X^{(i)}$ in \cref{sec:background:fit}). Let $\bm L$ be the Cholesky factor of $\bm M$ satisfying $\bm M = \bm L\bm L^\top$, and let $\tilde{\bm Q} := \bm L^\top\hat{\bm Q}$ be the $\bm M$-whitened snapshot matrix. The constraint touches only the $\bm M$-symmetric part of the unconstrained fit $\hat{\bm A}$. Clipping its eigenvalues $\lambda_i$ that exceed $-\eta$ down to $-\eta$ gives, to leading order,
\begin{equation}
\text{cost}\;\approx\;\sum_{i:\,\lambda_i>-\eta}(\lambda_i+\eta)^2\,e_i,
\qquad e_i=\bm v_i^\top(\tilde{\bm Q}\tilde{\bm Q}^\top)\bm v_i,
\label{eq:cost_formula}
\end{equation}
a sum over the energy-amplifying modes weighted by the trajectory energy $e_i$ they carry, with $\bm v_i$ the eigenvectors of the $\bm M$-symmetric part of $\hat{\bm A}_{\rm cert}$. The cost of Eq.~\eqref{eq:cost_formula} is measured in the $\bm M$-metric. The data-weighted growth rate is the natural summary of this sum, since both weight by $e_i$. The worst-case growth rate $a_{\max}$ is the unweighted maximum and ignores whether the data visit the unstable mode. The constraint enforced in this paper is the Euclidean instance of this construction, where $\bm M = \bm I$, the cost reduces to the Frobenius norm, and the constraint becomes $\operatorname{sym}\hat{\bm A}\preceq-\eta\bm I$. The kinetic-energy and physical kinetic-plus-potential energy metrics enter only as diagnostics that report the same cost law under alternative weights. Eq.~\eqref{eq:cost_formula} predicts the measured cost within each system at correlation $+0.81$ on the plume and $+0.73$ on Burgers. We report it as a mechanistic decomposition and make no claim of a tight bound, because the soft penalty re-optimizes the constant and quadratic operators and therefore pays less than the rigid clip.

The cost predictor $a_{\rm data}$ is a property of the fitted operator, yet across the seventeen $N$ it covaries with the domain-integrated buoyancy production $\langle w\,b\rangle$ computed directly from the full-order velocity and buoyancy fields, at Pearson $+0.96$ and Spearman $+0.84$. Both are largest in the weakly stratified free rise regime and decay as the stratification strengthens. This covariation is physical, because $a_{\rm data}$ runs high precisely where buoyancy drives the energy, and the two quantities are different expressions of the same buoyancy-production physics. We claim no independent validation from it. Enforcing the constraint is therefore expensive exactly in the weakly stratified free rise regime, where buoyancy drives the energy the constraint must suppress, and inexpensive in the strongly stratified trapping regime, where that production is quenched. The same ranking holds on the viscous Burgers problem, where the correlation of the cost with $a_{\rm data}$ is $+0.78$. There the fitted linear operator has a large worst-case growth rate ($\lambda_{\max}\approx+4.5$) yet the constraint is inexpensive, because the trajectory energy lives in dissipative directions where $a_{\rm data}<0$. The Burgers stability costs are uniformly small, all below $1\%$, and that comparison confirms the predictor ranking without probing an expensive regime. The plume supplies the wide spread of costs over which the law is tested. Across both systems and all three metrics $a_{\rm data}$ remains the reliable predictor while the worst-case growth rate is weaker throughout. The correlation of the cost with $a_{\rm data}$ is robust to the stability margin. Across the margins of \figref{fig:tradeoff} it stays near $+0.90$ on the plume and positive on Burgers, and the predictor ranking holds at every margin.

A further question is whether the cost law is constructive. Because $a_{\rm data}$ marks where the constraint is expensive, a parameter-adaptive margin $\eta(\mu)$ could in principle relax the constraint where $a_{\rm data}<0$ and enforce it where $a_{\rm data}>0$, recovering accuracy while still controlling dissipativity. The convex preservation of \propref{prop:cert_preserve} holds for any node-wise margin schedule, and the construction is therefore admissible. We tested such a schedule keyed to $a_{\rm data}$ directly, and it yields only a small gain. The per-node margin recovers only about two percentage points over the uniform constraint. It also downgrades the stability-constrained operator from strictly dissipative to merely marginally dissipative, and no schedule approaches the free-construction accuracy.

The estimator drives the recovery. Refitting the constant and quadratic operators under each clipped linear operator before blending keeps the hard margin and the guarantee intact and recovers about $60\%$ of the stability cost, the excess over the free $7.9\%$ model. Averaged over five basis seeds, the stability-constrained error falls from $37.3\%$, with a seed spread below $0.1$ points, to $18.7 \pm 0.5\%$. The starting value matches the clip-only row of \tabref{tab:certified}, which is therefore the upper bound obtained with the simplest estimator. The remainder behaves as the cost law predicts. The refit cost stays concentrated at the weakly stratified free rise nodes. Its correlation with $a_{\rm data}$ stays close to that of the clip-only cost, at Pearson correlations of $+0.84$ for the clip-only cost and $+0.79$ for the refit cost across the five seeds, each with a seed spread below $0.02$, and a rank correlation of $+0.75 \pm 0.02$ for the refit cost. The law therefore holds under both cost estimators. The price falls further only as the problem moves toward dissipation dominance, and the data-weighted growth rate remains the diagnostic of the part of the cost no estimator removes.

\subsection{Per-variant accuracy breakdown}
\label{sec:app_multimetric}

The body reports the aggregate plume accuracy of the proposed model in \cref{sec:plume:accuracy} and summarizes it in \tabref{tab:plume_accuracy}, which lists the metrics for the cyclic-skew configuration alone. \Figref{fig:plume_multimetric} shows the same metrics, in the same order, across all four structural variants. The cyclic-skew row reproduces \tabref{tab:plume_accuracy}, and the other three rows show how the alternative constructions perform on identical quantities. The columns run from the Frobenius error over the three windows, through the in-window $z_{\max}(t)$ temporal correlation and the two rise-quantity RMSEs, to the count of trajectories that integrate to $t = 60$. As in \cref{sec:burgers:matrix}, only the cyclic-skew variant is both accurate in window and stable to $t = 60$, and the figure extends this comparison from the aggregate in-window error to every window and QoI. The plain least-squares fit is accurate in window at the median but completes only $2/17$ trajectories, and the Lyapunov-only fit matches the cyclic-skew in-window error yet completes $0/17$.

\begin{figure}[htbp]
\centering
\includegraphics[width=\linewidth]{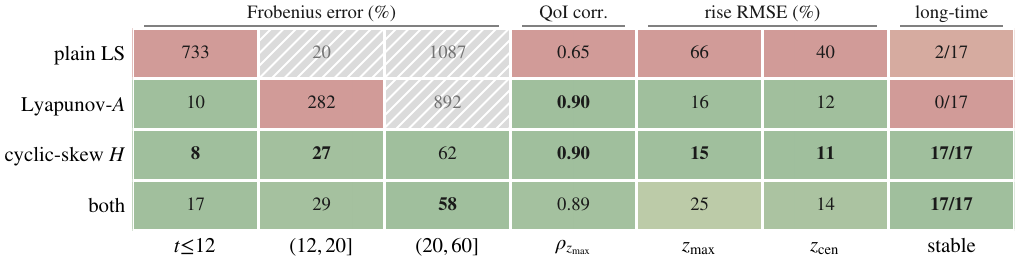}
\caption{The four structural variants over the $17$ training $N$ by leave-one-out, on the metrics of \tabref{tab:plume_accuracy}. Columns are the Frobenius error over the in-window ($t \leq 12$) and the two forecast windows, the in-window $z_{\max}(t)$ Pearson correlation $\rho_{z_{\max}}$, the relative RMSEs of $z_{\max}$ and the buoyancy-weighted centroid $z_{\rm centroid}$, and the number of the $17$ parameters integrated to $t = 60$. Cells are shaded by within-column standing (greener better, best in bold). Averages over fewer than $14$ of $17$ parameters are survivor-subset values, hatched in gray and excluded from the comparison.}
\label{fig:plume_multimetric}
\end{figure}

\phantomsection
\addcontentsline{toc}{section}{Data availability}
\section*{Data availability}

The code and data that support the findings of this study are available from the corresponding author on reasonable request.


\phantomsection
\addcontentsline{toc}{section}{References}
\bibliographystyle{elsarticle-num}
\bibliography{references/references}

\end{document}